\documentclass[11pt]{article}

\usepackage[T1]{fontenc}
\usepackage{amsmath,amssymb,amsthm}
\usepackage[margin=1.1in]{geometry}
\usepackage{booktabs}
\usepackage{graphicx}
\usepackage{enumitem}
\usepackage{tikz}
\usetikzlibrary{plotmarks}
\usepackage{pgfplots}
\pgfplotsset{compat=1.17}
\newcounter{algorithmctr}
\newcounter{algoline}
\newcommand{\aline}[2][0]{\stepcounter{algoline}%
  \noindent\makebox[2.2em][r]{{\footnotesize\thealgoline:}}\hspace{0.8em}%
  \hspace*{#1\parindent}#2\par}
\newcommand{\alinecont}[2][0]{%
  \noindent\makebox[2.2em][r]{}\hspace{0.8em}%
  \hspace*{#1\parindent}\hspace{1em}#2\par}
\newcommand{\acomment}[1]{\hfill $\triangleright$ {\small #1}}
\newenvironment{algo}[2]{%
  \begin{figure}[tp]%
  \refstepcounter{algorithmctr}\label{#2}%
  \setcounter{algoline}{0}%
  \noindent\textbf{Algorithm \thealgorithmctr}\,: #1\par\nobreak\vspace{3pt}%
  \hrule height 0.8pt\vspace{5pt}%
  \setlength{\parindent}{1.4em}}%
  {\vspace{4pt}\hrule height 0.8pt\end{figure}}

\usepackage[colorlinks=true,linkcolor=black,citecolor=black,urlcolor=black]{hyperref}

\usepackage{xcolor}
\newcommand{\new}[1]{\textcolor{black}{#1}}

\makeatletter
\let\mm@paragraph\paragraph
\renewcommand{\paragraph}[1]{\mm@paragraph{#1.}}
\makeatother

\newtheorem{lemma}{Lemma}
\newtheorem{theorem}{Theorem}
\newtheorem{corollary}{Corollary}

\newtheorem{proposition}{Proposition}

\newcommand*{\draft}{}
\ifdefined\draft
    \newcommand{\andy}[1]{\textcolor{red}{[Andy: #1]}}
\else
    \newcommand{\andy}[1]{}
\fi

\newcommand{\genesis}[1]{b_{\mathrm{gen}}^{#1}}
\newcommand{\lgenesis}{\ell_{\mathrm{gen}}}
\newcommand{\hd}{\mathsf{hd}}
\newcommand{\Tr}{\mathrm{Tr}}
\newcommand{\lead}{\mathsf{lead}}
\newcommand{\Tips}{\mathsf{Tips}}
\newcommand{\TipsF}{\mathsf{Tips}^{*}}
\newcommand{\Prop}{\mathsf{Prop}}
\newcommand{\Horiz}{\mathsf{Horiz}}
\newcommand{\Ord}{\mathsf{Ord}}
\newcommand{\Emit}{\mathsf{Emit}}
\newcommand{\ep}{\mathsf{e}}
\newcommand{\SelectAnchor}{\mathsf{SelectAnchor}}
\newcommand{\ProposeChains}{\mathsf{ProposeChains}}
\newcommand{\ProduceNext}{\mathsf{ProduceNext}}
\newcommand{\Positions}{\mathsf{Positions}}
\newcommand{\Extensions}{\mathsf{Extensions}}
\newcommand{\da}{\mathrm{da}}
\newcommand{\vote}{\mathrm{vote}}
\newcommand{\novote}{\mathrm{novote}}
\newcommand{\nullify}{\mathrm{nullify}}
\newcommand{\nullification}{\mathrm{nullification}}
\newcommand{\State}{\mathcal{S}}

\title{Multimmit: Low-Latency and Robust Transaction Ordering}
\author{\textsc{Andrew Lewis-Pye}$^{1,2}$, \textsc{Patrick O'Grady}$^{2}$, and \textsc{Ben Clabby}$^{2}$\\[4pt]
\normalsize $^{1}$London School of Economics, UK \ and \ $^{2}$Commonware, USA}
\date{}

\begin{document}
\maketitle

\begin{abstract}
To meet the throughput demands of modern blockchain systems,
protocols for State Machine Replication increasingly have many
processors disseminate transaction blocks in parallel. For these protocols, the
\emph{certified} approach requires availability certificates before
blocks enter the ordering process. This is robust, but adds latency.
By contrast, BlueBottle's core protocol, the leading \emph{uncertified} design under
$n\geq5f+1$ (where $n$ is the number of processors, of which at most
$f$ are Byzantine), reduces latency in favourable conditions but can
put fetching missing data on the critical path. Raptr, the state of
the art under $n\geq3f+1$, takes a middle course. It avoids fetching
on the voting path, but transactions still wait for a leader's
proposal and two voting rounds. The leader controls the proposed
sequence, and one withheld batch can delay everything later in it.
We prove that strengthening the resilience assumption alone cannot
remove Raptr's second voting round.

We present Multimmit, a protocol for $n\geq5f+1$ that builds on
Raptr's approach and finalises leader blocks after a single round of
voting. Latency improves on Raptr by at least one message delay,
and more in the common case. The protocol also guarantees much more
than is bought by the change in resilience. \emph{Extension votes}
let voters endorse fresh transaction blocks omitted from the leader's proposal
and give a strong censorship-resistance guarantee: Once a leader
block finalises, every fresh, well-circulated block of an honest
chain enters the total ordering by the next finalised leader block,
even if that block's leader is Byzantine. Moreover, leader discretion over ordering is constrained: the
protocol interleaves producer chains by a prescribed rule, rather
than accepting a leader-chosen transaction sequence. For payments, the protocol provides finality within a single
round trip. A payment is final once the payment's producer assembles
a certificate of $n-f$ acknowledgments, roughly $2\delta$ after
dissemination, for $\delta$ the bound on message delay: the payment
will then be ordered and executed, whatever any leader or any other
processor does, and across committee changes. The certificate is a
transferable receipt, verifiable from committee keys alone.

In experiments on a global deployment of fifty validators across
thirteen regions, Multimmit's median latency to placement in the
total ordering moves from 383 to 411\,ms as the offered rate rises
from 100,000 to 1,000,000 transactions per second. At 100,000
transactions per second, it is essentially unchanged under
$0.1\%$ message loss, which slows BlueBottle by more than an order
of magnitude, and rises by a factor of 1.4 with nine validators
crashed, which slows Raptr by more than an order of magnitude.
\end{abstract}
\section{Introduction}\label{sec:intro}

State Machine Replication (SMR) is a fundamental primitive for
distributed computing that allows a collection of processors to
maintain a consistent, shared log of transactions despite the
failure, or even malicious behaviour, of some participants.
Originally developed for fault-tolerant systems, SMR has become the
algorithmic backbone of modern blockchains and decentralised
applications. The core challenge is to ensure that all correct
processors agree on the same sequence of transactions (Consistency),
while also guaranteeing that new transactions are eventually included
(Liveness), even in the presence of Byzantine faults, where corrupted
processors may behave arbitrarily.

Protocols for SMR typically operate in the partially synchronous
model~\cite{DLS88}, in which messages may be delayed arbitrarily before an
unknown time GST (the Global Stabilisation Time), but are delivered
within a known bound $\Delta$ thereafter. This model captures
realistic network conditions, including temporary partitions, and
underpins many widely deployed protocols. Latencies are
stated throughout in terms of $\delta$, the \emph{actual} least
upper bound on message delay after GST, which is unknown to the
protocol and may be far smaller than $\Delta$.\footnote{So every message sent at time $t$ is
received by $\max\{t,\mathrm{GST}\}+\delta$, but $\delta$ is unknown.}

\paragraph{The throughput challenge} In classical leader-based
protocols such as PBFT~\cite{castro1999practical} and its many derivatives, the leader
of each view assembles transactions into a block and sends it to all
other processors, who vote on whether to accept it. A disadvantage of this approach is that the leader's
outgoing bandwidth is then a bottleneck for the whole system.
Erasure coding alleviates the bottleneck, the leader dispersing coded
chunks that the processors then relay, and is used to this end by
Kudzu~\cite{shoup2025kudzu} and Alpenglow~\cite{alpen}. It comes at
the price of a data expansion rate, typically around three, so a
leader with $x$ bits of transaction data still sends around $3x$
bits. Another response, now standard, is to
have \emph{all} processors disseminate transaction data in parallel,
with the consensus mechanism operating on small references to that
data rather than on the data itself. DAG-based
protocols~\cite{danezis2022narwhal,spiegelman2022bullshark,babel2023mysticeti}
realise this idea with every processor proposing concurrently,
though modern designs still organise the ordering of the resulting
graph around designated anchor blocks. They also couple
dissemination to consensus structurally, since a block enters the
ordering only once later blocks reference it, so that placement
depends on the continued growth of the graph. Multi-chain designs
such as Autobahn~\cite{giridharan2024autobahn} retain the
leader-based structure, decouple the two concerns, and are the
starting point for this paper. There, each
processor builds its own chain of transaction blocks, disseminating
each block as it is produced, and a leader-based protocol run on
small \emph{leader blocks} establishes a total ordering on the
transaction blocks of all chains. Transaction data then travels over
the network only once, from its producer to everyone else. Provided
production is spread across enough producers, the system can then
in principle finalise transactions at roughly the rate at which
individual processors can send and receive data
(Section~\ref{sec:final} returns to the accounting).

\paragraph{Certified and uncertified tips} Designs of this kind face
a basic choice as to when a transaction block may be referenced by
the ordering process. The \emph{certified} (or pessimistic) approach,
taken by Autobahn, allows a leader block to reference only
transaction blocks whose availability has been attested by a quorum.
This is robust, but it prices the certification round trip into the
latency of every transaction. Three message delays separate a
transaction block from its earliest possible inclusion in a leader
block, since the block must reach the attesters, their attestations
must return to be assembled into a certificate, and the certificate
must reach the leader. The \emph{uncertified} (or optimistic) approach allows
proposals to reference transaction blocks immediately, with voters
missing the corresponding data expected to fetch it before voting.
This achieves low latency in favourable conditions, but degrades
badly when data is missing, since fetching then sits on the critical
path to consensus. The effect is not hypothetical, and has been
measured in the DAG setting: a total message loss rate of $0.05\%$
across the network (one per cent of the egress traffic at five per
cent of processors) has been observed to increase the latency of
Mysticeti~\cite{babel2023mysticeti}, which references uncertified
data, by an order of
magnitude~\cite{arun2024shoalpp,tonkikh2025raptr}. The most
recent of the designs that reference uncertified data,
BlueBottle~\cite{vos2025bluebottle}, descends from Mysticeti and
makes the resilience assumption of the present paper, $n\ge 5f+1$,
for $n$ the number of processors and $f$ the number that may be
faulty.
Section~\ref{sec:experiments} verifies the same effect there: a
total message loss rate of $0.1\%$ raises BlueBottle's median
latency by orders of magnitude.

Raptr~\cite{tonkikh2025raptr}\footnote{BlueBottle was developed in part by researchers at Mysten Labs, the company behind Sui, which currently uses Mysticeti (the predecessor to BlueBottle); see \url{https://docs.sui.io/develop/sui-architecture/consensus}. Raptr was developed at Aptos Labs, which has deployed Baby Raptr as the first stage of a phased rollout towards the full Raptr protocol; see \url{https://aptosnetwork.com/currents/baby-raptr-lands-on-mainnet}.}, the current state of the art, takes a middle course.
Processors still disseminate batches of
transactions in parallel, but now leaders propose sequences of batch hashes without requiring that voters
hold the corresponding data. Each voter supports the longest initial
segment of the proposed sequence whose data it does hold, and the
protocol finalises a prefix supported by a quorum. No processor ever
blocks or fetches on the critical path, and whatever is finalised is certainly available.

\paragraph{Raptr's remaining costs} Raptr's remaining costs concern latency and the leader. The
latency of a transaction under Raptr includes several unavoidable
legs. The transaction must reach the leader, and the transaction
then waits for the leader's next proposal (a wait of
$\delta$ on average and longer when some processors are faulty). The
proposal must then be sent out, and two further rounds of voting follow.  A natural response would be for Raptr itself to assume
$n\ge 5f+1$ (as Solana did with Alpenglow~\cite{alpen}), and to remove one of the two rounds of voting. We show
in Section~\ref{sec:verification} that the stronger assumption does
not, by itself, allow this: finalising with one round of voting
requires a different design. 

The leader also retains significant control: a leader decides which
batches enter its proposed sequence and in what order, and a leader
can censor a transaction simply by not proposing it, so that a
sequence of faulty leaders censors a transaction for as long as the
sequence lasts.

\subsection{Our contributions} \label{OC} 

We present Multimmit, a protocol that
combines multi-chain data dissemination with the consensus skeleton
of Minimmit~\cite{chou2025minimmit}, in which a single round of voting suffices for
finality under the assumption that $n\ge 5f+1$.\footnote{$5f-1\le n$ is necessary and sufficient for
finality in one round of voting (see \cite{kuznetsov2021revisiting}
and
\url{https://decentralizedthoughts.github.io/2021-03-03-2-round-bft-smr-with-n-equals-4-f-equals-1/}).
As is standard~\cite{shoup2025kudzu,alpen,shrestha2025hydrangea}, we
assume $5f+1\le n$ for simplicity. The assumption is materially
stronger than the $n\ge 3f+1$ under which Autobahn and Raptr
operate, and is the price of finality in a single round of voting;
Section~\ref{sec:intuition-raptr} keeps the accounting explicit.
The two bounds answer different questions: $5f-1$ is the least
$n$ for one-round agreement on a single proposed value, while
$5f+1$ is the least at which one-round commitment of
well-circulated \emph{prefix} data is possible
(Corollary~\ref{cor:slack}), the two further processors purchasing
the per-position recovery threshold within block-inclusion
support.} We work in the standard
partially synchronous model: every message sent at time $t$ is
received by $\max\{t,\mathrm{GST}\}+\Delta$, for a known bound $\Delta$ and unknown
global stabilisation time GST, so that delays are arbitrary before GST
and bounded by $\Delta$ thereafter. Throughout, each processor maintains a
transaction chain of its own, appending the blocks it
disseminates, and consensus establishes one total ordering over
the blocks of all chains.

The design
is guided by a simple
principle, that \emph{latency must be optimised for the good case, but must also be robust to failures such as message loss, lagging transaction chains, crashed processors and faulty leaders}.  Three mechanisms are used to achieve
this:
\begin{itemize}[itemsep=2pt]
  \item \emph{Proposal-relative voting.} The leader proposes tips
    for every chain, certified or not, and each voter responds with a
    single vote, reporting \emph{per chain} how far up the chain it can
    endorse the proposal. No processor blocks or fetches on the critical path, and
    leader blocks are finalised after a single round of voting.
    Leaders propose per-chain tips rather than a transaction sequence;
    the protocol then fixes the interleaving of the voted chain prefixes. Proposal-relative voting also enables \emph{extension votes}, described next.
  \item \emph{Extension voting.} Votes may themselves endorse fresh
    blocks beyond the proposed tips. Extension votes have three functions: 
    \begin{itemize} 
    \item Leaders cannot censor: whatever the leader
    proposes, extension votes ensure that a fresh, well-circulated
    block of an honest chain anyway enters the ordering. 
    \item Lagging chains, or transaction blocks sent only to the leader, can delay the placement of transaction blocks from other chains by one view under correct leaders, and in any case only until the next finalised leader block (even if subsequent leader proposals do not reference those other chains).  
    \item  Extension
    votes also open a fast path, by which a transaction block that
    reaches the voters before they vote is finalised in the view of
    its own dissemination, saving a further $\delta$.\footnote{Reaching all
    correct voters within one message delay is however a demanding
    condition in geographically spread deployments, and none of the
    guarantees below depend on the fast path.}
    \end{itemize} 
  \item \emph{Payment receipts.} Our protocol also allows for
  \emph{payment receipts}, applicable for transactions such as
  payments that do not require total ordering for
  execution --- a guarantee in the tradition of consensusless
  payment designs~\cite{baudet2020fastpay,guerraoui2019consensus,sliwinski2019abc,blackshear2024sui},
  here integrated with a total ordering and surviving committee
  changes. The availability certificate
    that a payment's producer assembles in the course of
    dissemination becomes a binding, transferable receipt, and the
    payment is final roughly $2\delta$ after dissemination
    (Section~\ref{sec:payments}, and item 3 below), independent of the action of any leader.
\end{itemize}
The design has three principal benefits, which
Section~\ref{sec:intuition} quantifies and compares with the state
of the art for items 1 and 2, and Section~\ref{sec:payments}
quantifies for item 3.
\begin{enumerate}[itemsep=4pt]
\item \emph{Low latency.} Relative to Raptr, the protocol
reduces the latency of placement in the total ordering by \emph{at least}
$\delta$ in the good case, and by several message delays in the common
case.\footnote{Section~\ref{sec:intuition} defines the good and
common cases and gives the detailed accounting.} 

\item \emph{Robustness against the network, against other
producers, and against leaders}. Against the network: voters
report what they hold rather than fetching what they lack, so
message loss does not inflate the critical path.
 Against other
producers: a faulty transaction block producer can significantly delay only the blocks of its own
chain, costing other chains at most a one-view wait for placement in
the total ordering, irrespective of what future leaders propose. 
 Against leaders:
the protocol provides a guarantee of block inclusion that does not
depend on the leader. If a fresh transaction block produced by a
correct processor reaches all correct processors before they vote in
a view, then either the view finalises nothing at all, or the
block's membership of the eventual ledger is settled in that view
and it enters the total ordering by the next finalised leader block,
at the latest. A leader wishing to censor an honest producer's block
must therefore suppress its entire view, forgoing whatever rewards
finalisation brings. Leader discretion is further constrained by the prescribed rule
for interleaving producer chains.

\item \emph{Payment finality at $2\delta$.} A
payment is final once the producer assembles acknowledgments for
the block from $n-f$ processors, the chain's ordinary availability
votes at the receipt quorum: no conflicting payment is ever executed,
the block enters the total ordering within its
epoch,\footnote{For implementations with payment chains, time is divided into \emph{epochs}: the
committee of processors is fixed within an epoch and may change
between epochs. Epochs are introduced with the payment system in
Section~\ref{sec:payments}. The guarantee here is that a receipted
payment is placed at latest by the close of the epoch of its
block, even if the certificate is never broadcast.} whatever any leader or any other producer does, and the
payment executes, acknowledgments being issued only for blocks that
are entirely valid. The certificate is assembled in one round trip,
roughly $2\delta$ after dissemination, and is verifiable from
committee keys alone.
\end{enumerate}

\paragraph{Specification and analysis} We give a complete formal
specification of the protocol, with proofs of consistency and
liveness, and an accounting of data availability for every block
entering the ordering. We also determine the protocol's extraction
thresholds exactly, showing that its rules are optimal at $n=5f+1$
and identifying the precise relaxation available for larger $n$.
 Section~\ref{sec:sizes} calculates the concrete
sizes of every object the protocol transmits.

\paragraph{Experimental validation} Section~\ref{sec:experiments} compares Multimmit with Raptr and BlueBottle on global and North American testbeds with fifty validators. Multimmit\textquotesingle s global median placement latency rises from 383 to 411\,ms between 100k and 1M tx/s. Receipt latency is much lower than payment placement, and we separately test the effects of crashes, loss and censoring leaders.

\paragraph{Structure of the paper} Section~\ref{sec:setup} defines
the model. Section~\ref{sec:intuition} develops the protocol
informally, compares it quantitatively with Raptr, and collects the
guarantees case by case, including for \emph{greedy} Byzantine
leaders, meaning those that seek finalisation and avoid provable
misbehaviour. Section~\ref{sec:spec} gives the formal
specification, Section~\ref{sec:verification} the analysis, and
Section~\ref{sec:optimisations} optimisations, including the
leader-skipping rule that the implementation uses.
Section~\ref{sec:payments} presents payment chains.
Sections~\ref{sec:experiments}--\ref{sec:final} describe experiments,
related work and final comments.
\section{The setup}\label{sec:setup}

\noindent \textbf{The processors}. We consider a set $\Pi=\{p_1,\dots,p_n\}$ of $n$ processors. For $f$ such
that $5f+1\le n$, at most $f$ processors may become corrupted by the
adversary during the course of the execution, and may then display
Byzantine (arbitrary) behaviour. Processors that never become corrupted
are referred to as \emph{correct}.

\vspace{0.2cm}
\noindent \textbf{Cryptographic assumptions}. Our cryptographic
assumptions are standard for papers on this topic. Processors
communicate by point-to-point authenticated channels. We use a
cryptographic signature scheme, a public key infrastructure (PKI) to
validate signatures, and a collision resistant hash function $H$.
Beyond ordinary signatures, we make use of an aggregate signature
scheme (e.g.\ BLS): signatures on (possibly distinct) messages under
(possibly distinct) keys may be aggregated into a single group
element, verifiable given the constituent (message, key) pairs, at a
cost of one pairing computation per distinct message. We also use two
threshold schemes, with thresholds $n-2f$ and $2f+1$: each processor
thus holds an ordinary signing key and a key-share for each threshold
scheme. We assume a computationally bounded adversary. Following a
common standard in distributed computing and for simplicity of
presentation (to avoid the analysis of negligible error
probabilities), we assume these cryptographic schemes are perfect,
i.e., we restrict attention to executions in which the adversary is
unable to break them.

\vspace{0.2cm}
\noindent \textbf{The partial synchrony model}. We consider the
standard partial synchrony model, whereby the execution is divided
into discrete timeslots $t\in\mathbb{N}_{\ge 0}$ and a message
sent to another processor at time $t$ must arrive at some time $t'>t$ with
$t'\le\max\{\mathrm{GST},t\}+\delta$: here $\delta$ is the
(unknown) least upper bound on message delay after GST, messages
sent before GST included, and $\delta\le\Delta$ for a known bound
$\Delta$, which is what the protocol's timers use (noting that
$\delta$ may be significantly less than $\Delta$). While $\Delta$
is known, the value of GST is unknown to the protocol. The
adversary chooses GST and also message delivery times, subject to
the constraints just given. Correct processors begin the protocol
execution before GST
and are not assumed to have synchronised clocks. For simplicity, we do
assume that the clocks of correct processors all proceed in real time,
meaning that if $t'>t$ then the local clock of correct $p$ at time
$t'$ is $t'-t$ in advance of its value at time $t$. Using standard
arguments, our protocol and analysis can be extended in a
straightforward way to the case in which there is a known upper bound
on the difference between the clock speeds of correct processors.

\vspace{0.2cm}
\noindent \textbf{Transactions}. Transactions are messages of a
distinguished form, signed by the environment. Each timeslot, each
processor may receive some finite set of transactions directly from
the environment. We make the standard assumption that transactions are
unique (repeat transactions can be produced using an increasing
`ticker' or timestamps).

\vspace{0.2cm}
\noindent \textbf{State machine replication}. If $\sigma$ and $\tau$
are sequences, we write $\sigma\preceq\tau$ to denote that $\sigma$ is
a prefix of $\tau$, and say $\sigma,\tau$ are \emph{compatible} if
$\sigma\preceq\tau$ or $\tau\preceq\sigma$. If two sequences are not
compatible, they are \emph{incompatible}. If $\sigma$ is a sequence of
transactions, we write $\mathrm{tr}\in\sigma$ to denote that the
transaction $\mathrm{tr}$ belongs to the sequence $\sigma$. Each
processor $p_i$ is required to maintain an append-only log, denoted
$\mathrm{log}_i$, which at any timeslot is a sequence of distinct
transactions. We also write $\mathrm{log}_i(t)$ to denote the value of
$\mathrm{log}_i$ at the end of timeslot $t$. The log being append-only
means that, for $t'>t$, $\mathrm{log}_i(t)\preceq\mathrm{log}_i(t')$.
We require the following conditions to hold in every execution:

\vspace{0.1cm}
\noindent \textbf{Consistency}. If $p_i$ and $p_j$ are correct then,
for any timeslots $t$ and $t'$, $\mathrm{log}_i(t)$ and
$\mathrm{log}_j(t')$ are compatible.

\vspace{0.1cm}
\noindent \textbf{Liveness}. If $p_i$ and $p_j$ are correct and if
$p_i$ receives the transaction $\mathrm{tr}$ then, for some $t$,
$\mathrm{tr}\in\mathrm{log}_j(t)$.

\vspace{0.1cm}
\noindent As is common for protocols that separate data dissemination
from ordering, what the protocol explicitly guarantees for the
transaction blocks entering the log is \emph{data availability}
(sufficiently many correct processors hold the data of each such block
for it to be retrievable); how processors retrieve the data of blocks
they are missing is left as an implementation detail. Formally, we
thus solve Extractable SMR in the sense of~\cite{lewispye2026carnot}; since
the distinction is routine, we do not belabour it.

\vspace{0.2cm}
\noindent \textbf{Blocks, parents, and ancestors}. The protocol
produces \emph{blocks} of two kinds, defined precisely in
Section~\ref{sec:spec}: transaction blocks, which carry sequences of
transactions, and leader blocks, which carry none. In each case there
are distinguished \emph{genesis} blocks, and each block $b$ other than
a genesis block specifies a unique \emph{parent} (by hash value), of
which it is a \emph{child}. The \emph{ancestors} of $b$ are $b$ and
all ancestors of its parent, while a genesis block has only itself as
ancestor. We say $b$ \emph{extends} (or is a \emph{descendant} of)
$b'$ if $b'$ is an ancestor of $b$, and that two blocks are
\emph{incompatible} if neither is an ancestor of the other. The
\emph{height} of a block is its number of proper ancestors (so that
genesis blocks have height 0). Chains grow upwards with increasing
height: proper ancestors lie below a block and proper descendants
above it. We refer to the
greatest block for which a party has received all ancestors on a given chain of transaction
blocks as its \emph{tip} of that chain. When the protocol finalises
transaction blocks in a given order, processors append the
corresponding transactions to their logs, removing any duplicates; the
precise extraction of this ordering is specified in
Section~\ref{sec:spec}.
\section{The intuition}\label{sec:intuition}

Sections~\ref{sec:intuition}--\ref{sec:verification} describe
and analyse the protocol without the payment system, and for a
single \emph{epoch}: the committee of $n$ processors is fixed
throughout. Section~\ref{sec:payments} then adds payments and
epochs, under which the committee may change between epochs.

We want a protocol reflecting a commonly used blockchain architecture,
in which a designated set of processors is responsible for transaction-block production. In
general, these transaction-block producers might be distinct from the `validators'
that carry out consensus. For the sake of simplicity, we assume here
that the $n$ validators are the block producers, but one might more
generally consider $K$ block producers, for $K$ different from $n$. In
the style of Autobahn, we suppose that each block producer builds its
own chain, and consensus must then totally order the blocks from all
chains.

Multimmit therefore produces blocks of two kinds. Each processor $p_i$
builds its own chain of \emph{transaction blocks}. Transaction blocks
carry transactions but are never themselves the subject of consensus
votes. A leader-based protocol, essentially Minimmit, is run on
\emph{leader blocks}; leader blocks carry no transactions, but
determine a total ordering on transaction blocks. In this section we
describe the design informally. We first recall Minimmit, then
describe the chain layer, and then build up the consensus layer in two
stages, first with voting relative to the leader's proposals alone,
and then with \emph{extension votes}, which sharpen both latency and
censorship resistance.

\paragraph{The good and common cases} Throughout the section it will be useful to distinguish two
regimes. In the \emph{common case}, the network is synchronous and
leaders are correct, while up to $f$ other processors may be faulty.
In the \emph{good case}, the network is synchronous and \emph{all}
processors are correct. The common case is the most appropriate benchmark
for a protocol's performance, since some faulty processors are to be
expected in normal operation, and a protocol that performs well only
in the good case is fragile.

\subsection{Recalling Minimmit}

First, we recall the Minimmit protocol~\cite{chou2025minimmit}.
In this recap, a block is a Minimmit consensus block. Multimmit
applies this consensus structure to its leader blocks. 

\paragraph{One round of voting} Minimmit is a view-based protocol.
Each view $v$ has a designated leader, who proposes a block $b$. The
proposal specifies as parent a block from an earlier view. Other
processors then send signed votes for $b$ to all. Upon receiving $n-f$
votes for $b$ (an \emph{L-notarisation}, `L' for `large'), a processor
finalises $b$. This is `2-round finality', one round to send the block
and one round of voting, for which $5f+1\le n$
suffices (see Section~\ref{sec:intro}).

\paragraph{View progression} A view whose leader is faulty may never
produce an L-notarisation, so processors need a sound licence to move
to the next view. To this end, a processor that sees no progress in view $v$ will
eventually time out and send a \emph{nullify($v$)} message, indicating
that it wishes to abandon the view. While no protocol can ensure that
every view produces an L-notarisation, Minimmit ensures that every
view produces at least one of the following, either of which suffices:
\begin{itemize}[itemsep=0pt]
  \item An \emph{M-notarisation} (`M' for `mini') for a view $v$ block
    $b$: a set of $2f+1$ votes for $b$. This proves that no
    \emph{other} view $v$ block can receive an L-notarisation, since
    the two vote sets would share $(2f+1)+(n-f)-n\ge f+1$ processors,
    hence a correct one, and correct processors vote at most once per
    view. It is
    therefore safe to enter view $v+1$ and to build on $b$.
  \item A \emph{nullification} for view $v$: a set of $2f+1$
    nullify($v$) messages. The rules governing when processors may
    send nullify($v$) messages are arranged to ensure that a
    nullification proves that \emph{no} view $v$ block receives an
    L-notarisation. It is therefore safe to enter view $v+1$ and to
    skip view $v$ entirely.
\end{itemize}
Arranging that every view really does produce one of the two is the
one slightly delicate part of the protocol, and we ask the reader to
take this guarantee on trust for now (the rules appear in
Section~\ref{sec:spec}). Granting this guarantee, the rest of the
protocol and its correctness can be seen directly, as follows.

\paragraph{Anchoring proposals} The leader of view $v$ selects the
\emph{greatest} $v'<v$ for which it has received an M-notarisation,
and proposes a child of the corresponding block. By the rules for view
progression, the leader must then have received nullifications for
every view in the open interval $(v',v)$, since it progressed through
each of those views while holding an M-notarisation for none of them. To vote
for the proposal, other processors require an M-notarisation for the
parent and nullifications for every view in $(v',v)$, which together
prove that no block that could have been finalised is being skipped.
If the leader is correct then, during synchrony, this requirement is
automatically satisfied, since the leader has received these objects, and
every processor forwards each M-notarisation and nullification to all
others upon first receipt.

\paragraph{Consistency} These rules already give consistency. Suppose
a block $b$ for view $v$ receives an L-notarisation and, towards a
contradiction, consider the least view $v'>v$ in which some block
$b'$ incompatible with $b$ receives an M-notarisation. (If
consistency fails, such a view exists, since an L-notarisation
contains an M-notarisation.) Let $b''$, from view $v''<v'$, be the
parent of $b'$. Since at least $f+1$ correct processors voted for
$b'$, the block $b''$ has an M-notarisation, and nullifications exist
for every view in $(v'',v')$. Now:
\begin{itemize}[itemsep=0pt]
  \item If $v''=v$ then $b''=b$, since $b$ is the only view $v$ block
    with an M-notarisation. Then $b'$ extends $b$: a
    contradiction.
  \item If $v''\in(v,v')$ then $b''$ is compatible with $b$, by the
    choice of $v'$. Since views increase along chains, $b''$ extends
    $b$, so again $b'$ extends $b$: a contradiction.
  \item So $v''<v$. Then a nullification exists for view $v$
    itself, contradicting $b$'s L-notarisation.
\end{itemize}

\paragraph{Liveness} View progression and the forwarding of M-notarisations and nullifications give liveness.
During synchrony, a correct leader selects the greatest $v'$ for which
it holds an M-notarisation and proposes. By the time the proposal
arrives (strictly, within $\delta$ of this time), every correct processor holds the parent's M-notarisation and
the interim nullifications. All correct processors therefore vote, and
there are at least $n-f$ of them, so the proposal receives an
L-notarisation and is finalised.

\paragraph{What carries over} Multimmit inherits this skeleton
wholesale: views with one round of voting for finality at $n-f$,
every view producing either a notarisation of the smaller kind or a
nullification, and proposals anchored as above. There are two  substantive
changes. First, leader blocks contain no transactions; their
role is to coordinate the chains of transaction blocks (described next).
Second, the objects playing the role of M-notarisations, which we call
V-notarisations, certify somewhat more (and are correspondingly larger). We
return to this once the chain layer is in place.

\subsection{The chain layer}

\paragraph{Building $n$ chains} Following Autobahn~\cite{giridharan2024autobahn}, each processor
builds its own chain of transaction blocks, disseminating each block as
it is produced. Upon receiving a block $b$ on $p_i$'s chain, each
processor sends a \emph{DA-vote} (`data availability vote', a
threshold signature share) for $b$
back to $p_i$, who forms and disseminates the resulting
\emph{DA-certificate} from $n-2f$ shares.\footnote{Sending DA-votes to
the block producer alone, who then disseminates the certificate, costs
one extra message delay in certificate formation. The alternative, in
which every processor sends its DA-vote to all, would save that delay
but at the cost of cubic communication per transaction block height: $n$ processors sending to $n$ recipients,
for each of the $n$ chains. Since (as we shall see) leader blocks may
reference transaction blocks before their certificates form,
certificate formation is off the critical path, and the extra delay is
of little consequence.} A DA-certificate ensures
\emph{data availability} (at least $n-3f\ge 2f+1$ correct
processors hold the block, so it can always be recovered) and
\emph{uniqueness}, since two certified blocks at the same height of
the same chain would give two sets of $n-2f$ DA-voters intersecting
in $n-4f\ge f+1$ processors, and correct processors DA-vote at most
once per (chain, height).

\paragraph{Pipelining} Here we deviate from Autobahn, which requires
each transaction block to contain a certificate for its parent. This
puts a full round-trip between consecutive blocks, and so caps the
rate at which any single chain can grow. In Multimmit, a producer may instead run up
to $d$ blocks ahead of its last certified block (the \emph{pipelining
window}, e.g.\ $d=3$), with certificates forming in the background. The
price is that DA-votes may now be cast for a block before any
certificate for its parent is seen, and the DA-voting rules must be
formulated with some care to ensure that certified blocks remain
unique and compatible per chain.

\paragraph{The DA-voting rules} The rules are simple to state. Upon
receiving a block $b$ of height $h$ on $p_i$'s chain, $p_j$ considers
the greatest height $h'<h$ at which it has seen a DA-certificate for
the chain, and DA-votes for $b$ provided:
\begin{itemize}[itemsep=0pt]
  \item $p_j$ has not previously DA-voted for any block of height $h$
    on the chain;
  \item $h-h'\le d$;
  \item $p_j$ holds the full sequence of blocks linking the certified
    block at height $h'$ to $b$, and has DA-voted for each block in
    that sequence below $b$.
\end{itemize}
We leave the reader to convince themselves (or to await
Section~\ref{sec:verification}) that these conditions ensure the
compatibility of certified blocks.

\subsection{Uncertified tips} \label{UT} 

\paragraph{The latency problem} The simplest way to combine the two
layers is the method employed by Autobahn, in which each leader block
references, for each chain, the highest transaction block for which the leader
has seen a DA-certificate, and finalising the leader block then finalises everything below
the referenced tips. This, however, puts three message delays (transaction block,
DA-votes, certificate) between a transaction block and its earliest
possible inclusion in a leader block. We would like the leader to be
able to reference a transaction block \emph{immediately}, with certification
proceeding in the background.

\paragraph{Referencing uncertified tips} Referencing a transaction block before
it is certified is dangerous, because the reference itself does not prove that the data
behind it is available. Autobahn does suggest an optimisation of
exactly this kind, whereby leaders may propose uncertified tips (via
hash values), and a voter
missing the corresponding data holds off voting and fetches it from
the leader. The cost is twofold. Timeouts must be inflated to cover
the time the leader spends serving such requests, so faulty leaders
cause longer delays. Beyond this, the leader may be required to transmit
transaction data, on the critical path to voting, at volumes of
precisely the order that multi-proposer dissemination exists to avoid.
As an example of the latter, suppose that each of $f=\Omega(n)$ faulty
producers sends its latest block to the leader alone. If the leader
references these tips, then every correct voter is missing all
$\Omega(n)$ blocks and must fetch them from the leader before voting.
The leader is then forced to send $\Omega(n)$ blocks to each of $\Omega(n)$
participants, i.e.\ $\Omega(n^2)$ block transmissions through a single
uplink, before the view can progress.

\paragraph{Prefix voting} Raptr~\cite{tonkikh2025raptr} addresses the same problem without
blocking or fetching on the critical path. There, a leader's proposal specifies a sequence
of transaction batches by reference, without requiring that the
referenced data be certified, or even that voters hold it. Each voter
votes for the longest initial segment of the sequence for which it
\emph{does} hold all the referenced data, and the protocol finalises
an initial segment that sufficiently many voters support (roughly
speaking, the prefix voted for within a quorum).
No processor waits or fetches on the critical path, and whatever is finalised is
certainly available.

 Raptr's remaining sensitivity to the order of the proposed
sequence is that votes are
aggregated by minimum over the one sequence that the leader
arranged, so a single withheld batch delays everything later in the
sequence. Avoiding both failures (Autobahn's blocking fetch from the
leader, and Raptr's sensitivity to order) motivates a guiding
principle: \emph{faulty  processors should only be able
to significantly delay the placement of transaction blocks in their
own chains}. (We make precise the extent to which other chains can be
affected at all in Section~\ref{sec:intuition-ordering} below.)

\paragraph{Our approach, in outline} Multimmit responds with two
mechanisms, developed in the next two subsections.
\begin{itemize}[itemsep=2pt]
  \item \emph{Proposal-relative voting.} The leader issues a
    \emph{chain proposal}, proposing tips for every chain, and each
    voter responds with a \emph{single} vote, which is now a vector
    with one coordinate per chain, reporting how far up each chain the
    voter can endorse the chain
    proposal.\footnote{The reader may worry that votes thereby become
    large. Such vectors condense well. A vote that simply endorses
    every proposed tip can be expressed in constant size, and in
    general a vote need record little more than its deviations from
    full endorsement (see Section~\ref{sec:sizes}).} No processor blocks or
    fetches, a withheld block costs only its own chain's slot, and
    leader blocks are finalised after a single round of voting.
  \item \emph{Extension votes.} Votes may themselves endorse fresh
    blocks beyond the proposed tips. For a transaction block to be
    included, it need then only reach the \emph{voters} before their
    next vote, rather than reaching the \emph{leader} before its next
    proposal, removing the proposal leg (one message delay, plus the
    wait for the next leader proposal) from the critical path.
\end{itemize}
Extension votes are also what stands behind the
block-inclusion guarantee of Section~\ref{sec:intro}: a leader
wishing to exclude an honest producer's fresh block must suppress
its entire view.  (The
precise statement, including the sense of `fresh', is the block
inclusion theorem of Section~\ref{sec:verification}.) The guarantee is aimed at precisely such
processors, motivated to finalise their leader blocks but tempted to censor
particular producers or transactions. 

\paragraph{The relationship between V-QCs and L-QCs}\label{sec:intuition-outline}
It
will help throughout to keep the following picture in mind.
Similar to Minimmit (with V-notarisations now replacing M-notarisations), every view
produces either a nullification or a \emph{V-notarisation}, a set of
$n-f$ view $v$ messages, sufficiently many of which are votes for the
view's leader block. Since votes are now vectors, notarisations are
larger objects than their Minimmit counterparts, and the protocol in
fact manipulates succinct certificates representing them, \emph{V-QCs}
and \emph{L-QCs} respectively. It
does no harm to conflate each notarisation with its certificate, and
we do so freely below. Section~\ref{sec:sizes} gives concrete
object-size estimates, while Figure~\ref{fig:traffic-breakdown}
shows that consensus traffic remains a small fraction of application
traffic in the measured configurations. 

A leader block, then,
consists of exactly two components: a chain proposal, proposing tips
for every chain as in the first bullet point above, and a reference
(by hash) to the V-QC of an earlier view, the latter playing the role
of the M-notarisation for its parent. A V-QC specifies, for each chain, a tip that is safe
to extend, and thereby a total ordering, namely everything below the
specified tips, in a fixed sweep order (Section~\ref{sec:intuition-ordering}).
Validity requires the block's chain proposal to extend,
chain by chain, the tips specified by the block's V-QC, just as a
Minimmit proposal must extend its parent. The ordering therefore only ever
grows. Figure~\ref{fig:layers} depicts the two layers as assembled
by the consensus layer.

\begin{figure}[t]
\centering
\begin{tikzpicture}[scale=0.78, every node/.style={font=\small}]
  \foreach \c/\y/\len in {1/0/6, 2/1.05/5, 3/2.1/6}{
    \node[anchor=east] at (-0.35,\y+0.28) {chain $\c$};
    \foreach \i in {1,...,\len}{
      \draw[fill=gray!14] ({(\i-1)*1.02},\y) rectangle ++(0.78,0.56);
    }
  }
  \foreach \x/\y in {2.04/0, 1.02/1.05, 1.02/2.1}{
    \draw[fill=red!25] (\x,\y) rectangle ++(0.78,0.56);
  }
  \foreach \x/\y in {5.1/0, 4.08/1.05, 5.1/2.1}{
    \draw[fill=green!30] (\x,\y) rectangle ++(0.78,0.56);
  }
  \draw[fill=blue!10] (1.55,3.6) rectangle ++(1.5,0.7);
  \node at (2.3,3.95) {$\ell_v$};
  \draw[fill=blue!10] (4.6,3.6) rectangle ++(1.62,0.7);
  \node at (5.41,3.95) {$\ell_{v+1}$};
  \draw[->,thick] (4.6,3.95) -- (3.05,3.95);
  \node[above] at (3.8,3.97) {\footnotesize V-QC};
  \draw[->,dashed] (2.0,3.6) -- (1.41,2.66);
  \draw[->,dashed] (2.3,3.6) -- (1.41,1.61);
  \draw[->,dashed] (2.6,3.6) -- (2.43,0.56);
  \draw[->,dashed] (5.1,3.6) -- (5.49,2.66);
  \draw[->,dashed] (5.4,3.6) -- (4.47,1.61);
  \draw[->,dashed] (5.7,3.6) -- (5.49,0.56);
  \node[anchor=west] at (6.6,0.28) {\footnotesize $\longrightarrow$ height};
\end{tikzpicture}
\caption{The two layers, with chain height increasing from left to right.
Producers extend their own chains of
transaction blocks continuously (bottom), while consensus votes only
on leader blocks (top), each of which references the previous view's
V-QC (solid arrow) and proposes a tip per chain (dashed arrows), red
for $\ell_v$ and green for $\ell_{v+1}$. Each proposal extends the
tips of the view it builds on, and chains grow between proposals
without waiting for the consensus layer.}
\label{fig:layers}
\end{figure}
 An L-QC for a leader block extends the ordering already specified by
the V-QC that leader block references, and every subsequently finalised
leader block extends it further. Since, as we shall see, the tips specified by any
V-QC include the fresh tips of every honest chain (in the common
case), whatever leader block is finalised next, its ordering already contains
all honest chains' transaction blocks. 

\subsection{Proposal-relative voting} \label{PRV} 

Throughout this subsection we describe a simplified version of the
protocol, \emph{without} extension votes. The next subsection adds
them.

\paragraph{Chain proposals and positions} A leader block's chain
proposal contains, for each chain, a \emph{base} (a certificate, or
a tip already agreed safe) followed by references to up to $d$
uncertified blocks. The referenced blocks are not themselves included. Each is
specified by a single hash value, and the hash structure of blocks
ensures, with no need for signatures within the proposal, that each
referenced block is the child of the one before, so that a proposal
names, per chain, a segment of chain extending the base. Nothing
certifies that the referenced blocks \emph{exist}, since a leader may
reference blocks that were never produced; we shall see that this
is harmless. (The reader may also worry that leader blocks grow large;
we give calculations in Section~\ref{sec:sizes} showing that they
remain modest.) Voters respond with  a vote reporting, for
each chain, the greatest position up to which the voter has DA-voted
the referenced blocks. A processor missing blocks on chain $i$ votes a
lower position on chain $i$ and full positions elsewhere. No
processor blocks or fetches on the critical path.  

\paragraph{Two extraction rules} Let us now be precise about the two
kinds of notarisation:
\begin{itemize}[itemsep=2pt]
  \item An \emph{L-notarisation} for a leader block is, exactly as in
    Minimmit, a set of $n-f$ votes for the block. It will be
    convenient to speak also of \emph{pools}: a processor's view-$v$
    pool for a leader block is the set of view-$v$ votes for the
    block that the processor has received. Any $n-f$ votes of a pool
    form an L-notarisation, and the extraction rules below read tips
    directly off pools.
  \item A \emph{V-notarisation} for view $v$ is a set of $n-f$ view
    $v$ messages, sent by distinct processors, each of which is either
    a vote for a view $v$ leader block or a signed \emph{novote}
    message, declaring that its sender did not vote in view $v$. At
    least $2f+1$ of the votes must be for a single leader block, whose
    chain proposal the notarisation is said to \emph{designate}.
\end{itemize}
(Correct processors that time out without voting send novote messages
along with their nullify messages.\footnote{The reader may wonder why
novote messages are needed at all, given nullify messages. The
distinction is that, as we shall see, a correct processor may
sometimes vote and \emph{later} also nullify (this is needed to ensure
view progression), but a correct processor never both votes and sends
a novote message for the same view. This exclusivity is what lets a
V-notarisation account exactly for who did and did not vote, on which
the extraction rules below rely.} Since a faulty leader may
equivocate, we cannot insist that all votes be for the same block, but
if some view $v$ block receives an L-notarisation, one can check that
every view $v$ V-notarisation must designate \emph{its} proposal. The
requirement that a V-notarisation account for $n-f$ distinct
processors, rather than comprising just $2f+1$ votes, is what powers
the quorum intersection arguments below.) As promised, the protocol
actually manipulates succinct certificates standing for these sets,
L-QCs and V-QCs respectively.

With votes now vectors of positions, each kind of notarisation
specifies tips, extracted chain by chain, by rules differing only in
a threshold:
\begin{itemize}[itemsep=2pt]
  \item An L-notarisation finalises the leader block. A pool
    of votes for a finalised leader block then determines
    \emph{finalised tips}: on each chain, take the greatest proposed
    position endorsed by at least $3f+1$ of the pool's votes (with a
    pool of exactly $n-f$ votes, that is to discard the top $3f$).
    Different pools for the same proposal may determine different
    tips, compatible with one another and growing as votes arrive
    (see below).
  \item A V-notarisation licenses view exit (in place of Minimmit's
    M-notarisation), and its \emph{safe-to-extend tips} are what the
    next leader must build on: on each chain, take the greatest
    proposed position endorsed by at least $f+1$ of the
    notarisation's votes for the designated proposal (that is, discard
    the top $f$).
\end{itemize}
The requirement relating the two rules is that, for any view $v$
pool and any view $v$ V-notarisation, the total ordering
extracted from the V-notarisation must extend the ordering
extracted from the pool. This is what consistency needs. A
processor may extract the ordering from its pool, while the next leader
builds on the safe-to-extend tips of whichever V-notarisation it
happens to hold; were some V-notarisation's tips to fall short of a
finalised tip, a later view could extend the chains past finalised
transaction blocks, and the logs of correct processors would fork. Quorum
intersection delivers the requirement, for any $n\ge 5f+1$. Behind
any finalised tip stand at least $3f+1$
votes, hence at least $2f+1$ correct voters, of whom at most $f$ are
absent from any given V-notarisation. At least $f+1$ of the
V-notarisation's votes therefore sit at or above the finalised tip,
so the greatest position endorsed by $f+1$ of its votes, which is
exactly what the safe-to-extend rule selects, is at or above the
finalised tip, on every chain. 
Figure~\ref{fig:rules} illustrates the two rules for a single
chain.

\begin{figure}[t]
\centering
\begin{tikzpicture}[scale=0.9, every node/.style={font=\small}]
  \foreach \i/\col in {0/gray!14, 1/gray!14, 2/gray!14, 3/red!25, 4/green!30}{
    \draw[fill=\col] ({\i*1.15},0) rectangle ++(0.85,0.6);
    \node at ({\i*1.15+0.425},-0.32) {$\i$};
  }
  \node at (0.425,-0.75) {\footnotesize (base)};
  \foreach \i/\cnt in {1/1, 2/1, 3/4, 4/3}{
    \foreach \j in {1,...,\cnt}{
      \fill ({\i*1.15+0.425},{0.85+(\j-1)*0.33}) circle (0.075);
    }
  }
  \draw[->] (6.1,0.95) -- (5.13,0.95);
  \node[anchor=west] at (6.15,0.95) {\footnotesize safe to extend: discard top $f$};
  \draw[->] (6.1,0.3) -- (4.0,0.3);
  \node[anchor=west] at (6.15,0.3) {\footnotesize finalised: discard top $3f$};
\end{tikzpicture}
\caption{The two extraction rules on one chain, for $n=11$, $f=2$.
Chain height increases from left to right. The proposal makes four
entries above the base, and the dots are the reported positions of a pool of $n-f=9$ votes,
here $4,4,4,3,3,3,3,2,1$. Discarding the top $3f=6$ leaves position
$3$ as the finalised tip (red), while discarding the top $f=2$ from
any V-notarisation's tally containing these votes leaves position $4$
as safe to extend (green). The safe-to-extend tip always sits at or
above the finalised tip.}
\label{fig:rules}
\end{figure}

\paragraph{The total ordering extracted from a V-QC}\label{sec:intuition-ordering}
It remains to specify how tips are turned into a log. A transaction
block is \emph{placed in the total ordering} once its position in a
processor's irrevocable ordered output is determined. We call this
\emph{placement}, and say that the processor \emph{emits} the resulting
ordered prefix. Finalised tips alone do not determine that prefix,
and placement does not mean that the data have been retrieved or
executed. We first define the ordering specified by V-QCs, by a
deterministic sweep. This subsection still omits extensions, so only
the proposal pass is needed. The ordering extracted from
a pool of votes, defined below, will then be an initial segment (typically all or most)
of the ordering extracted from any same-view V-QC. To be concrete,
consider a leader block for view $v_2$, and let $Q_1$ be the V-QC it
references, for some earlier view $v_1<v_2$. As described above, $Q_1$
identifies the block's parent, and the block's chain proposal must
extend the tips specified by $Q_1$. Suppose now that view $v_2$
produces a V-QC $Q_2$ for this block. The ordering extracted from
$Q_2$ extends the ordering already extracted from $Q_1$, as follows.
The new transaction blocks on each chain (those lying above the tip specified by
$Q_1$ and at or below the tip specified by $Q_2$) are appended by
height relative to that chain's \emph{previous} tip, interleaving the
chains at equal offsets. Figure~\ref{fig:ordering}
gives an example with four chains, in which the red cells are the tips
specified by $Q_1$, the green cells those specified by $Q_2$, and the
sweep first appends every chain's first new block, then every chain's
second, and so on, producing
\[
1B,\;2C,\;3B,\;4B,\quad 1C,\;2D,\;3C,\;4C,\quad 1D.
\]
Note that chain 2's block $C$ is its \emph{first} new block, and so is
ordered at offset one, alongside the other chains' first new
blocks.

\begin{figure}[t]
\centering
\begin{tikzpicture}[scale=0.95]
  \foreach \c/\h/\x/\y/\col in {%
    1/A/0/0/red!25, 2/A/1/0/white, 3/A/2/0/red!25, 4/A/3/0/red!25,
    1/B/0/1/white, 2/B/1/1/red!25, 3/B/2/1/white, 4/B/3/1/white,
    1/C/0/2/white, 2/C/1/2/white, 3/C/2/2/green!30, 4/C/3/2/green!30,
    1/D/0/3/green!30, 2/D/1/3/green!30}{%
    \draw[fill=\col] (\x,\y) rectangle ++(1,1);
    \node at (\x+0.5,\y+0.5) {$\c\h$};}
  \foreach \x/\c in {0/1,1/2,2/3,3/4}{\node at (\x+0.5,-0.45) {\small $\c$};}
  \foreach \y/\h in {0/A,1/B,2/C,3/D}{\node at (-0.45,\y+0.5) {\small $\h$};}
  \node at (2,-1.0) {\small chains};
\end{tikzpicture}
\caption{Extracting the total ordering. Columns are chains $1$--$4$,
and rows are heights $A$--$D$ ($A$ lowest). This example has no
extension votes, so every displayed new block belongs to the proposal
pass. Red cells are the tips
specified by $Q_1$, and green cells the tips specified by $Q_2$. The
horizontal sweep appends the new blocks at equal offsets above the red
tips, giving $1B,2C,3B,4B$, then $1C,2D,3C,4C$, then $1D$.}
\label{fig:ordering}
\end{figure}

\paragraph{All honest tips are ordered} The crucial point is that the
sweep orders \emph{all} blocks below the tips specified by a V-QC,
however little the leaders of the intervening views proposed for any
given chain, and the tips specified by any V-QC include the proposed
tip of every honest chain (in the common case). To see the latter,
note that, on
an honest chain whose proposed blocks reached all correct processors
in time, \emph{every} correct vote reports the full position, so among
any $2f+1$ votes of a V-QC's tally at least $f+1$ sit at the proposed
tip, and the tip survives the discard of the top $f$. A faulty
processor can lower the extracted tip only on its own chain, by
withholding its own blocks.

\paragraph{The total ordering extracted from a pool} We now
turn to pools. Let $P$ be a pool of votes for a view $v_2$ leader
block that has been finalised, the block referencing the V-QC
$Q_1$. The ordering extracted from $P$ extends the ordering
extracted from $Q_1$ by exactly the sweep above, performed towards
$P$'s finalised tips, but subject to one stopping rule: the sweep halts at
the first slot at which some chain has fallen short of its proposed
tip.\footnote{This form of the rule is exact for the simplified
protocol of this subsection, where the proposed tip is a hard ceiling
on a chain's growth. Once votes may carry extensions (next
subsection), the first pass remains unchanged. A second pass visits
positions above the proposed tips and needs an additional stopping
rule, since these positions may still gain blocks (see the definition
of \emph{settled} chains in Section~\ref{sec:spec}). One
might hope to also \emph{relax} the rule, passing over a lagging chain
when no vote of the pool endorses any block beyond the finalised tip.
This is sound for the beyond-tip region, where growth requires $2f+1$
endorsing votes, but not below the proposed tip, where a chain's
reported position advances at the threshold $f+1$. There, $f$ faulty
votes and a single correct vote not yet received suffice, and no
pool can certify that the latter does not exist.} The stopping rule is needed because the sweep interleaves chains.
A chain that has fallen short of its proposed tip may yet have blocks
appear in that slot (via the view's V-QC), and those blocks would
precede, in sweep order, other chains' blocks at later offsets. To see
that this defines an initial segment of the ordering extracted from
\emph{any} view $v_2$ V-QC, recall the quorum intersection argument
above, by which the finalised tips extracted from a pool sit at or below the
safe-to-extend tips extracted from any same-view V-QC, chain by chain.
In terms of Figure~\ref{fig:ordering}, any V-QC's green frontier
dominates any pool's finalised frontier on every chain. The full protocol makes two separate passes, visiting positions at
proposal heights first, then restarting the horizontal scan from offset
$1$ for positions above the proposed tips, with offsets still measured
from the previous tips. Extensions are introduced next; disagreement
over one chain's extensions defers extension content only. Everything
before the halting slot is thus placed in the total ordering upon
one round of voting. Any finalised blocks deferred by the halt enter
the total ordering by the next finalised leader block, whose inherited
V-QC ordering contains them. Different pools
for the same view yield compatible orderings, and a growing pool
only extends its conclusions, so a processor may safely extend its emitted ordering as further
votes arrive.

This also lets us make the guiding principle of Section \ref{UT} precise. A faulty
non-leader processor cannot keep other chains' blocks out of the
total ordering at all. What it can delay is their placement in the
total ordering, at offsets beyond its own chain's frontier, for
roughly a view.\footnote{The one-view wait is robust, in that it does not rely
on the goodwill of subsequent leaders. A valid proposal must anchor
on a V-QC, and anchoring is all or nothing: the ordering extracted
extends \emph{everything} below the V-QC's safe-to-extend tips,
which carry every honest chain's frontier. Even a Byzantine or
censoring later leader therefore places the deferred blocks in the total ordering the
moment it finalises a leader block; its only alternative is to
finalise nothing and sacrifice its view, after which the same choice
confronts the leader that follows.} This is the precise content of the qualifier `significantly' in
the guiding principle: other chains' blocks are placed in the
total ordering by the next finalised leader block, while the faulty processor's
own chain is the only one whose placement it can significantly hold
up. 

\subsection{Extensions}\label{sec:extensions}

We now add the second mechanism, extension voting. The reader should
take two things from this subsection. First, extensions provide the possibility of removing a
message delay from the critical path, since a transaction block need
only reach the \emph{voters} before they vote, rather than the
\emph{leader} before it proposes. Second, they are what provides the
block inclusion guarantee (Theorem~\ref{thm:inclusion}),
our form of censorship resistance. The mechanism itself is simple,
and the one subtlety lies in how extension blocks enter the total
ordering.

\paragraph{Racing the vote event} A transaction block disseminated just before the
leader formed its proposal would, so far, wait a full view before a
leader can reference it. Multimmit instead lets votes carry
\emph{extensions}, meaning up to $e$ fresh blocks (the
\emph{extension bound}, e.g.\ $e=2$) that the voter has DA-voted
above its reported position on each chain. Extensions are anchored at the voter's reported
position rather than at the proposed tip, helping protect the
inclusion of eligible well-circulated blocks against proposal omissions:
even if a leader's entries for a chain are junk, correct voters report
a lower position and attach the chain's true recent blocks from there. The effect on latency is that,
a fresh block within the extension bound can be endorsed by the
voters' next votes, even if it missed the leader's proposal. When
those votes establish a finalised tip at or above it and the sweep
reaches its position, its transactions enter the total ordering in
that same view.

\paragraph{Ordering extension blocks} Recall how tips were
extracted in the previous subsection. On each chain, we took the
greatest proposed position endorsed by sufficiently many votes. That
rule is sound because the chain proposal names one specific block at
each position, so that all votes on a chain refer to blocks along a
single path, and `the greatest position endorsed by $3f+1$ votes'
unambiguously identifies a block. Extension votes, by contrast, need not refer to blocks on the path
named by the proposal. They begin at the voter's reported position,
which may be below the proposed tip. An equivocating producer may
send incompatible continuations to different voters, so two votes
reaching the same height may endorse different blocks. Endorsement must then be counted
branch by branch. The two extraction rules of the previous subsection
extend accordingly. It is
important to keep the two cases separate.
\begin{itemize}[itemsep=2pt]
  \item \emph{Safe-to-extend tips (V-notarisations).} By the
    chain's \emph{position block} we mean the block selected by the
    rank rule in Section~\ref{PRV} above: the block at the greatest proposed position
    that at least $f+1$ of the notarisation's votes endorse. With extension votes incorporated, the safe-to-extend tip
    becomes the highest block extending the chain's position block
    that at least $2f+1$ of the notarisation's votes endorse (the
    position block standing if there is none). This carry needs no further
    condition, and extraction from V-QCs remains exactly as robust as
    before. In the common case an honest chain's fresh blocks are
    endorsed by every correct vote, hence by at least $2f+1$ votes of
    any V-notarisation's tally, and so enter its safe-to-extend tips
    \emph{whatever the leader proposed} --- junk entries included,
    since extensions are anchored at voters' reported positions. This rule prevents proposal omissions from excluding eligible
    well-circulated blocks from the V-QC's safe-to-extend tips.
  \item \emph{Finalised tips (pools).} A place in the total
    ordering requires more, and at $n=5f+1$ the requirement is
    forced (Section~\ref{sec:exact-threshold} proves this, and
    determines the exact threshold for every $n$). The chain's
    finalised position must first reach the proposed tip. A block
    extending that tip then enters the pool's finalised tips once
    $n-f$ of its votes endorse it. Thus extension finality requires
    both a fully endorsed proposed path and sufficient support for
    the particular continuation. The
    $n-f$ votes need not agree on which extension blocks exist: each
    must merely endorse the block in question, by reporting an
    endorsed tip at or beyond it, and different votes may endorse
    different continuations above it.
\end{itemize}
\paragraph{The two-pass ordering} Tip extraction and placement are
separate steps. The sweep first visits all positions at or below
each chain's proposed tip, by offset above the previous tip and then
by chain index. It then restarts the horizontal scan from offset $1$,
visiting only positions above the proposed tips, in the same offset-first
order. Both passes measure offsets from the previous tips, not the
proposed tips; the second pass can therefore revisit lower offsets.
The pass is determined by a position's
height, not by whether a voter recorded its block in an extension
field. In the first pass, the sweep halts if the finalised position
does not reach a visited position. In the second, it halts at an
unfilled position unless the chain is settled, meaning that no
same-view V-QC can advance beyond its finalised tip. A halt in the first pass
prevents either pass from proceeding; a halt in the second cannot
delay the proposal positions already emitted. V-QC orderings use the
same two passes without a stopping rule.

Suppose a transaction block of a correct producer and the intervening
blocks have been DA-voted by every correct processor before its
consensus vote. Suppose also that the block lies within $e$ of the
previous tip and the chain's proposed entries are genuine. Then the
block enters the finalised tips of the pool containing all correct
votes, whatever the faulty voters do. If the sweep reaches its
position, it is placed in the total ordering in that view. Otherwise,
it is placed by the next finalised leader block.

Junk proposal entries can prevent this in-view placement, but cannot
remove the block-inclusion guarantee. Under the same dissemination
and height conditions, if the faulty leader's block finalises,
every L-notarisation contains at least $n-2f$ correct endorsements
of the transaction block. Every V-QC of that view therefore carries
it into the ordering inherited by the next finalised leader block,
as proved in Section~\ref{sec:verification}.

\subsection{What is achieved, and the comparison with
Raptr}\label{sec:intuition-raptr}\label{sec:achieves}

With the mechanisms now in place, we collect what the protocol
guarantees, case by case, making the comparison with Raptr
quantitative along the way. Precise statements and proofs appear in
Section~\ref{sec:verification}. Recall from
Section~\ref{sec:intuition-ordering} the distinction between the
portion of the total ordering extracted from a pool immediately and
the tail placed in the total ordering by the next finalised leader
block. Two guarantees hold
unconditionally in all cases below: consistency, and the damage
localisation of Section~\ref{sec:intuition-ordering}. We say a
Byzantine leader is \emph{greedy} if it seeks the finalisation of
its leader block (say, for block rewards) and avoids leaving
evidence of provable misbehaviour (say, to avoid slashing)
and of attributable misconduct, junk proposal entries
included (say, to protect reputation).
Operationally, a greedy leader proposes a valid leader block
promptly and equivocates nowhere, since anything else forfeits
finalisation or leaves evidence. None of
our results require this assumption. Greedy leaders are arguably
the realistic censorship adversary, however, and the guarantees
against them are significantly stronger than those against an
adversary with nothing to lose.

Throughout this section we measure latency from the
dissemination of the transaction block itself, rather than from the
leader's proposal, this being the leg of the journey that competing
designs sometimes leave unmeasured. The distinction matters for
user-facing latency. 

 \begin{itemize}[itemsep=3pt]
  \item \emph{The good case} (synchrony, all processors correct).
    In this case, every view finalises. Consider first a transaction block
    placed in the total ordering via inclusion in a chain proposal.
    The block, disseminated at time $t$, reaches the leader by
    $t+\delta$ and waits on average $\delta$ for the next leader
    block, leader blocks being produced every $2\delta$. The leader
    block reaches the voters one $\delta$ later, and a single round
    of voting completes the ordering: the transaction block is placed in the
    total ordering at $t+4\delta$ on average. Extension votes remove
    the journey through the leader: a transaction block that reaches
    the voters before they vote can be endorsed by every vote and
    placed in the total ordering by the view's pool, at $t+3\delta$ on average and $t+2\delta$ at
    best.  In Raptr, a batch disseminated at time
    $t$ reaches the leader by $t+\delta$, and waits on average
    $\delta$ for inclusion in the next leader block, leader blocks
    being produced every $2\delta$. The batch is ordered after that
    block is sent and two rounds of voting complete, at $t+5\delta$
    on average. On the proposal path, Multimmit thus saves one
    $\delta$, bought by the assumption $5f+1\le n$ (one round of
    voting rather than two); where the extension path completes, a
    further $\delta$ is saved (the block need not travel to the
    leader and back out inside a proposal).
  \item \emph{The common case} (synchrony, correct leader, up to $f$
    other processors faulty). Consider first a transaction
    block placed in the total ordering via the correct leader's chain
    proposal, having been disseminated at time $t$ on an honest chain.
    Suppose all correct voters have DA-voted that chain's proposed
    entries before their consensus votes. Each then reports the full
    proposed position, so every correct processor's pool has a
    finalised tip at or above the chain's proposed tip, whatever the
    $f$ faulty voters do. Two subcases arise. If
    the sweep's first pass does not halt ahead of the chain's slots,
    the block is placed in the total ordering in the view itself, at
    $t+4\delta$ on average, exactly as in the good case. If instead
    the first pass halts at a lagging faulty chain's slot, the
    deferred tail is placed in the total ordering, via the view's V-QC, by the next view
    that finalises a leader block: $2\delta$ later at $t+6\delta$.
    Transaction blocks placed via extension votes behave the same
    way, one $\delta$ earlier, giving corresponding placement times $t+ 3\delta$ and $t+ 5\delta$. 
    Raptr's figures can, in contrast, degrade twice over. First, its
    views  stretch from
    $2\delta$ to $3\delta$.\footnote{The early view exit requires a quorum
    certificate on the \emph{full} prefix of the proposal, for
    which at least $f+1$ votes of the quorum, counted in total and
    Byzantine votes included, must each cover every referenced
    batch, Raptr's availability decoupling letting the rest vote
    shorter. A producer that disseminates a batch only to the
    leader, Byzantine voters withholding their full votes, leaves
    a single correct full vote, short of the $f+1$ required, so
    the correct leader can propose the batch it holds while no
    full-prefix certificate forms,} and views then
    advance only upon the certificate formed by the second round of
    voting. With
    leader blocks every $3\delta$, the average wait for inclusion
    grows from $\delta$ to $1.5\delta$. For a batch disseminated
    at time $t$, a proposal including the batch is finalised at
    $t+5.5\delta$ on average. Finalised here means only that the
    \emph{proposal} commits. The
    committed prefix may stop short of the batch, ordering little or
    nothing beyond included certificates. If the committed prefix reaches the batch, it enters the total
    ordering at $t+5.5\delta$; otherwise it enters at the end of the
    next correct-leader view, at $t+8.5\delta$ on average. A third difference
    concerns the next view. In Raptr, a batch that misses one
    proposal waits for the next, so placement in the total
    ordering in the next view is conditional on the next leader
    being correct. In Multimmit, the next-view placement in the
    total ordering above is not: the next leader block must
    anchor on a view $v$ V-QC, whose
    safe-to-extend tips carry every honest chain's frontier. If that
    block is finalised at all, Byzantine or not, it places the
    transaction blocks in question in the total ordering. 
    This brings us to the leader cases, considered next.
  \item \emph{A greedy Byzantine leader.} The two conditions in the
    definition of greed give two corresponding guarantees. Because
    the leader seeks finalisation, the block inclusion theorem
    applies: \emph{either the view finalises nothing, or every fresh
    honest block that reached the voters is endorsed by at least
    $n-2f$ votes of every pool of the view} (no competing block
    ever certifiable, Lemma~\ref{lem:membership}) \emph{and placed
    in the total ordering at latest by the next finalised leader
    block}. A greedy
    leader is therefore unable to finalise its leader block while
    censoring. The guarantee of ordering does not depend on the
    behaviour of subsequent leaders, since \emph{any} later
    finalised leader block orders the blocks in question. A greedy
    leader finalises during synchrony, so the wait is a single view
when that view finalises, and otherwise runs to the next finalised
leader block.
    Because the leader also avoids leaving attributable evidence,
    junk proposal entries are unavailable, as is equivocation. (Junk
    entries are the \emph{leader's} one lever against an honest
    chain's in-view ordering, and are unanswerable under a
    challenge to exhibit the signed headers; the accountability
    paragraph of Section~\ref{sec:verification} describes which
    misbehaviours are provable and which merely attributable.) The
    leader's remaining option is omission, i.e.\ proposing less than
    it received, which is indistinguishable from slowness. Extension
    votes bound the resulting damage. Blocks within the
    extension bound $e$ of the previous tips are endorsed by every
    correct vote whatever the leader proposes, enter the finalised
    tips of any pool of the correct votes, and are placed in the
    total ordering at latest via the next finalised leader block;
    a pool containing a faulty vote that omits such a block may
    defer it until the next finalised leader block. Only the excess on
    under-proposed chains waits for the next view.
  \item \emph{An arbitrary Byzantine leader.} A leader with nothing
    to lose can suppress its entire view, in which case correct
    processors nullify within bounded time, nothing is finalised by
    anyone, and the view is skipped. (The leader-skipping rule of
    Section~\ref{sec:optimisations} shortens the induced wait.)
\end{itemize}

\subsection{Sizes: some concrete numbers}\label{sec:sizes}

The measurements in Section~\ref{sec:measurement-traffic}, particularly Figure~\ref{fig:traffic-breakdown}, show that consensus traffic remains a small fraction of application traffic at the measured high throughput rates. We now give concrete object-size estimates.

We have promised at several points that the objects the protocol
manipulates remain small. Here are rough numbers. Suppose hashes
are 32 bytes, signatures and signature shares (e.g.\ BLS) are 48
bytes, and take $d=3$ and $e=2$. We consider $n=100$ validators, and,
recalling from the start of this section that the number of block
producers may differ from the number of validators, we let $K$ denote
the number of chains and consider both $K=100$ and $K=40$ (the latter
being representative of intended applications). Per-chain vectors
scale with $K$, while quorum-related components scale with $n$.
\begin{itemize}[itemsep=0pt]
  \item \emph{DA-certificates} are constant size: a block header
    ($\approx 75$ bytes) plus one threshold signature, so $\approx
    125$ bytes.
  \item \emph{Chain proposals} contain, per chain, a base (at worst
    a DA-certificate) and up to $d$ payload hashes: under $230$ bytes
    per chain, so $\approx 23$KB at $K=100$ and $\approx 9$KB at
    $K=40$.
  \item \emph{Votes} contain a view number, a leader-block hash and one
    signature ($\approx 90$ bytes), plus a position per chain ($2$
    bits each), plus extensions: at most $e$ hashes per chain, and
    only for chains on which the voter holds blocks beyond the
    proposal. Even with every chain extended by $e$ blocks, this is
    under $7$KB at $K=100$ and under $3$KB at $K=40$; in light load,
    votes are a few hundred bytes.
  \item \emph{Closing reports} (payment chains,
    Section~\ref{sec:payments}) contain a base and up to
    $d_{\mathrm{DA}}$ full headers. With ordinary payment headers
    of about $110$ bytes and $d_{\mathrm{DA}}=3$, a report with its
    base certificate is roughly $500$ bytes per chain, excluding
    shared boundary evidence. A \emph{closing vote} with $40$
    payment chains is then about $20$KB; an uncompressed closing
    L-QC is of order $1$--$2$MB. These are costs per closing view;
    a successful close commonly needs just one such view.
  \item \emph{Receipts} contain a block-level acknowledgment proof
    of about $61$ bytes at $n=100$ (aggregate signature and bitmap),
    the signed payment-block header, and the signed payment with
    its inclusion proof. A binary Merkle proof for a block of
    $1{,}024$ payments adds $320$ bytes. Uncached early-handoff
    adoption evidence is additional. Naming the cut adds one
    $32$-byte hash to each leader block.
  \item \emph{L-QCs and V-QCs}, in the good case, comprise the
    proposal, one reference extension vector, a voter bitmap
    ($n/8\approx 13$ bytes) and a \emph{single} aggregate signature:
    $\approx 30$KB at $K=100$ and $\approx 12$KB at $K=40$. Outside
    the good case, each deviating vote adds its differences from the
    reference. In the worst case, with every quorum vote deviating
    from every other, QCs grow to roughly $0.5$MB at $K=100$ and
    $0.2$MB at $K=40$ (the
    price of per-chain vote vectors); $f$ Byzantine voters alone cause
    at most a quarter of that, benign disagreement concentrates on a
    few chains, and the aggregate signature remains a single group
    element throughout.
  \item \emph{Leader blocks} comprise a hash reference to a V-QC and
    a fresh proposal: $\approx 23$KB at $K=100$ and $\approx 9$KB at
    $K=40$. (The referenced V-QC is not included, travelling
    separately.)
\end{itemize}
Two points deserve emphasis.
\begin{itemize}[itemsep=2pt]
  \item For fixed $n$, a fixed number of chains, and with a fixed rate of transaction block production, the sizes of the steady-state consensus objects are independent of transaction volume: the
    consensus layer never carries transaction data, which travels only
    once, from each producer to all, in the chain layer.
      \item The totals are modest in comparison with what they
    coordinate. Precisely
    because transaction data travels only once, the protocol can
    (modulo execution costs) place transactions in the total
    ordering at roughly the rate
    at which validators can receive data. With views lasting around
    $2\delta$, say $100$--$200$ms, a single view at line rate then
    orders tens of megabytes of transactions on a commodity gigabit
    link ($125$MB/s), and hundreds of megabytes at ten gigabits,
    against consensus objects totalling tens of kilobytes either way.
\end{itemize}
\section{The formal specification}\label{sec:spec}

As in Section~\ref{sec:intuition}, the specification of this
section is of the protocol without the payment system and for a
single epoch, with a fixed committee. Section~\ref{sec:payments}
adds payments and epochs, with the modifications stated there.

In what follows, we suppose that all messages are signed by the
sender. We say `disseminate' to mean `send to all processors'. When a
correct processor is instructed to send a message to itself, it
regards that message as immediately received. The protocol is
parameterised, per chain, by three small constants, discussed in
Section~\ref{sec:intuition}: the \emph{DA window}
$d_{\mathrm{DA}}\in\mathbb{N}_{\ge 1}$ (the pipelining window), the
\emph{proposal window} $d_{\mathrm{prop}}\in\mathbb{N}_{\ge 0}$,
and the \emph{extension bound} $e\in\mathbb{N}_{\ge 0}$.
Outside Section~\ref{sec:payments} every chain has the same
configuration, with $d_{\mathrm{prop}}=d_{\mathrm{DA}}$, and we
write $d$ for the common value; Section~\ref{sec:payments}
configures designated chains differently. For concreteness, one may
take $d=3$ and $e=2$.

The pseudocode uses a number of message types, local variables,
functions and procedures, detailed below. Message types and local
variables are summarised in Tables~\ref{tab:messages}
and~\ref{tab:variables}, and the functions extracting tips and
orderings in Table~\ref{tab:extraction}; the pseudocode appears in
Algorithms~\ref{alg:chain} and~\ref{alg:consensus}.

\paragraph{The local variable $\State$} This variable is maintained
locally by each processor $p_i$ and stores all messages received. It is
considered to be automatically updated, i.e., we do not give explicit
instructions in the pseudocode updating $\State$. Several of the
objects defined below (DA-certificates, nullifications, and QCs) are
formed from sets of constituent messages: any processor may form such
an object upon receiving its constituents.  Accordingly, $\State$ is
regarded as containing such an object whenever it contains an object of
which it is an entry, or messages from which it may be formed; this
too happens automatically, with no explicit instructions in the
pseudocode. Initially, $\State$ contains only the `genesis objects',
defined below: the genesis blocks $\genesis{1},\dots,\genesis{n}$ and a
DA-certificate for each, and the genesis leader block $\lgenesis$
together with the V-QC $Q_{\mathrm{gen}}$ and an L-QC for $\lgenesis$.

\subsection{The chain layer}

\paragraph{Transaction blocks} For each $i\in\{1,\dots,n\}$, the
\emph{genesis block of chain $i$} is the tuple
$\genesis{i}:=(i,0,\lambda,\lambda)$, where $\lambda$ denotes the empty
sequence. A transaction block other than a genesis block is a tuple
\[
b=(i,h,x,\Tr),
\]
signed by $p_i$, where:
\begin{itemize}[itemsep=0pt]
  \item $i\in\{1,\dots,n\}$ (thought of as the chain to which $b$ belongs);
  \item $h\in\mathbb{N}_{\ge 1}$ (thought of as the height of $b$);
  \item $x$ is a hash value (used to specify $b$'s parent), and;
  \item $\Tr$ is a sequence of distinct transactions.
\end{itemize}
We write $b.\mathrm{chain}$, $b.\mathrm{height}$, $b.\mathrm{par}$ and
$b.\Tr$ to denote the corresponding entries of $b$. The \emph{header}
of $b$ is $\hd(b):=(i,h,x,H(\Tr))$; the required signature is
$p_i$'s signature on $\hd(b)$, and we define
$H(b):=H(\hd(b))$.\footnote{The signature is thus verifiable, and
equivocation provable, from headers alone: two blocks signed by $p_i$
at the same height, or a signed block at height $h+1$ whose $x$
differs from the hash of $p_i$'s signed block at height $h$, constitute
constant-size evidence of misbehaviour. Note that the signature is
excluded from $H(b)$, so that the identifier of a block is a function
of its content only. In implementation, $H(\Tr)$ may be taken to be the
root of a Merkle tree over $\Tr$.} The \emph{parent} of $b$ is the
block whose hash is $b.\mathrm{par}$. The parent of a block of height
$h$ is required to be of height $h-1$ on the same chain. The
\emph{ancestors} of $b$ are $b$ and all ancestors of its parent (a
genesis block has only itself as ancestor). Two transaction blocks on
the same chain are \emph{incompatible} if neither is an ancestor of the
other. A block \emph{extends} another if the latter is an ancestor of
the former.

\paragraph{DA-votes} A \emph{DA-vote} (`data availability vote') for the transaction block $b$ by
$p_j$ is a message $(\da,\hd(b),\rho)$, where $\rho$ is $p_j$'s
$(n-2f)$-threshold share on $(\da,\hd(b))$. (A DA-vote is thus
self-describing, specifying the chain, height, and hash of the block
it votes for.)

\paragraph{DA-certificates} A \emph{DA-certificate} for the
transaction block $b$ is a pair $(\hd(b),\tau)$, where $\tau$ is the
$(n-2f)$-threshold signature on $(\da,\hd(b))$, formed from the shares
of $n-2f$ DA-votes for $b$ by different processors. A DA-certificate is thus of
constant size, and specifies the chain, height, parent, and hash of the
block it certifies, and is interpretable (and verifiable) by processors
that have not received $b$ itself. If $c$ is a DA-certificate for $b$,
we write $c.\mathrm{chain}$, $c.\mathrm{height}$ and $c.\mathrm{hash}$
to denote $b.\mathrm{chain}$, $b.\mathrm{height}$ and $H(b)$
respectively. Each genesis block $\genesis{i}$ is considered
DA-certified at the start of the protocol execution.

\paragraph{When $p_j$ may DA-vote for $b$} Suppose $p_j$ has received
a correctly signed transaction block $b=(i,h,x,\Tr)$. Let $h'$ be the greatest
height $<h$ such that $\State$ (as locally defined for $p_j$) contains a DA-certificate $c$ for some
block of height $h'$ on chain $i$. Then $p_j$ \emph{may DA-vote for
$b$} if:
\begin{enumerate}[label=(\roman*),itemsep=0pt]
  \item $p_j$ has not previously sent a DA-vote for any block of height
    $h$ on chain $i$;
  \item $h-h'\le d_{\mathrm{DA}}$;
  \item $\State$ contains blocks $b_{h'+1},\dots,b_{h}=b$ on chain $i$,
    with $b_{h'+1}.\mathrm{par}=c.\mathrm{hash}$ and
    $b_{h''}.\mathrm{par}=H(b_{h''-1})$ for each $h''\in(h'+1,h]$, and
    $p_j$ has sent a DA-vote for $b_{h''}$ for each $h''\in(h',h)$.
\end{enumerate}
Condition (i) ensures that correct processors send at most one DA-vote
per (chain, height) pair; line 4 of Algorithm~\ref{alg:chain} is
understood to process eligible blocks in increasing order of height,
each DA-vote sent counting as already sent for the conditions above
within the same timeslot, so that a processor receiving a path of new
blocks DA-votes them in sequence. Together with condition (iii),
condition (i) ensures that all DA-certified blocks for a given chain
are compatible with each other (Lemma~\ref{lem:cert-compat}). Condition (ii) bounds the number of uncertified blocks at the
head of any chain. DA-votes for $b$ are sent to $p_{b.\mathrm{chain}}$
(rather than disseminated), who forms and disseminates the
DA-certificate.

\paragraph{The predicate $\mathsf{Ready}(\State)$} This predicate is used by $p_i$ to
decide whether to produce the next block on its own chain. In what
follows, $b$ denotes the greatest block $p_i$ has produced on chain
$i$ (or $\genesis{i}$), $h^{c}$ denotes the greatest height at which
$\State$ contains a DA-certificate for chain $i$, and
$\mathsf{pending}(\State)$ denotes the set of transactions received by
$p_i$ and not included in $b$ or any of its ancestors. The predicate
$\mathsf{Ready}(\State)$ holds if:
\begin{enumerate}[label=(\roman*),itemsep=0pt]
  \item (\emph{window}) $b.\mathrm{height}+1-h^{c}\le d_{\mathrm{DA}}$;
  \item (\emph{work}) $\mathsf{pending}(\State)\neq\emptyset$, and;
  \item (\emph{pacing}) $|\mathsf{pending}(\State)|\ge B$, or at least
    $\theta$ time has passed since $p_i$ produced $b$,
\end{enumerate}
where the \emph{batch bound} $B\in\mathbb{N}_{\ge 1}$ and the
\emph{production interval} $\theta\ge 0$ are further protocol
parameters. Condition (i) is required for the analysis, while conditions
(ii) and (iii) are policy, and the analysis requires of them only that,
whenever $\mathsf{pending}(\State)$ remains non-empty, the pacing
condition is satisfied within bounded time.\footnote{Setting $B=1$ and
$\theta=0$ gives eager production: minimal inclusion latency, but under
low load each straggling transaction costs a block of its own, namely a
header, a signature, a round of DA-voting, and a slot of the
$d$-window. Larger $B$ (naturally denominated in bytes rather than
transaction count) amortises this per-block overhead under high load,
while $\theta$ bounds the delay a lone transaction can suffer under low
load. The familiar `batch full or timer expired' policy is exactly
(iii). Under sustained load, production is in any case paced by the
window: at most $d$ blocks per certificate round-trip. Values of
$\theta$ commensurate with the view cadence are natural, since a chain
contributes at most $d$ new blocks to any one leader block.}

\paragraph{The procedure $\ProduceNext()$} This procedure is executed
by $p_i$ to extend its own chain, and may be executed only when
$\mathsf{Ready}(\State)$ holds. Recall from the definition of
$\mathsf{Ready}(\State)$ that $b$ denotes the greatest block $p_i$
has produced on chain $i$ (or $\genesis{i}$). To execute the
procedure, $p_i$:
\begin{itemize}[itemsep=0pt]
  \item Forms a sequence $\Tr$ of distinct transactions comprising
    all of $\mathsf{pending}(\State)$,\footnote{Any inclusion policy
    would do, provided every transaction entering
    $\mathsf{pending}(\State)$ is included by some later execution
    of $\ProduceNext$. Together with the
    requirement, stated with $\mathsf{Ready}(\State)$, that the
    pacing condition be satisfied within bounded time while
    transactions are pending, this is exactly what Step 1 of the
    proof of Theorem~\ref{thm:liveness} uses.}
    and;
  \item Disseminates the block $(i,\,b.\mathrm{height}+1,\,H(b),\,\Tr)$.
\end{itemize}

\subsection{The consensus layer}

\paragraph{The function $\lead(v)$} The value $\lead(v)$ specifies the
leader for view $v$. To be concrete, we set $\lead(v):=p_{j+1}$, where
$j=v \bmod n$.

\paragraph{Chain proposals} A \emph{chain proposal} is a function $C$
mapping each $i\in\{1,\dots,n\}$ to a pair $C(i)=(a_i,x_i)$, where:
\begin{itemize}[itemsep=0pt]
  \item $a_i$ (the \emph{base specifier} for chain $i$) is either a pair
    $(y,h)$ consisting of a hash value $y$ and a height $h$, or a
    DA-certificate for chain $i$, and;
  \item $x_i=(x_i^1,\dots,x_i^{m_i})$ is a possibly empty sequence
    of $m_i\le d_{\mathrm{prop}}$ hash
    values, called \emph{payloads} (thought of as the values $H(\Tr)$
    of a sequence of consecutive blocks on chain $i$, the first
    extending the anchor).
\end{itemize}
The \emph{base} for chain $i$ is the block specified by the anchor. Note that the anchor always specifies the base's
hash and height explicitly, so that a chain proposal is interpretable in
isolation.\footnote{This is deliberate: were a base to default to the
corresponding coordinate of $\Tips(Q)$, the value $\Tips$ (defined
below) would become a recursion over the entire history of V-QCs.
Explicit bases keep every V-QC interpretable in isolation, and tell
a recovering processor exactly which blocks to fetch.} For $0\le k\le m_i$, we write $\Prop(C,i,k)$
for the block specified at \emph{position} $k$ on chain $i$: position
$0$ is the base and, recursively, position $k\ge 1$ is the block with
header $(i,\,h_0+k,\,H(\Prop(C,i,k-1)),\,x_i^k)$, where $h_0$ is the
base's height. That is, entries specify blocks by payload alone, with
chain, height and parent supplied by context. The specified blocks thus form
a chain \emph{by construction}, and the producer's signature on any of
them is verifiable against the reconstructed header.\footnote{Payload
entries do not certify that the specified blocks exist. A leader may
propose junk payloads, in which case (correct) voters simply report
lower positions, with extensions (below) anchored at their reported
positions. Junk entries thus capture neither votes nor extensions,
though they do render the chain unsettled (in the sense defined under
`The total ordering' below), deferring emission of the slots the sweep
visits after it by a view. Such junk entries may defer emission, but cannot defeat
the block-inclusion guarantee.} The \emph{proposed tip} for chain $i$ is
$\Prop(C,i,m_i)$; when $m_i=0$, as is natural for an idle chain, the
proposed tip is the base itself, the proposal making no new entries
for the chain. Whether an anchor of the form $(y,h)$ is
\emph{correct} (names the tip specified by the accompanying V-QC) is a
validity condition on leader blocks, defined below.

\paragraph{Leader blocks} The \emph{genesis leader block} is the tuple
$\lgenesis:=(0,\lambda,\lambda)$. A leader block other than
$\lgenesis$ is a tuple
\[
\ell=(v,q,C),
\]
signed by $\lead(v)$, where $v\in\mathbb{N}_{\ge 1}$ (thought of as the
view corresponding to $\ell$), $q$ is a hash value, and $C$ is a chain
proposal. The value $q$ is required to be the hash $H(Q)$ of a V-QC
$Q$ (defined below) for an earlier view. We then say that $\ell$
\emph{references} $Q$, which is thought of as identifying $\ell$'s
parent, and write $\ell=(v,H(Q),C)$. Note that the V-QC itself is not
included in the block. V-QCs are forwarded to all upon first
receipt, the leader's anchor among them (see $\SelectAnchor$
below), so processors obtain them separately, and leader blocks stay
small. We write $\ell.\mathrm{view}$, $\ell.\mathrm{q}$ and
$\ell.\mathrm{C}$ for the corresponding entries, and define
$H(\ell):=H(v,q,H(C))$. As for transaction blocks, the leader's
signature is excluded from $H(\ell)$, so possession of the tuple
$(v,q,C)$ itself thus determines $H(\ell)$, and so suffices to
verify signatures on votes for $\ell$ (defined below). The QCs below
accordingly include leader blocks directly (without the leader's
signature, which they do not need).

\paragraph{Votes} A \emph{vote} for the leader block $\ell=(v,H(Q),C)$ by
$p_j$ is a message of the form
\[
(\vote,\,v,\,H(\ell),\,\pi,\,\varepsilon,\,s),
\]
where $\pi$ maps each chain $i$ to a position
$\pi(i)\in\{0,\dots,m_i\}$, $\varepsilon$ maps each chain $i$ to a
sequence $\varepsilon(i)$ of at most $e$ payloads, and $s$ is $p_j$'s
signature on $(\vote,v,H(\ell),\pi,\varepsilon)$. Since the single
signature covers the entire content, no party can exhibit a vote while
withholding or altering any part of $(\pi,\varepsilon)$.\footnote{For
the same reason, the QCs below never \emph{abridge} a vote: though
votes are recorded succinctly, the record always determines each
constituent vote's full $(\pi,\varepsilon)$. Were a
`participated in view $v$' component usable in a QC without the
vote's content, the assembler of a V-QC could conceal the positions
of correct voters, and the quorum-intersection argument underlying
$\Tips$ would fail.} The intended semantics, enforced for correct voters by the
voting rules below: $p_j$ has DA-voted for $\Prop(C,i,k)$ for every
$k\in[1,\pi(i)]$, and $\varepsilon(i)$ specifies (by payloads, with
chain, heights and parents supplied by context, as for proposal
entries) a sequence of consecutive blocks on chain $i$ extending
$\Prop(C,i,\pi(i))$, each of which $p_j$ has DA-voted for. We refer to
$\varepsilon$ as the vote's \emph{extensions}; note that extensions are
anchored at the vote's reported position, not at the proposed
tip.\footnote{This protects extension endorsements against junk or
stale proposals: correct voters can report a lower position and attach
the chain's true recent blocks as extensions from there.}

The \emph{endorsed tip} of a vote, for chain $i$, is the last block
specified by $\varepsilon(i)$ if $\varepsilon(i)\neq\lambda$, and
$\Prop(C,i,\pi(i))$ otherwise. We say a vote \emph{endorses} a
transaction block $b$ if $b$ is an ancestor of, or equal to, the vote's
endorsed tip on $b$'s chain. (Thus a vote endorses every block on
the path it endorses, however that path is split between positions and
extensions.) The extraction rules below are stated in terms of this
relation.

\paragraph{Standard votes, deviations, and tallies} Fix a leader block
$\ell$ and a \emph{reference extension vector} $E^*$ (a value of the
same type as $\varepsilon$). Recalling that $m_i$ denotes the number
of payload entries that $\ell$'s chain proposal makes for chain $i$
(so that $\pi(i)=m_i$ expresses full endorsement of the proposed tip), a
vote for $\ell$ is \emph{standard} (w.r.t.\ $E^*$) if $\pi(i)=m_i$
for every $i$ and $\varepsilon=E^*$.
The \emph{deviation record} of a vote is the list of pairs $(i,\pi(i))$
for those chains with $\pi(i)<m_i$, together with $\varepsilon$ if
$\varepsilon\neq E^*$. A \emph{tally} for $\ell$ is a tuple
$T=(E^*,P,D)$, where $P$ is a set of processor indices and $D$ assigns
deviation records to some subset of $P$. The tally describes one vote
for $\ell$ per processor in $P$: standard w.r.t.\ $E^*$ for
processors outside the domain of $D$, and otherwise as given by $D$.
The QCs below consist of a leader block, a tally, and an aggregate
signature (plus, for V-QCs, records of the quorum members not voting
for $\ell$). In the good case (synchrony, all processors correct, no transaction block still in flight at the vote event) all
votes are standard w.r.t.\ a single $E^*$, so $D$ is empty and all
constituent signatures are on a single common message: verification
requires one pairing computation. Each deviating vote adds its
deviation record and one message to the aggregate, and the worst case is
bounded by the quorum size.\footnote{Deviation records may in turn be
compressed, e.g.\ by recording only the coordinates at which a
deviating extension vector differs from $E^*$, and by sharing
identical records across voters; we do not specify this further. A
Raptr-style scheme of per-position keys, under which position
deviations would not fragment the aggregate, does not pay here:
extensions cannot be keyed, ranging as they do over arbitrary hash
values, and extensions concern precisely the blocks racing the vote
event, which is where benign disagreement concentrates.}

\paragraph{The functions $\Positions(\State,\ell)$ and
$\Extensions(\State,\ell)$} Used by $p_j$ to form its vote for
$\ell=(v,H(Q),C)$. $\Positions(\State,\ell)$ maps each $i$ to the greatest
$k\le m_i$ such that $p_j$ has DA-voted for $\Prop(C,i,k')$ for all
$k'\in[1,k]$. $\Extensions(\State,\ell)$ maps each $i$ to the longest
sequence of payloads $(H(b^1.\Tr),\dots,H(b^{m}.\Tr))$ with $m\le e$
such that $b^1,\dots,b^m$ are blocks on chain $i$ that $p_j$ has
DA-voted for, with $b^1$'s parent the block
$\Prop(C,i,\Positions(\State,\ell)(i))$ and $b^{k}$ the parent of
$b^{k+1}$ for each $k<m$ (and to $\lambda$ if there is no such
sequence).

\begingroup\emergencystretch=16pt
\paragraph{Novote and nullify messages} For $v\in\mathbb{N}_{\ge 1}$,
a \emph{novote($v$)} message by $p_j$ is a message $(\novote,v,s)$
with $s$ $p_j$'s ordinary signature on $(\novote,v)$, and a
\emph{nullify($v$)} message by $p_j$ is a message
$(\nullify,v,\varsigma)$ with $\varsigma$ $p_j$'s
$(2f+1)$-threshold share on $(\nullify,v)$. A
correct processor that times out in view $v$ without voting sends
both, while a correct processor that nullifies view $v$ \emph{after} voting
(see Algorithm~\ref{alg:consensus}, lines 16--19) sends only the
latter.
Correct processors never send both a vote for a view $v$ leader block
and a novote($v$) message.\footnote{This property is what separates V-QCs (which certify how far each chain may safely be
extended) from mere view progression, and the analysis relies on it. It
is the reason novote and nullify are distinct message types, and the
reason the former is an ordinary signature (novotes appear individually
attributed inside V-QCs) while the latter is a threshold share
(nullifications need certify nothing about who nullified).}

\par\endgroup
\paragraph{Nullifications} A \emph{nullification} for view $v$ is a
message $(\nullification,v,\tau)$, where $\tau$ is the
$(2f+1)$-threshold signature on $(\nullify,v)$, formed from the shares
of nullify($v$) messages by $2f+1$ different processors. (A
nullification is thus of constant size, and self-describing: it
specifies the view it nullifies and what $\tau$ is a signature on. The
tag also syntactically distinguishes nullifications from nullify($v$)
messages, whose shares are group elements of the same kind as
$\tau$.)

\paragraph{L-QCs} An \emph{L-QC} for the leader block $\ell$ is a
tuple
\[
(\ell,\,T,\,s),
\]
where $T=(E^*,P,D)$ is a tally for $\ell$ with $|P|=n-f$, and $s$ is
the aggregate of the signatures of the $n-f$ votes described by $T$. An
L-QC thus certifies a set of $n-f$ votes for $\ell$ by distinct
processors, and we identify it with that set.

\paragraph{V-QCs} A \emph{V-QC} for view $v$ is a tuple
\[
Q=(\ell,\,T,\,N,\,X,\,s),
\]
certifying a set $W$ of $n-f$ view~$v$ messages by distinct processors,
where, writing $C$ for the chain proposal of $\ell$:
\begin{itemize}[itemsep=0pt]
  \item $T=(E^*,P,D)$ is a tally for $\ell$ with $|P|\ge 2f+1$; the
    votes it describes are referred to as the votes \emph{for $C$} in
    $W$;
  \item $N$ is a set of processor indices, each contributing a
    novote($v$) message to $W$;
  \item $X$ assigns, to some further set of processor indices, records
    determining complete votes $(v,H(\ell'),\pi,\varepsilon)$ for
    other view $v$ leader blocks $\ell'$,\footnote{Votes for other leader blocks can exist
    only if $\lead(v)$ equivocates, in which case any two such blocks,
    being signed, are evidence of misbehaviour. Note that the records
    in $X$ determine such votes in full: per the atomicity footnote
    above, a QC's record of a vote (however succinctly encoded) must
    always determine the vote's complete content.} and;
  \item $s$ is the aggregate of the signatures of all messages in $W$
    (the novotes contributing a single common message to the
    aggregate).
\end{itemize}
The sets $P$ and $N$ and the domain of $X$ are disjoint, with sizes
summing to $n-f$. We write $Q.\mathrm{view}$ for $v$, and identify $Q$
with the pair $(W,C)$. There is a distinguished initial V-QC
$Q_{\mathrm{gen}}$ for view $0$, considered received by all processors
at the start of the execution, with $\Tips(Q_{\mathrm{gen}})_i :=
\genesis{i}$ for each $i$.

\paragraph{The function $\Tips(Q)$} Defined for any V-QC $Q=(W,C)$,
the value $\Tips(Q)$ maps each chain $i$ to a transaction block, and is
thought of as specifying, for each chain, a tip that is \emph{safe to
extend} given $Q$. For each chain $i$, let $V_i$ be the votes for $C$
in $W$ (so $|V_i|\ge 2f+1$), and proceed as follows.
\begin{itemize}[itemsep=0pt]
  \item \emph{Positions.} Order the values $\{\pi(i):\pi\in V_i\}$ from
    greatest to least (with multiplicity), discard the first $f$, and
    let $k_i$ be the next value. Set $b_i:=\Prop(C,i,k_i)$.
  \item \emph{Extension carry.} Redefine $b_i$ to be a block of greatest height
    on chain $i$ that extends $b_i$ and that at least $2f+1$ votes in
    $V_i$ endorse, if there is one, choosing deterministically (say,
    by least hash) if there is more than one.\footnote{For $n=5f+1$
    the choice never arises: two incompatible candidates would each
    require $2f+1$ endorsing votes, no vote endorses both, and
    $|V_i|\le n-f=4f+1$. For larger $n$ it may arise, but any
    deterministic choice is safe. To see this, suppose an extension
    block $t$ on chain $i$ is finalised by some same-view L-QC: all
    $n-f$ of that L-QC's votes endorse $t$, so at least $n-2f$
    correct processors do, of which at least $n-3f\ge 2f+1$ appear in
    $V_i$; meanwhile any block incompatible with $t$ has at most $f$
    correct endorsers, so musters at most $2f$ votes. Every candidate
    here is thus compatible with $t$, and every candidate of \emph{greatest height}
    (having height at least that of $t$, which itself qualifies)
    extends $t$.}
\end{itemize}
Then $\Tips(Q)_i:=b_i$. Since anchors specify bases explicitly, each
$\Tips(Q)_i$ is given as an explicit (hash, height) pair by $Q$ alone:
no other V-QC is needed to interpret it. Note also that each coordinate
of $\Tips(Q)$ has been DA-voted by at least one correct processor (at
least $f+1$ when the extension carry applies), which therefore
possesses the block and all of its ancestors down to a DA-certified
block. This secures data availability for all blocks entering the
total ordering defined below, quantified precisely by the
availability lemma of Section~\ref{sec:verification}.

\paragraph{The function $\TipsF(Q)$} Defined for any L-QC $Q$ for
the leader block $\ell=(v,H(Q'),C)$, the value $\TipsF(Q)$ maps each chain
$i$ to a transaction block, and is thought of as specifying, for each
chain, the tip \emph{finalised} by $Q$. For each chain $i$, proceed as
follows.
\begin{itemize}[itemsep=0pt]
  \item \emph{Positions.} Order the values $\{\pi(i)\}$ over the $n-f$
    votes in $Q$ from greatest to least (with multiplicity), discard
    the first $3f$, and let $k^*_i$ be the next value.
  \item \emph{Extensions.} If $k^*_i=m_i$, then $b^*_i$ is the
    highest block on chain $i$ extending $\Prop(C,i,m_i)$ that
    \emph{every} vote in $Q$ endorses (or $\Prop(C,i,m_i)$ itself if
    there is none). If $k^*_i<m_i$, then $b^*_i:=\Prop(C,i,k^*_i)$.
\end{itemize}
Then $\TipsF(Q)_i:=b^*_i$. Note the asymmetry of thresholds: positions
use the rank rules (discard $f$, respectively $3f$), while extensions
use endorsement thresholds ($2f+1$, respectively unanimity above a
fully endorsed proposal path). One might hope
to treat both by a single rank rule, counting extension lengths along
with positions. The rank rules are sound precisely because the
proposal specifies one block per position, so votes on a chain's
positions lie on a common path and a rank determines a block.
Extensions need not follow that path. Under an equivocating producer, votes may split
between incompatible branches, a rank then determines only a height,
and the branch-aware repair at the same thresholds is impossible;
whence the strengthened thresholds, which are forced at
$n=5f+1$.\footnote{Why these thresholds, and why the condition
$k^*_i=m_i$ cannot be dropped, is the business of
Section~\ref{sec:verification}: Lemma~\ref{lem:safe-extend} proves
the extraction sound, and Theorem~\ref{thm:exact} shows the
thresholds exact at $n=5f+1$. Guarantees short of a place in the
total ordering are available at weaker thresholds (the block
inclusion theorem).}

\paragraph{The total ordering} Suppose $T$ maps each chain to a
transaction block on that chain. The \emph{positions above $T$} are the pairs $(o,i)$, for
offsets $o\ge 1$ and chains $1\le i\le n$; position $(o,i)$ refers to
height $T_i.\mathrm{height}+o$ on chain $i$. The \emph{horizontal
order} on these positions is lexicographic with offset taken first:
$(o,i)$ precedes $(o',i')$ if $o<o'$, or if $o=o'$ and $i<i'$. (So the
positions at offset $1$ come first, ordered by chain index, then those
at offset $2$, and so on.) Where a chain proposal $C$ is in
scope, as everywhere below, we use two separate passes. The proposal
pass visits only positions at or below their chain's proposed tip,
in horizontal order. After completing it, restart the horizontal scan
from offset $1$ for the extension pass, visiting only positions strictly
above those tips. Both passes measure offsets from $T$, not from the
proposed tips. Equivalently, positions are ordered lexicographically by
(pass number, offset, chain index), so the offset can decrease at the
pass boundary. Each position's pass is read off $C$, which every
processor computing the sequence holds. The resulting cross-chain order is deterministic
conditional on the proposal: for a fixed sequence of valid
proposals, every correct processor computes the same
interleaving.
If $T'$ also maps chains to blocks, $\Horiz_C(T,T')$ denotes the
sequence obtained by visiting the positions above $T$ in the two-pass
order for $C$ and appending, at each position whose height is at
most $T'_i.\mathrm{height}$ on its chain $i$, the ancestor of $T'_i$
at that height. (In the intended case $T'_i$ extends $T_i$; the
definition does not require this, however, and when a proposal's
anchor jumps to a DA-certificate on a branch incompatible with $T_i$
--- possible only under an equivocating producer --- the appended
blocks are simply read off $T'_i$'s own path.) Note that, for
each $i$, \emph{every} block on $T'_i$'s path above height
$T_i.\mathrm{height}$ appears in the sequence: the definition
determines only how the chains are interleaved, never which blocks are
included.\footnote{This is the
\emph{horizontal} ordering; heights are measured relative to each
chain's previous tip.} Each V-QC $Q$ then
specifies a sequence $\Ord(Q)$ of transaction blocks, defined
recursively: $\Ord(Q_{\mathrm{gen}}):=\lambda$ and, where $Q=(W,C)$ is
a V-QC for the leader block $\ell=(v,H(Q'),C)$,
\[
\Ord(Q):=\Ord(Q')^\frown\Horiz_C\bigl(\Tips(Q'),\Tips(Q)\bigr).
\]
Similarly, each L-QC $Q$ for $\ell=(v,H(Q'),C)$ specifies a sequence
\[
\Emit(Q):=\Ord(Q')^\frown\rho,
\]
defined as follows. Say chain $j$ is \emph{settled} (w.r.t.\ $Q$) if
$k^*_j=m_j$ and no vote in $Q$ endorses any block strictly extending
$\TipsF(Q)_j$; otherwise $j$ is \emph{unsettled}.\footnote{If $j$ is settled, then no
V-QC for the view can specify a tip for chain $j$ beyond
$\TipsF(Q)_j$. Positions cannot advance, being capped at the proposed
tip, an ancestor of $\TipsF(Q)_j$ (or equal to it). An extension
carry beyond $\TipsF(Q)_j$ would require
$2f+1$ votes endorsing a block beyond, hence at least $f+1$ correct
ones, while no vote in $Q$ endorses such a block and at most $f$
correct votes lie outside $Q$. The condition $k^*_j=m_j$ cannot be
dropped, since below the proposed tip a V-QC's tip advances at the
threshold $f+1$, which $f$ faulty votes and a single correct vote from
outside $Q$ can meet, and no L-QC can certify that the latter does not
exist. An unsettled
chain may gain blocks above $\TipsF(Q)_j$ in the eventual ordering,
and blocks so gained could precede blocks of other chains later in
the same pass. This is why the sweep must halt when an unsettled
chain cannot fill a visited position. Uncertainty in the extension
pass does not affect the completed proposal pass. The possible
heights are bounded: a vote reports a position at or below the
proposed tip and at most $e$ extensions above that position, so no
endorsed tip lies more than $e$ above the proposed tip.} To form $\rho$, visit the positions above
$\Tips(Q')$ in the two-pass order relative to $C$ and, at
each position, on chain $j$ say, act as follows. In the first
pass (positions at heights at most the proposed tip's):
\begin{itemize}[itemsep=0pt]
  \item if the position's height is at most that of
    $\Prop(C,j,k^*_j)$, append the ancestor of $\TipsF(Q)_j$ at
    that height;
  \item otherwise, halt, abandoning both passes.
\end{itemize}
If the first pass completes, restart from offset $1$ for the second
pass (positions strictly above the proposed tips), retaining $\Tips(Q')$
as the origins:
\begin{itemize}[itemsep=0pt]
  \item if the position's height is at most that of $\TipsF(Q)_j$,
    append the ancestor of $\TipsF(Q)_j$ at that height;
  \item otherwise, if chain $j$ is settled w.r.t.\ $Q$,
    append nothing and continue;
  \item otherwise, halt.
\end{itemize}
Then $\rho$ is the sequence of blocks appended before the halt.
A halt in the first pass abandons both passes. A halt in the second
therefore defers extension content only, leaving every emitted
proposal position unchanged. If
every chain is settled, no halt ever occurs and $\rho$ is the whole of
$\Horiz_C(\Tips(Q'),\TipsF(Q))$. If an L-QC's votes agree on the full proposed positions and the
same extensions, every chain is settled and the entire ordering
through those tips is emitted immediately. Correctness of all
processors alone does not require their votes to agree when fresh
transaction blocks arrive at different times. The sequence of transactions corresponding to $\Emit(Q)$
(concatenating the values $b.\Tr$ in order, removing duplicates) is
what a processor appends to its log upon receiving the L-QC $Q$. The
analysis establishes that $\Emit(Q)\preceq\Ord(Q'')$ for any V-QC
$Q''$ of the same view, so that logs remain consistent as later views
extend the ordering.

\paragraph{Finalisation from the vote pool} Since any $n-f$ votes for
$\ell$ form an L-QC, a processor holding a pool $P$ of at least
$n-f$ votes for $\ell$, by distinct processors (a faulty processor
may sign more than one vote for $\ell$; the pool retains the first
it receives, and never replaces it), need never enumerate L-QCs. The rules below apply to $P$
directly, and it is these that the analysis of
Section~\ref{sec:verification} treats, with an L-QC's rules as the
special case $|P|=n-f$; they dominate the rules for every individual
L-QC assemblable from $P$, and settledness from a pool is strictly
more permissive than from any single L-QC, so pooling can emit
strictly more. Formally, we extend the domains of $\TipsF$,
settledness and $\Emit$ from L-QCs to pools, the rules below
reducing to the rules already given in the case $|P|=n-f$. Define
$\Emit(P):=\Ord(Q')^\frown\rho$ exactly as for L-QCs, with
$\TipsF(P)$ and settledness with respect to $P$ given, for each
chain $i$, as follows, writing $b^*_i:=\TipsF(P)_i$ for brevity:
\begin{itemize}[itemsep=0pt]
  \item \emph{Positions.} The position $k^*_i$ is the $(3f+1)$-th largest value
    $\pi(i)$ over the votes in $P$ (this is the drop-$3f$ rule applied
    to the best $n-f$ votes for chain $i$);
  \item \emph{Extensions.} If $k^*_i=m_i$, then $b^*_i$ is the highest
    block extending $\Prop(C,i,m_i)$ that at least $n-f$ votes in $P$
    endorse (or $\Prop(C,i,m_i)$ itself if there is none). If
    $k^*_i<m_i$, then $b^*_i:=\Prop(C,i,k^*_i)$;
  \item \emph{Settledness.} Chain $i$ is settled if $k^*_i=m_i$ and
    $\beta_i+(n-|P|)\le f$, where $\beta_i$ is the number of votes in
    $P$ endorsing some block strictly extending $b^*_i$.
\end{itemize}
The settledness condition generalises the L-QC rule (to which it
reduces when $|P|=n-f$). It ensures that at most $f$ correct votes
endorsing blocks beyond $b^*_i$ can exist at all, seen or unseen, so no V-QC can
carry chain $i$ past $b^*_i$. As votes arrive, finalised tips never
retreat and $\beta_i+(n-|P|)$ never increases: a new vote adds at
most one to $\beta_i$ and removes one from $n-|P|$, while advancing
the tip only removes extending votes from the count. Thus emission
is monotone and dominates every L-QC assemblable from $P$, so a
processor may re-emit as its pool grows.

\paragraph{The function $\SelectAnchor(\State,v)$} Used by
the leader of view $v$ to select the V-QC to build on. If
$v'<v$ is the greatest view such that $\State$ contains a V-QC
for view $v'$, the function outputs the V-QC for view $v'$ that
$p_i$ forwarded upon its becoming new (the choice among certificates simultaneously assemblable at that
timeslot being fixed by the forwarding rule; see `New nullifications
and V-QCs' below). The leader thus anchors on an object it has already
sent to all.

\paragraph{The procedure $\ProposeChains(Q,v)$} Executed by the leader
$p_i$ of view $v$ to determine a new leader block. To execute the
procedure, $p_i$ forms the chain proposal $C$ as follows, for each
chain $j$: let $b:=\Tips(Q)_j$;
\begin{itemize}[itemsep=0pt]
  \item if $\State$ contains a DA-certificate for a block on chain $j$
    of height greater than $b.\mathrm{height}$, then $a_j$ is such a
    certificate of greatest height; otherwise
    $a_j:=(H(b),b.\mathrm{height})$;
  \item $x_j$ is the longest sequence of payloads
    $(H(b^1.\Tr),\dots,H(b^{m}.\Tr))$, $m\le d_{\mathrm{prop}}$, such that
    $b^1,\dots,b^m$ are blocks on chain $j$ that $p_i$ has DA-voted
    for, with $b^1$'s parent the base and $b^k$ the parent of
    $b^{k+1}$ for $k<m$.
\end{itemize}
Then $p_i$ disseminates the leader block $(v,H(Q),C)$. The V-QC
$Q$ requires no dissemination of its own, since $\SelectAnchor$
returns exactly the object that $p_i$ forwarded to all upon its
becoming new.\footnote{That the anchor be an object the leader has
already sent to all is necessary, not stylistic. A V-QC assembled
at proposal time from the leader's received votes may contain a
privately delivered vote, so that no other correct processor holds
the exact object referenced, and none could then vote. Anchoring on
the forwarded V-QC rules this out, and the forwarding also delivers
the anchor in time for the correct-leader lemma of
Section~\ref{sec:verification}.}

\paragraph{When $\State$ contains a valid proposal for view $v$} This
condition is satisfied when:
\begin{enumerate}[label=(\roman*),itemsep=0pt]
  \item $\State$ contains precisely one leader block $\ell$ of the
    form $\ell=(v,q,C)$ signed by $\lead(v)$;
  \item $\State$ contains a V-QC $Q$ with $H(Q)=q$, for some view
    $v'<v$, and $C$ is well-formed with
    respect to $\Tips(Q)$;\footnote{Well-formedness requires, for each
    chain $j$: either $a_j=(H(\Tips(Q)_j),\,\Tips(Q)_j.\mathrm{height})$,
    or $a_j$ is a valid DA-certificate for chain $j$ of height greater
    than that of $\Tips(Q)_j$; and $|x_j|\le d_{\mathrm{prop}}$. Note that $\Tips(Q)$
    is computable from $Q$ alone, so this condition is checkable given
    $\ell$ and the referenced V-QC. Parent-links among proposed transaction blocks hold by
    construction (entries are payloads); voters do not verify that the
    specified blocks \emph{exist}, and simply report a lower position
    when they have not DA-voted them.} and;
  \item $\State$ contains a nullification for each view in the open
    interval $(v',v)$.
\end{enumerate}
When (i)--(iii) are satisfied w.r.t.\ $\ell$, we say $\State$ contains
a valid proposal $\ell$ for view $v$. Note that this condition guards
only a processor's \emph{first} route to voting (line 8 of
Algorithm~\ref{alg:consensus}). A processor may also vote upon
receiving a V-QC for the current view (lines 12--13), \emph{without}
checking validity. In this case,  the V-QC's $2f+1$ votes for the designated proposal
ensure that at least $f+1$ correct processors found it valid, and the
V-QC contains everything needed to form the vote (it contains $\ell$,
which determines $H(\ell)$, and $\Positions$ and $\Extensions$
require only $C$ and the voter's own DA-voting history).

\paragraph{Further local variables} The variable $\mathtt{v}$
(initially $1$) specifies the present view. The local timer
$\mathtt{T}$ (initially $0$) increments in real time and is reset upon
entering each view. The variables $\mathtt{nullified}$,
$\mathtt{proposed}$ (initially $\mathrm{false}$) and
$\mathtt{notarised}$ (initially $\bot$) record whether $p_i$ has sent a
nullify($\mathtt{v}$) message, whether it has proposed a leader block
for view $\mathtt{v}$, and the leader block it has voted for in the
present view; all are reset upon entering a new view.

\paragraph{New nullifications and V-QCs} As in Minimmit, processors
forward all newly received or assembled nullifications and V-QCs to
all. At timeslot $t$, $p_i$ regards a nullification for view $v$ as
\emph{new} if $\State$ did not contain one for view $v$, nor shares
forming one, at any smaller timeslot (since threshold signatures are
unique, no tie-breaking rule is needed), and regards a V-QC for view
$v$ as \emph{new} if $\State$ contained no V-QC for view $v$ at any
smaller timeslot, forwarding one such V-QC, chosen by a fixed
efficient rule, say assembled greedily from the earliest-received
qualifying messages.\footnote{Unlike Minimmit's
M-notarisations, V-QCs are not constant-size
(Section~\ref{sec:sizes}), but forwarding them is unalarming. In
terms of egress, note that under standard operation the leader of the
next view in any case sends both a V-QC (its forwarding duty, which
doubles as the dissemination of its anchor, per $\SelectAnchor$)
and its own leader block. Each other processor's
forwarding duty, a single V-QC to all, is therefore less than the
per-view egress the protocol already demands of each leader.
Moreover, in the common case a processor forwards at
a moment when it has nothing else to send, and nothing in the protocol
waits on the transmission, so the cost is bandwidth off the critical
path. In terms of ingress, a processor already holding the V-QC need
not attend to more than a brief identifier at the head of the
message.}

\begin{table}[t]
\centering
\small
\begin{tabular}{ll}
\toprule
Message type & Description\\
\midrule
$\genesis{i}$ & The tuple $(i,0,\lambda,\lambda)$\\
transaction block $b\neq\genesis{i}$ & A tuple $(i,h,x,\Tr)$, signed by $p_i$ on $\hd(b)=(i,h,x,H(\Tr))$\\
DA-vote for $b$ & $(\da,\hd(b),\rho)$: $\rho$ an $(n-2f)$-threshold share on $(\da,\hd(b))$\\
DA-certificate for $b$ & $(\hd(b),\tau)$: $\tau$ the $(n-2f)$-threshold signature on $(\da,\hd(b))$\\
$\lgenesis$ & The tuple $(0,\lambda,\lambda)$\\
leader block $\ell\neq\lgenesis$ & A tuple $(v,H(Q),C)$, signed by $\lead(v)$; $H(\ell)=H(v,H(Q),H(C))$\\
vote for $\ell$ & A message $(\vote,v,H(\ell),\pi,\varepsilon,s)$\\
tally $T$ for $\ell$ & $(E^*,P,D)$: describes one vote for $\ell$ per processor in $P$\\
novote($v$) & $(\novote,v,s)$: $s$ an ordinary signature on $(\novote,v)$\\
nullify($v$) & $(\nullify,v,\varsigma)$: $\varsigma$ a $(2f+1)$-threshold share on $(\nullify,v)$\\
nullification for $v$ & $(\nullification,v,\tau)$: $\tau$ the $(2f+1)$-threshold signature on $(\nullify,v)$\\
V-QC for $v$ & $(\ell,T,N,X,s)$: certifies $n-f$ votes/novotes, $\ge 2f+1$ for $C$\\
L-QC for $\ell$ & $(\ell,T,s)$: certifies $n-f$ votes for $\ell$, distinct signers\\
\bottomrule
\end{tabular}
\caption{Messages.}\label{tab:messages}
\end{table}

\begin{table}[t]
\centering
\small
\begin{tabular}{ll}
\toprule
Variable & Description\\
\midrule
$\mathtt{v}$ & Initially 1, specifies the present view\\
$\mathtt{T}$ & Initially 0, a local timer reset upon entering each view\\
$\mathtt{nullified}$ & Initially false, whether already sent nullify($\mathtt{v}$)\\
$\mathtt{proposed}$ & Initially false, whether already proposed a leader block for view $\mathtt{v}$\\
$\mathtt{notarised}$ & Initially $\bot$, records leader block voted for in present view\\
$\State$ & Records all received messages, automatically updated\\
 & Initially contains genesis blocks, certificates and quorums as specified above\\
\bottomrule
\end{tabular}
\caption{Local variables. The parameters $d_{\mathrm{DA}}$, $d_{\mathrm{prop}}$ (DA and proposal windows, equal and then written $d$, unless a chain is configured otherwise), $e$ (extension bound), $B$ (batch bound) and $\theta$ (production interval) are fixed protocol parameters, with the per-chain configurations stated above.}\label{tab:variables}
\end{table}

\begin{table}[t]
\centering
\small
\begin{tabular}{lp{10.4cm}}
\toprule
Notion & Description\\
\midrule
$\Prop(C,i,k)$ & The transaction block specified at position $k$ on chain $i$ by chain proposal $C$ (position $0$ being the base)\\
$m_i$ & The number of payload entries $C$ makes for chain $i$; the proposed tip is $\Prop(C,i,m_i)$\\
endorsed tip & Of a vote, per chain: the last block its extensions specify, or its position block $\Prop(C,i,\pi(i))$ if none\\
endorses & A vote endorses $b$ if $b$ is an ancestor of, or equal to, its endorsed tip on $b$'s chain\\
$\Tips(Q)$ & Safe-to-extend tips of the V-QC $Q$: per chain, the $(f{+}1)$-th largest reported position, then the highest block extending it that $2f{+}1$ tally votes endorse\\
$k^*_i$ & The $(3f{+}1)$-th largest reported position for chain $i$, over an L-QC's votes (or a pool)\\
$\TipsF(Q)$, $\TipsF(P)$ & Finalised tips of the L-QC $Q$, or of the pool $P$ (abbreviated $b^*_i:=\TipsF(P)_i$): per chain, $\Prop(C,i,k^*_i)$ and, if $k^*_i=m_i$, the highest block extending it that \emph{every} vote endorses ($n-f$ votes, for a pool)\\
settled & Chain $i$, w.r.t.\ an L-QC or pool: $k^*_i=m_i$ and at most $f$ correct votes endorsing blocks beyond $b^*_i$ can exist ($\beta_i+(n-|P|)\le f$)\\
$\Horiz_C(T,T')$ & The transaction blocks above $T$ on $T'$'s paths: proposal positions first, then a scan restarted at offset $1$ for extension positions; both scans use offsets from $T$\\
$\Ord(Q)$ & The total ordering extracted from the V-QC $Q$: recursively, $\Ord$ of the referenced V-QC followed by $\Horiz_C$ towards $\Tips(Q)$\\
$\Emit(Q)$, $\Emit(P)$ & The emitted prefix: $\Ord$ of the referenced V-QC, then a proposal pass halting at an unreached position, followed by an extension pass skipping unfilled positions only on settled chains\\
\bottomrule
\end{tabular}
\caption{The extraction apparatus, as defined in this section.}\label{tab:extraction}
\end{table}

\begin{algo}{the chain layer, instructions for $p_i$}{alg:chain}
\aline{At every timeslot $t$:}
\aline[1]{If $\mathsf{Ready}(\State)$: \acomment{As defined in Section~\ref{sec:spec}}}
\aline[2]{$\ProduceNext()$; \acomment{Extend own chain}}
\aline[1]{For each block $b$ such that $p_i$ may DA-vote for $b$: \acomment{As defined in Section~\ref{sec:spec}}}
\aline[2]{Send a DA-vote for $b$ to $p_{b.\mathrm{chain}}$;}
\aline[1]{If $\State$ contains $n-2f$ DA-votes, with distinct signers, for a block $b$ produced by $p_i$,}
\aline[1]{and $p_i$ has not yet disseminated a DA-certificate for $b$:}
\aline[2]{Form and disseminate the DA-certificate for $b$;}
\end{algo}

\begin{algo}{the consensus layer, instructions for $p_i$}{alg:consensus}
\aline{At every timeslot $t$:}
\aline[1]{Disseminate new nullifications and new V-QCs in $\State$; \acomment{`new' as defined in Section~\ref{sec:spec}}}
\aline[1]{If $p_i=\lead(\mathtt{v})$ and $\mathtt{proposed}=\mathrm{false}$:}
\aline[2]{$\ProposeChains(\SelectAnchor(\State,\mathtt{v}),\mathtt{v})$; \acomment{Send out a new leader block}}
\aline[2]{Set $\mathtt{proposed}:=\mathrm{true}$;}
\aline[1]{If $\State$ contains a valid proposal $\ell$ for view $\mathtt{v}$: \acomment{As defined in Section~\ref{sec:spec}}}
\aline[2]{If $\mathtt{notarised}=\bot$ and $\mathtt{nullified}=\mathrm{false}$:}
\aline[3]{Set $\mathtt{notarised}:=\ell$ and disseminate the vote for $\ell$ with}
\alinecont[3]{$\pi=\Positions(\State,\ell)$ and $\varepsilon=\Extensions(\State,\ell)$; \acomment{Send vote}}
\aline[1]{If $\mathtt{T}=2\Delta$, $\mathtt{nullified}=\mathrm{false}$ and $\mathtt{notarised}=\bot$:}
\aline[2]{Set $\mathtt{nullified}:=\mathrm{true}$ and disseminate $(\novote,\mathtt{v})$ and $(\nullify,\mathtt{v})$; \acomment{Timeout}}
\aline[1]{If $\State$ contains a nullification or a V-QC for view $\mathtt{v}$:}
\aline[2]{If $\State$ contains a V-QC for $\mathtt{v}$, designating the proposal of a leader block $\ell$,}
\aline[2]{with $\mathtt{notarised}=\bot$ and $\mathtt{nullified}=\mathrm{false}$: disseminate a vote for $\ell$ as in line 8;}
\aline[2]{Set $\mathtt{v}:=\mathtt{v}+1$, $\mathtt{nullified}:=\mathrm{false}$, $\mathtt{proposed}:=\mathrm{false}$,}
\alinecont[2]{$\mathtt{notarised}:=\bot$, $\mathtt{T}:=0$;}
\aline[3]{\acomment{Go to next view}}
\aline[1]{If $\mathtt{nullified}=\mathrm{false}$, $\mathtt{notarised}\neq\bot$ and $\State$ contains $\ge 2f+1$ messages, each signed}
\aline[1]{by a different processor, and each either (i) of the form $(\nullify,\mathtt{v})$ or $(\novote,\mathtt{v})$, or}
\aline[1]{(ii) a vote for some $\ell'\neq\mathtt{notarised}$ for view $\mathtt{v}$:}
\aline[2]{Set $\mathtt{nullified}:=\mathrm{true}$ and disseminate $(\nullify,\mathtt{v})$;}
\aline[3]{\acomment{Minimmit's mechanism ensuring view progression; see the walk-through}}
\aline[1]{If, for some leader block $\ell$, the pool $P\subseteq\State$ of votes for $\ell$,}
\alinecont[1]{by distinct processors, has $|P|\ge n-f$:}
\aline[2]{Finalise $\ell$ and set $\mathrm{log}_i$ to extend the transactions of $\Emit(P)$;}
\aline[3]{\acomment{Pool finalisation as defined in Section~\ref{sec:spec}; monotone, so re-run as $P$ grows}}
\end{algo}

\paragraph{Walk-through: lines 16--19 of Algorithm~\ref{alg:consensus}
ensure view progression} Section~\ref{sec:intuition} asked the reader
to take on trust that every view produces a V-notarisation or a
nullification; lines 16--19, inherited from Minimmit, are the delicate
part of the mechanism, and we now explain them. The concern is a
correct processor that has voted in view $\mathtt{v}$ but sees neither
a V-QC nor a nullification form. Perhaps the leader equivocated and
votes are split, or other processors timed out before the proposal
reached them. Waiting indefinitely is not an option, but nor may the
processor simply nullify, since a nullification must prove that
\emph{no} view $\mathtt{v}$ leader block receives an L-notarisation, and the
processor's own vote may be helping some block towards exactly that.
The rule resolves the tension: the processor nullifies only upon
receiving, from $2f+1$ distinct processors, messages none of which
support the block $\ell$ it voted for (nullify or novote messages, or
votes for other blocks). At least $f+1$ of these come from correct
processors and, considering the \emph{first} correct supporter of any
given block to nullify in this way, one can check that none of those
correct processors ever votes for $\ell$, capping its support at
$n-f-1$, so nullifying is safe. Conversely, the rule ensures
progression. After GST, every correct processor eventually holds a
view $\mathtt{v}$ message from every correct processor. If $2f+1$ votes for
a single leader block are among them, a V-notarisation exists, and otherwise
every correct voter eventually accumulates $2f+1$ messages of the
stated forms and nullifies, so that a nullification forms. Note
finally that the processor sends its nullify share \emph{without} a
novote message, since it did vote and a novote would be false. (Full
arguments appear in Section~\ref{sec:verification}.)
\section{Verification}\label{sec:verification}

\new{The guarantees of this section are for the single-epoch
protocol of Sections~\ref{sec:intuition} and~\ref{sec:spec};
Section~\ref{sec:payments} introduces epochs, and extends the
relevant guarantees across the close of an epoch.}

\subsection{Consistency}\label{sec:consistency}

Following Minimmit, we reason about \emph{notarisations}: sets of
messages actually sent, whether or not any processor assembles the
corresponding certificate. Precisely, we say a transaction block $b$
is \emph{certifiable} if at least $n-2f$ processors send DA-votes for
$b$; a leader block $\ell$ \emph{receives an L-notarisation} if at
least $n-f$ processors send votes for $\ell$; a \emph{V-notarisation}
for view $v$ is a set $W$ of view $v$ votes and novote($v$) messages,
one each from $n-f$ distinct processors, at least $2f+1$ of which are
votes for a common leader block $\ell$, whose proposal $W$ is said to
\emph{designate}; and view $v$ \emph{receives a nullification} if at
least $2f+1$ processors send nullify($v$) messages.\footnote{One caution
regarding designation: for $n=5f+1$ a V-notarisation designates
exactly one proposal (two disjoint sets of $2f+1$ votes would require
$4f+2>n-f$ messages), but for larger $n$ it may designate two. The
ambiguity is harmless. A V-QC always names a \emph{single} designated
proposal, that of the leader block it contains, and it is this
designation that parenthood (below) and the protocol's rules
consult. Wherever
the lemmas below extract values from a V-notarisation
(Lemmas~\ref{lem:safe-extend} and~\ref{lem:prefix}), an
L-notarisation for the view is in force, and
Lemma~\ref{lem:designation} then makes the designation unique.}  Any QC certifies a
notarisation of the corresponding kind (this uses the assumption that
signatures cannot be forged), and the functions $\Tips$, $\TipsF$,
$\Ord$ and $\Emit$ depend only on the underlying notarisation and the
designated proposal, so all statements below transfer to QCs.  Recall from
Section~\ref{sec:setup} that a message received from another processor
at timeslot $t$ was sent at a timeslot strictly before $t$.

We begin with the chain layer.

\begin{lemma}[Certified blocks are compatible]\label{lem:cert-compat}
Any two certifiable transaction blocks on the same chain are
compatible.
\end{lemma}

\begin{proof}
Towards a contradiction, suppose otherwise.  Among all pairs of incompatible certifiable blocks
(on some fixed chain), choose a pair $b,b'$, of heights $h\le h'$, so that $(h,h')$ is lexicographically minimal.  Each of $b,b'$ has DA-votes from $n-2f$ distinct processors. 
The two voter sets intersect in at least $2(n-2f)-n=n-4f\ge f+1$
processors, so some correct processor $p$ must DA-vote for both. If $h=h'$ this
contradicts condition (i) for DA-voting (one DA-vote per chain and height), so
$h<h'$.

Consider the timeslot at which $p$ sent its DA-vote for $b'$, and
let $b''$, of height $h''<h'$, be the block certified by the
\new{base certificate} at that vote (so that $h''$
is the greatest certified height below $h'$ in $p$'s local value $\State$ at that timeslot). By
condition (iii), $p$ then held blocks $b_{h''+1},\dots,b_{h'}=b'$
forming the path from $b''$ to $b'$.

First, suppose $h''\ge h$. The block $b''$ is certifiable, and the
pair $b,b''$ has both heights below $h'$. It follows from our choice
of $(h,h')$ that the blocks
$b$ and $b''$ are compatible. Then $b$
is an ancestor of $b''$, hence of $b'$, giving the required contradiction.

So $h''<h$, and we may set $z$ to be the block on the path from
$b''$ to $b'$ of height $h$. Since $b$ and $b'$ are incompatible, $b$ and
$z$ are incompatible; in particular $z\neq b$. By condition (iii),
$p$ sent a DA-vote for $z$ at or before the timeslot at which it
voted for $b'$. Since $p$ also DA-votes for $b$, this means $p$ DA-votes at height $h$ for
two distinct blocks, contradicting condition (i). 
\end{proof}

We next turn to the consensus layer. 

\begin{lemma}[One vote per view]\label{lem:one-vote}
A correct processor sends at most one vote for each view.
\end{lemma}

\begin{proof}
Votes are sent only at lines 7--8 and 12--13 of
Algorithm~\ref{alg:consensus}, both guarded by
$\mathtt{notarised}=\bot$; voting sets $\mathtt{notarised}$ to the
block voted for, and the variable is reset only upon entering a new
view.
\end{proof}

\begin{lemma}[Vote--novote exclusivity]\label{lem:exclusivity}
No correct processor sends both a vote for a view $v$ leader block and
a novote($v$) message.
\end{lemma}

\begin{proof}
Novote messages are sent only at line 10, guarded (at line 9) by
$\mathtt{notarised}=\bot$ and $\mathtt{nullified}=\mathrm{false}$. A
prior vote in the view sets $\mathtt{notarised}\neq\bot$, blocking
line 10. A novote sets $\mathtt{nullified}=\mathrm{true}$,
blocking both voting guards for the remainder of the view.
\end{proof}

\begin{lemma}[Designation]\label{lem:designation}
Suppose the view $v$ leader block $\ell$ receives an L-notarisation.
Then at most $2f$ processors ever send view $v$ votes for leader
blocks other than $\ell$, or send novote$(v)$ messages. In particular, every view $v$
V-notarisation designates $\ell$'s proposal, and no other view $v$
leader block receives $2f+1$ votes.
\end{lemma}

\begin{proof}
At least $n-2f$ of the $n-f$ processors voting for $\ell$ are correct
and, by Lemmas~\ref{lem:one-vote} and \ref{lem:exclusivity},  send no other view $v$ vote and no novote$(v)$ message. Any
view $v$ vote for another block, or any novote$(v)$ message, is therefore sent by one of the
remaining at most $2f$ processors. 
\end{proof}

\begin{lemma}[No nullification]\label{lem:no-null}
Suppose the view $v$ leader block $\ell$ receives an L-notarisation.
Then view $v$ receives no nullification.
\end{lemma}

\begin{proof}
Let $A$ be the set of correct processors that vote for $\ell$ and let $B$ be
the set of remaining correct processors. If $f_a\le f$ is the
number of faulty processors, $|A|\ge n-f-f_a$, so $|B|\le
(n-f_a)-(n-f-f_a)=f$.

We claim no member of $A$ ever sends nullify($v$). A member of $A$
has $\mathtt{notarised}\neq\bot$ for the remainder of the view, so
can send nullify($v$) only via lines 16--19, upon holding $2f+1$
messages by distinct processors, each a nullify($v$), a novote($v$),
or a view $v$ vote for a block other than $\ell$. Towards a contradiction, suppose there is some least timeslot $t$ at which some $p\in A$ does so.  Each of $p$'s $2f+1$ evidence messages was
sent strictly before $t$. Consider any correct sender among them: a
novote($v$) sender lies in $B$ by Lemma~\ref{lem:exclusivity}; the
sender of a vote for a block other than $\ell$ lies in $B$ by
Lemma~\ref{lem:one-vote}; and a nullify($v$) sender lies in $B$, or
else is a member of $A$ that sent nullify($v$) strictly before $t$,
contradicting the minimality of $t$. So the correct senders all lie
in $B$, and number at most $f$. With at most $f$ faulty senders, the
evidence totals at most $2f<2f+1$: a contradiction.

Correct nullify($v$) senders therefore all lie in $B$. With at most
$f$ faulty processors, fewer than $2f+1$ processors send nullify($v$)
messages, and view $v$ receives no nullification.
\end{proof}

Next, we establish consistency of the leader chain itself. Define the \emph{parent} of a leader block
$\ell=(v,H(Q),C)$ to be the leader block that $Q$ contains  (and define the parent of
any leader block referencing $Q_{\mathrm{gen}}$ to be $\lgenesis$).
 Ancestry and compatibility for leader blocks are then
defined as for transaction blocks.

\begin{lemma}[The leader chain]\label{lem:leader-chain}
Suppose the view $v$ leader block $\ell$ receives an L-notarisation.
Then every leader block receiving $2f+1$ votes in any view $v'\ge v$
has $\ell$ as an ancestor. In particular, leader blocks receiving
L-notarisations are pairwise compatible.
\end{lemma}

\begin{proof}
By induction on $v'$. For $v'=v$ the claim is
Lemma~\ref{lem:designation}. So let $v'>v$, suppose the claim holds
for all views in $[v,v')$, and let $\ell'=(v',H(Q''),C')$
receive $2f+1$ votes, at least $f+1$ of them from correct processors.
Let $p_j$ be a correct processor voting for $\ell'$ at the earliest
timeslot, $t_j$ say. If $p_j$ voted via lines 12--13, its $\State$
contained a V-QC \emph{for view $v'$ itself}, designating $\ell'$'s
proposal. (This rule fires only upon a V-QC for the
processor's current view; this is not the earlier-view V-QC $Q''$ that
$\ell'$ references, which belongs to the line 8 route below.) Such a
V-QC certifies $2f+1$ view $v'$ votes for $\ell'$, at least $f+1$ of
them correct, all sent strictly before $t_j$. This contradicts the
choice of $p_j$. So $p_j$ voted via lines 7--8, and its $\State$ contained a valid
proposal, in particular the V-QC $Q''$, for some view
$v''<v'$, together with a nullification for every view in the open
interval $(v'',v')$. A nullification for a view exists only if that view
receives one ($2f+1$ shares are needed to form the threshold
signature, and shares by correct processors cannot be forged), so by
Lemma~\ref{lem:no-null} we have $v\notin(v'',v')$, i.e.\ $v''\ge v$.
Now the tally of $Q''$ certifies $2f+1$ view $v''$ votes for
$\ell'$'s parent $\ell''$, which therefore receives $2f+1$ votes in
view $v''\in[v,v')$. By the induction hypothesis (or
by Lemma~\ref{lem:designation} if $v''=v$), $\ell$ is an ancestor of
$\ell''$, hence of $\ell'$.

For the final claim, if $\ell$ and $\ell'$ receive L-notarisations
in views $v\le v'$, then $\ell'$ in particular receives $2f+1$ votes,
so has $\ell$ as an ancestor.
\end{proof}

 We state the next lemma  for pools. Taking $P$ to be the votes of an L-QC
recovers the statement for $\TipsF$.

\begin{lemma}[Safe extension]\label{lem:safe-extend}
Let the leader block $\ell=(v,H(Q),C)$ receive an L-notarisation,
let $P$ be a set of votes for $\ell$ by at least
$n-f$ distinct processors, and let $k^*_i$, $b^*_i:=\TipsF(P)_i$ and settledness with respect to $P$ be
as in `Finalisation from the vote pool' (Section~\ref{sec:spec}). Let $W$ be any view $v$ V-notarisation.
Then, for every chain $i$:
\begin{enumerate}[label=(\alph*),itemsep=0pt]
  \item $\Tips(W)_i$ extends $b^*_i$, and;
  \item if chain $i$ is settled w.r.t.\ $P$, then $\Tips(W)_i=b^*_i$.
\end{enumerate}
\end{lemma}

\begin{proof}
By Lemma~\ref{lem:designation}, $W$ designates $\ell$'s proposal, so
$\Tips(W)$ is computed from the tally $V_i$ of $W$'s votes for
$\ell$, relative to the same chain proposal $C$. A correct
processor appearing in $W$ appears with its unique actual vote, in
$V_i$ if it voted for $\ell$ (in that case it cannot contribute a novote by
Lemma~\ref{lem:exclusivity}, nor a vote for another block by
Lemma~\ref{lem:one-vote}). Fix a chain $i$, and let $f_a\le f$ be the
number of faulty processors.

\emph{Positions.} By definition of $k^*_i$, at least $3f+1$ votes in
$P$, by distinct processors, have $\pi(i)\ge k^*_i$; at least $2f+1$
of these processors are correct, and at most $f$ processors have no
message in $W$, so at least $f+1$ votes in $V_i$ have $\pi(i)\ge
k^*_i$. The $(f+1)$-th largest position in $V_i$ is therefore at
least $k^*_i$. Writing $k_i$ for the position rank of $\Tips(W)_i$,
we have $k_i\ge k^*_i$.

\emph{Proof of (a).} If $b^*_i=\Prop(C,i,k^*_i)$, then the position
block $\Prop(C,i,k_i)$ extends $b^*_i$, and the extension-carry step
only ever replaces a block by one extending it, so $\Tips(W)_i$
extends $b^*_i$. Otherwise $k^*_i=m_i$ and at least $n-f$ votes in $P$,
by distinct processors, endorse $b^*_i$, which strictly extends
$\Prop(C,i,m_i)$. Then $k_i=m_i$ (positions are capped at $m_i$), so
the position block is $\Prop(C,i,m_i)$, an ancestor of $b^*_i$. At
least $n-f-f_a$ of the processors endorsing $b^*_i$ are correct, and
at least $n-f-f_a-f\ge n-3f\ge 2f+1$ of them appear in $V_i$,
endorsing $b^*_i$, so $b^*_i$ is a carry candidate. Moreover \emph{every} carry
candidate is compatible with $b^*_i$: if $b$ is incompatible with
$b^*_i$, then a correct processor endorsing $b$ has DA-voted a block
incompatible, at some height, with the one DA-voted by each of the
$\ge n-f-f_a$ correct processors endorsing $b^*_i$, so (correct
processors DA-voting once per height) at most $(n-f_a)-(n-f-f_a)=f$
correct processors endorse $b$. Therefore $b$ gains at most $f+f_a\le 2f$
votes in $V_i$, which is below the carry threshold. A candidate of greatest height has
height at least that of $b^*_i$ and is compatible with it, so extends it,
whichever candidate the tie-breaking rule selects. $\Tips(W)_i$ therefore 
extends $b^*_i$.

\emph{Proof of (b).} Suppose chain $i$ is settled, meaning  $k^*_i=m_i$ and
$\beta_i+(n-|P|)\le f$. Any correct processor whose vote endorses a
block strictly extending $b^*_i$ either has its vote in $P$ (there are
at most $\beta_i$ such) or has no vote in $P$ (there are at most
$n-|P|$ such processors), giving at most $f$ correct processors in all. A
carry candidate $b'$ strictly extending $b^*_i$ would need $2f+1$ votes
in $V_i$ endorsing $b'$, of which at most $f$ correct and at most
$f_a\le f$ faulty exist, so no such $b'$ exists. Candidates incompatible with $b^*_i$
are excluded as in (a), and positions are capped at
$m_i$, so every candidate, and the position block, is an ancestor of
or equal to $b^*_i$. By (a), $\Tips(W)_i$ extends $b^*_i$. Combining,
$\Tips(W)_i=b^*_i$.
\end{proof}

\begin{lemma}[Emission is a prefix]\label{lem:prefix}
In the setting of Lemma~\ref{lem:safe-extend},
$\Emit(P)\preceq\Ord(Q')$ for every V-QC $Q'$ for view $v$. In
particular, any two values $\Emit(P)$, $\Emit(P')$ extracted from
pools of votes for $\ell$ (by different processors, or at different
times) are compatible.
\end{lemma}

\begin{proof}
By Lemma~\ref{lem:designation}, the leader block contained in $Q'$ is
$\ell$, so $Q'$ references $Q$ and
$\Ord(Q')=\Ord(Q)^\frown\Horiz_C(\Tips(Q),\Tips(Q'))$, while
$\Emit(P)=\Ord(Q)^\frown\rho$: it suffices to show that 
$\rho\preceq\Horiz_C(\Tips(Q),\Tips(Q'))$. Note that
Lemma~\ref{lem:safe-extend} applies to $Q'$, taking $W$ to be the
V-notarisation that $Q'$ certifies. \new{Both sequences visit the
positions above $\Tips(Q)$ in the two-pass order relative to $C$,
the same order for both, the proposal being shared by
Lemma~\ref{lem:designation}. Both scans restart from offset $1$ at the
pass boundary, retaining $\Tips(Q)$ as their origins. Reaching the
second pass in $\Emit(P)$ requires $k_i^*=m_i$ on every chain, so the
whole proposal pass has been emitted.} Consider any
position visited by $\rho$ strictly before its halt (if any), on
chain $j$ say. If the position's height is at most that of $b^*_j$,
then $\rho$ appends the ancestor of $b^*_j$ at that height; by
Lemma~\ref{lem:safe-extend}(a) this is also the ancestor of
$\Tips(Q')_j$ at that height, which is what
$\Horiz_C(\Tips(Q),\Tips(Q'))$ appends. Otherwise the position
lies in the second pass, and $\rho$ passes it without halting
only if chain $j$ is settled w.r.t.\ $P$. By
Lemma~\ref{lem:safe-extend}(b), $\Tips(Q')_j=b^*_j$ lies below the
position, so both sequences append nothing. The two sequences thus
agree at every position up to $\rho$'s halt, where $\rho$ ends:
$\rho\preceq\Horiz_C(\Tips(Q),\Tips(Q'))$.

For the final claim, note that a V-notarisation for view $v$ exists:
any $n-f$ votes of the L-notarisation received by $\ell$ form one.
Both $\Emit(P)$ and $\Emit(P')$ are prefixes of $\Ord(Q')$ for a
V-QC $Q'$ certifying it.
\end{proof}

\begin{theorem}[Consistency]\label{thm:consistency}
If $p_i$ and $p_j$ are correct then, for any timeslots $t$ and $t'$,
$\mathrm{log}_i(t)$ and $\mathrm{log}_j(t')$ are compatible.
\end{theorem}

\begin{proof}
Every value taken by the log of a correct processor is the
transaction sequence of $\Emit(P)$ for some pool $P$ of $n-f$ or more
votes for a leader block $\ell$ (lines 21--22 of
Algorithm~\ref{alg:consensus}). In particular $\ell$ receives an
L-notarisation. Since prefixes of block sequences yield prefixes of
the corresponding deduplicated transaction sequences, it suffices to
show that any two such values $\Emit(P)$, $\Emit(P')$, for
leader blocks $\ell$ and $\ell'$ of views $v\le v'$, are
compatible as block sequences.

If $v=v'$ then $\ell'=\ell$ by Lemma~\ref{lem:designation},
and the claim is Lemma~\ref{lem:prefix}. So suppose $v<v'$. By
Lemma~\ref{lem:leader-chain}, $\ell$ is an ancestor of
$\ell'$. Let
$\ell=\ell_0,\ell_1,\dots,\ell_r=\ell'$ be the path, and for
each $s<r$ let $Q^{(s)}$ be the V-QC referenced by $\ell_{s+1}$, so
that $Q^{(s)}$ designates the proposal of $\ell_s$. Then:
\begin{itemize}[itemsep=0pt]
  \item $\Emit(P)\preceq\Ord(Q^{(0)})$, by Lemma~\ref{lem:prefix},
    since $Q^{(0)}$ is a V-QC for view $v$;
  \item $\Ord(Q^{(s)})\preceq\Ord(Q^{(s+1)})$ for each $s<r-1$, since
    $\ell_{s+1}$, the leader block contained in $Q^{(s+1)}$,
    references $Q^{(s)}$, so that $\Ord(Q^{(s+1)})$ extends
    $\Ord(Q^{(s)})$ by definition, and;
  \item $\Ord(Q^{(r-1)})\preceq\Emit(P')$, since
    $\ell'$ references $Q^{(r-1)}$ and
    $\Emit(P')=\Ord(Q^{(r-1)})^\frown\rho'$.
\end{itemize}
Chaining, $\Emit(P)\preceq\Emit(P')$.
\end{proof}

\subsection{Liveness}\label{sec:liveness}

Recall from Section~\ref{sec:spec} that $\State$ contains a
certificate (a V-QC, L-QC, nullification or DA-certificate) whenever
it contains messages from which one can be formed. We add one reading
convention: the instructions of Algorithms \ref{alg:chain} and
\ref{alg:consensus} are executed repeatedly within each timeslot,
until no instruction applies, so that a processor holding
view-advancing triggers for several consecutive views passes through
all of them in a single timeslot. Bounds below are stated in terms
of $\delta$, the (unknown) least upper bound on message delay after
GST (Section~\ref{sec:setup}), wherever possible: a message sent at
timeslot $t$ arrives by $\max\{\mathrm{GST},t\}+\delta$, and the
known bound $\Delta\ge\delta$ enters only through the $2\Delta$
timeouts.

\begin{lemma}[View synchronisation]\label{lem:sync}
If the first correct processor $p$ to enter view $v$ does so at timeslot
$t$, then every correct processor enters a view $\ge v$ by
$\max\{\mathrm{GST},t\}+\delta$.
\end{lemma}

\begin{proof}
Processor $p$ passed through every view $v'<v$, so its local value $\State$ at $t$
contains, for each such $v'$, a nullification or V-QC for $v'$. For
each such $v'$ it has therefore forwarded a nullification or V-QC for $v'$ to all processors at the timeslot $\le t$
at which one first became new to it (line 2 of
Algorithm~\ref{alg:consensus}). All of these messages arrive by
$\max\{\mathrm{GST},t\}+\delta$. A correct processor holding, for
every $v'<v$, a nullification or V-QC for $v'$ advances (lines 11 and
14), within the timeslot, to a view $\ge v$.
\end{proof}

\begin{lemma}[View exit]\label{lem:exit}
Suppose the first correct processor to enter view $v$ does so at
timeslot $t$, and let $T:=\max\{t,\mathrm{GST}\}$. Then every correct
processor enters a view $\ge v+1$ by $T+2\Delta+3\delta$.
\end{lemma}

\begin{proof}
Let $t'$ be the first timeslot at which a correct processor enters a
view $\ge v+1$ (if ever). If $t'\le T+2\Delta+2\delta$ then, by
Lemma~\ref{lem:sync}, we are done, since
$\max\{\mathrm{GST},t'\}+\delta\le T+2\Delta+3\delta$. So suppose no correct processor
enters a view $\ge v+1$ by $T+2\Delta+2\delta$. We show $t'\le
T+2\Delta+3\delta$.

By Lemma~\ref{lem:sync}, every correct processor enters view $v$ by
$T+\delta$, and, remaining in the view, sends exactly one of the
following by $T+2\Delta+\delta$: a vote (lines 7--8 or 12--13), or,
at $2\Delta$ on its local timer, a novote($v$) and nullify($v$) pair
(lines 9--10). By $T+2\Delta+2\delta$, therefore, every correct
processor holds a view $v$ vote or novote from every correct
processor, giving at least $n-f$ view $v$ messages by distinct
processors. There are two cases.

Suppose first that at least $2f+1$ correct processors vote for a
common leader block $\ell$. Then by $T+2\Delta+2\delta$ the $\State$
of every correct processor contains a V-QC for view $v$ (assembled
from $n-f$ of the correct messages), and every correct processor
enters view $v+1$ at $T+2\Delta+2\delta$ (lines 11 and 14), a
contradiction to our supposition.

Otherwise, every leader block receives at most $2f$ correct votes.
Consider any correct processor $p$ that voted, for $\ell_p$ say. Of
the $\ge n-f$ correct view $v$ messages that $p$ holds by
$T+2\Delta+2\delta$, those that are votes for $\ell_p$ number at most
$2f$, so at least $n-3f\ge 2f+1$, by distinct processors, are each a
nullify($v$), a novote($v$), or a vote for a leader block other than
$\ell_p$. At $T+2\Delta+2\delta$, $p$ then sends nullify($v$) via
lines 16--19, if it has not already done so. Every correct processor
has thus sent nullify($v$) by $T+2\Delta+2\delta$. The shares arrive
by $T+2\Delta+3\delta$, every correct processor assembles a
nullification, and enters view $v+1$.
\end{proof}

Since the first correct processor to enter view $v+1$ does so at or
after the first correct entry into view $v$, an induction gives:

\begin{corollary}\label{cor:progression}
Every view is eventually entered by every correct processor.
\end{corollary}

\begin{lemma}[Correct-leader views]\label{lem:leader-view}
Suppose $\lead(v)$ is correct and the first correct processor to
enter view $v$ does so at timeslot $t\ge\mathrm{GST}$. Then every
correct processor votes for the leader's block $\ell$ by $t+2\delta$,
so that $\ell$ receives an L-notarisation. By $t+3\delta$ every
correct processor holds a pool of at least $n-f$ votes for $\ell$, so
finalises $\ell$ (lines 21--22) and enters view $v+1$.
\end{lemma}

\begin{proof}
By Lemma~\ref{lem:sync}, every correct processor enters view $v$ by
$t+\delta$. Entry times thus lie in $[t,t+\delta]$, so no correct
processor's view $v$ timer reaches $2\Delta$ before $t+2\Delta$.

\emph{Every correct processor holds a valid proposal by $t+2\delta$.}
The leader enters view $v$ at some $t'\le t+\delta$ and immediately
proposes (lines 3--5), disseminating $\ell=(v,H(Q),C)$ for the V-QC
$Q$ returned by $\SelectAnchor$, of view $v'$ say, and $\ell$ arrives
everywhere by $t'+\delta\le t+2\delta$. Since $\SelectAnchor$
returns a V-QC for the \emph{greatest} view for which the leader
holds one, the leader advanced through each view $u\in(v',v)$ by
nullification, and holds a nullification for each. The V-QC $Q$ was forwarded to all by the
leader itself upon its becoming new, at some timeslot $\le t'$
($\SelectAnchor$ returns exactly the forwarded object), and so
arrives everywhere by $\max\{\mathrm{GST},t'\}+\delta\le
t+2\delta$; each
nullification became new to the leader at some timeslot $\le t'$
and was forwarded to all upon becoming new (line 2, by the leader or
by whoever first supplied it), so arrives everywhere by
$\max\{\mathrm{GST}, t'\}+\delta\le t+2\delta$ (threshold
signatures being unique, what was forwarded is the nullification). The leader being
correct, $\ell$ is the unique view $v$ leader block signed by
$\lead(v)$ and $C$ is well-formed with respect to $\Tips(Q)$:
conditions (i)--(iii) of proposal validity hold at every correct
processor by $t+2\delta$.

\emph{No correct processor sends a novote($v$) or nullify($v$)
message.} Timeouts (lines 9--10) fire no earlier than $t+2\Delta\ge
t+2\delta$, by
which time $\State$ contains the valid proposal, and lines 6--8
precede lines 9--10 within a timeslot, so a processor reaching $2\Delta$
on its timer votes rather than timing out (its guards hold, as we
check below). For nullify($v$) via lines 16--19: consider the first
correct processor to send nullify($v$). It has voted, necessarily for
$\ell$ (a vote via lines 7--8 requires a valid proposal, and $\ell$
is the unique candidate; a vote via lines 12--13 requires a view $v$
V-QC designating another block, hence $f+1$ correct votes for that
block, which do not exist, by this same argument applied inductively
to earlier voters). Its evidence must contain $2f+1$ messages, none
supporting $\ell$. Correct processors supply no novotes and no
earlier nullifies (it is first), and no votes for other blocks, so at
most $f$ faulty messages qualify, i.e., the evidence never accumulates.

\emph{Every correct processor votes for $\ell$ by $t+2\delta$.} At
the first timeslot ($\le t+2\delta$) at which it holds the valid
proposal while in view $v$, its guards hold:
$\mathtt{nullified}=\mathrm{false}$ since it sends no nullify($v$)
(above); $\mathtt{notarised}$ is $\bot$ unless it has already voted
for $\ell$ (the unique votable block, as above). It remains to check
that it is still in view $v$. It cannot have entered $v+1$ via a
nullification for $v$ ($f+1$ correct shares would be needed, and
none exist), and if via a V-QC for $v$, then lines 12--13 fire before
line 14, so it voted for $\ell$ first.

Thus at least $n-f$ correct processors vote for $\ell$ by
$t+2\delta$, providing an L-notarisation. The votes arrive at all correct processors by
$t+3\delta$, whereupon every correct processor finalises $\ell$ (lines 21--22)
and, holding a V-QC for view $v$ (any $n-f$ of the votes), enters
view $v+1$.
\end{proof}

Lemma~\ref{lem:leader-view} is the $3\delta$ figure of
Section~\ref{sec:intuition-raptr}. In a correct-leader view after
GST, the $2\Delta$ timeouts never bind, and every bound is governed
by the actual network delay.

\paragraph{Data availability} One convention remains before the
theorem. Recall from Section~\ref{sec:setup} that we solve
Extractable SMR: the protocol guarantees that the data of every block
entering the ordering is retrievable, with the retrieval mechanism
left to implementation. Note first that $\Emit(P)$ is computable
without block data: votes, chain proposals and headers determine the
emitted sequence of block \emph{identifiers}. We accordingly read
lines 21--22 modulo retrieval: the emitted identifiers are determined
at once, and the corresponding transactions enter the log as the data
of each emitted block is obtained. Retrieval terminates by
Lemma~\ref{lem:availability}: every emitted block and each of its
ancestors has at least one correct holder. Certified blocks and
newly finalised positions or extensions have at least $2f+1$
correct holders; a hash-anchor merely repeating a previous
safe-to-extend tip may retain only the weaker guarantee, which
still suffices for retrieval.
A correct processor
missing the data of an emitted block requests it from all processors,
and after GST a correct holder answers within $2\delta$. The same
convention covers the V-QCs of the reference chain, which computing
$\Ord$ requires and which a processor finalising via the pool need
not hold: a processor missing a V-QC named by hash requests it from
all. Possession is guaranteed, since any leader block on the
reference chain of a finalised leader block received $2f+1$ votes, and the
first $f+1$ correct processors to vote for it necessarily did so via
the valid-proposal route (a vote via lines 12--13 presupposes a view
V-QC containing $f+1$ earlier correct votes), each therefore holding
the referenced V-QC. In the common case no such retrieval occurs,
each view's referenced V-QC having been forwarded to all upon
first receipt, by the leader among others.

\begin{theorem}[Liveness]\label{thm:liveness}
If $p_i$ and $p_j$ are correct and $p_i$ receives the transaction
$\mathrm{tr}$ then, for some $t$,
$\mathrm{tr}\in\mathrm{log}_j(t)$.
\end{theorem}

\begin{proof}
We may suppose $\mathrm{tr}$ never enters $\mathrm{log}_j$, and
derive a contradiction. Throughout, we use the fact that all delays
after GST are bounded, so that `eventually' claims compose.

\emph{Step 1: $\mathrm{tr}$ enters $p_i$'s chain.} We first observe
that every block $p_i$ produces is eventually DA-certified, with the
certificate eventually held by every correct processor. By induction
on height $h$: once every correct processor holds the certificate for
$p_i$'s block at height $h-1$ together with $p_i$'s block $b_h$
(which $p_i$ disseminated), the block $b_h$ is eligible at every
correct processor: condition (i) holds since a correct processor's
DA-votes on chain $i$ are all for blocks of $p_i$'s unique chain
($p_i$ signs one block per height, and forgeries are excluded), so no
conflicting vote exists; (ii) holds with $h-h'\le 1$; and (iii)
holds along $p_i$'s chain, the votes below $b_h$ having been sent in
previous steps of the induction. Correct processors DA-vote every eligible
block (line 4 of Algorithm~\ref{alg:chain}), so $p_i$ receives $n-2f$
shares, forms the certificate, and disseminates it. Consequently the
window condition of $\mathsf{Ready}$ is satisfied within bounded time
whenever it fails, and, since $\mathrm{tr}\in\mathsf{pending}$
persistently (it is never included, else Steps 2--3 below apply to
the including block), the pacing condition ensures that $p_i$
eventually executes $\ProduceNext$, producing a block $b$ containing
$\mathrm{tr}$, at height $h_b$ say.

\emph{Step 2: $b$ is finalised.} Let $h^*\ge h_b$ be the height of
$p_i$'s tip at some fixed timeslot. Choose $t^*\ge\mathrm{GST}$
large enough that, by $t^*$, every correct processor has received and
DA-voted every block of chain $i$ up to height $h^*$, and holds every
DA-certificate for chain $i$ up to height $h^*-d$. (Such certificates
exist: by the window condition, $p_i$ produced its block at height
$h^*$ only with a certified block at height at least $h^*-d$, and
certificates are disseminated.) By Corollary~\ref{cor:progression}
and the round-robin $\lead$ function, some view $v$ with correct
leader has its first correct entry at a timeslot $t\ge t^*$. By
Lemma~\ref{lem:leader-view}, every correct processor votes for the
leader's block $\ell$ by $t+2\delta$ and finalises it by
$t+3\delta$.

Now consider the proposal's chain $i$ coordinate. \new{The leader
sets chain $i$'s base at $\Tips(Q)_i$ or at a DA-certificate of
greater height, whichever is higher.} It holds certificates for chain $i$ through height
$h^*-d$, so the base has height at least $h^*-d$. The entries
comprise every block the leader has DA-voted extending the base, up
to $d$ of them, and the leader has DA-voted chain $i$ through height
$h^*$. The proposed tip therefore has height at least
$\min\{(h^*-d)+d,\,h^*\}=h^*\ge h_b$. If the base itself has height
at least $h^*$, then the finalised tip extends the base, which
extends $b$ on the chain's single path, and we are done. Otherwise
every proposed entry at height at most $h^*$ lies on the single path
and was DA-voted by every correct processor by $t^*$, so every
correct voter reports a position of height at least $h^*$. (The
leader may also propose entries above $h^*$, which correct voters
need not hold; this only lowers reported positions towards the
height-$h^*$ entry, never below it.) The $(3f+1)$-th largest
position over the $n-f$ votes of any L-notarisation for $\ell$
therefore has height at least $h^*\ge h_b$, and the finalised tip
extends the block at that position. Since the blocks of chain $i$
form a single path, $b$ is an ancestor of the view $v$ finalised tip
on chain $i$.

\emph{Step 3: $b$ is emitted.} The value $\Emit(P)$ computed at view
$v$ may halt before reaching $b$'s slot, since some other chain may
be unsettled. Let $v'>v$ be the next view with correct leader, with
first correct entry at a timeslot $t'\ge t$
(Corollary~\ref{cor:progression} again). By
Lemma~\ref{lem:leader-view}, every correct processor finalises the
view $v'$ block $\ell_{v'}$ by $t'+3\delta$. In doing so it computes
$\Emit(P')=\Ord(Q^*)^\frown\rho'$, where $Q^*$ is the V-QC
referenced by $\ell_{v'}$. Now $\ell$ receives an L-notarisation, so
$\ell$ is an ancestor of $\ell_{v'}$ by Lemma~\ref{lem:leader-chain}.
The chain of references from $\ell_{v'}$ therefore passes through a
V-QC $Q_v$ for view $v$, and $\Ord(Q^*)\succeq\Ord(Q_v)$. By
Lemma~\ref{lem:safe-extend}, $\Tips(Q_v)_i$ extends the view $v$
finalised tip on chain $i$, hence extends $b$. Each segment
$\Horiz_C(T,T')$ of the recursion defining $\Ord$ includes every block
on $T'$'s path in the heights it covers, so $b$ appears in
$\Ord(Q_v)$, hence in $\Ord(Q^*)$, hence in $\Emit(P')$. Every correct processor therefore emits a sequence of
identifiers containing $b$ by $t'+3\delta$, plus bounded time for
any reference-chain V-QCs it must retrieve, and by the retrieval
convention it appends the corresponding transactions within bounded
further time.
(For $b$ itself no retrieval is needed, since every correct processor
holds $b$, having DA-voted it.) Either $\mathrm{tr}$ was already in
its log, as a removed duplicate, or it enters now. In either case
$\mathrm{tr}\in\mathrm{log}_j$,  a contradiction.
\end{proof}

\subsection{Remaining properties}

We collect the remaining claims made in earlier sections: the
availability lemma promised in Section~\ref{sec:liveness}, membership
finality, the block inclusion theorem, and the lower bound forcing the
extraction thresholds. Throughout, `held in full' means possession of
a block's complete data $(i,h,x,\Tr)$, as opposed to its header or
hash, and $f_a\le f$ denotes the number of actually faulty processors
in the execution.

\paragraph{Availability}

\begin{lemma}[Availability]\label{lem:availability}
\begin{enumerate}[label=(\alph*),itemsep=0pt]
  \item Every certifiable transaction block, and each of its
    ancestors, is held in full by at least $2f+1$ correct processors
    (the set of holders may vary from block to block).
  \item For every V-notarisation $W$ and chain $i$, the block
    $\Tips(W)_i$, and each of its ancestors, is held in full by at
    least one correct processor.
  \item Every block of an emitted sequence $\Emit(P)$ is held in
    full by at least one correct processor; and by at least $2f+1$
    if, on its chain, it lies at or below the finalised tip of some
    emission, where that tip certifies something new: a position
    block of rank $k^*\ge 1$, a DA-certified base, or an extension
    block endorsed by $n-f$ votes. (The excluded case is $k^*=0$ with a
    hash-anchor base: the finalised tip is then the anchor itself,
    merely re-affirming the previous safe-to-extend tip, whose
    endorsing votes DA-voted nothing; blocks covered by no other
    finalisation retain only the single-holder guarantee, which
    suffices for retrieval.)
\end{enumerate}
\end{lemma}

\begin{proof}
(a) Let $b$ be an ancestor of a certifiable block and, among the
certifiable blocks extending $b$ (or equal to it), let $b^*$ be one
of least height. At least $n-2f$ processors DA-voted $b^*$, so at
least $n-3f\ge 2f+1$ correct ones. By condition (iii), each possessed,
when it voted, the path from its anchor (a certifiable ancestor of
$b^*$ of smaller height) to $b^*$. No such anchor has height at least
that of $b$, since it would then be a certifiable block extending $b$
(anchors and $b$ are ancestors of $b^*$, hence comparable) of height
smaller than $b^*$'s, contradicting minimality. So each of the $2f+1$
correct voters held $b$.

(b) We argue by induction on views. If $\Tips(W)_i$ is a position
block $\Prop(C,i,k_i)$ with $k_i\ge 1$, then at least
$f+1$ votes in the tally report $\pi(i)\ge k_i$, so at least one
correct processor DA-voted positions $1,\dots,k_i$, and held, at each
vote, the corresponding block and the path down to its certifiable
anchor. If the extension carry advances beyond the base, at least $2f+1$ votes
endorse $\Tips(W)_i$, so at least $f+1$ correct processors did the same
along the carried path. In either case, ancestors below the relevant
anchors are covered by (a). If $k_i=0$ and the carry does not advance beyond the base, then
$\Tips(W)_i$ is the base itself, which well-formedness requires to
be either a DA-certificate, covered by (a), or a hash anchor equal to
$\Tips(Q')_i$ for the V-QC $Q'$ referenced by the proposal, which
is for an earlier view, and the induction hypothesis applies (the
base case being genesis, held by all).

(c) For the stronger claim, suppose the block lies at or below a
finalised tip $b^*_j$ of one of the three kinds hypothesised, and argue
relative to that emission's pool. If $b^*_j=\Prop(C,j,k^*_j)$ with
$k^*_j\ge 1$, then at least $3f+1$ votes report $\pi(j)\ge k^*_j$, so at
least $2f+1$ correct processors DA-voted positions $1,\dots,k^*_j$;
each held any given ancestor $b$ of $b^*_j$ directly, or else had a
certifiable anchor extending $b$, in which case (a) applies to $b$.
If $b^*_j$ is a DA-certified base, it is certifiable and (a) applies
directly. If $b^*_j$ is an extension block, at least $n-f$ votes endorse
it, so at least $n-2f\ge 2f+1$ correct processors DA-voted its path,
and the first dichotomy applies. The excluded case, $k^*_j=0$ over a
hash-anchor base, is exactly the one in which $b^*_j$ is the previous
safe-to-extend tip re-affirmed, a block that need not be certifiable,
so that (a) does not apply. For the weaker claim, note that every
block of $\Emit(P)$ lies, on its chain, at or below either a
finalised tip of the final segment or a coordinate of $\Tips$ for
some V-QC on the reference chain; a rank-$0$ hash-anchored finalised
tip equals the previous safe-to-extend tip, so in every case (b)
applies.
\end{proof}

So, the guarantee is as follows: quorum-scale redundancy ($2f+1$
correct holders) for everything at or below finalised tips that
certify new material, with a single guaranteed holder otherwise. A
single correct holder suffices for retrieval, but the redundancy gap is worth quantifying. It concerns,
per chain, only the blocks above the last certifiable block or
certified base, and its \emph{persistence} requires a faulty
producer, whether equivocating or merely withholding, say by
disclosing blocks to a single correct processor. On an honest chain
every block is eventually DA-voted by all correct processors
(Theorem~\ref{thm:liveness}, Step 1), hence certifiable, and (a)
restores $2f+1$ holders.

\paragraph{Exclusion at $3f+1$} \new{The following counting lemma
is used repeatedly below: endorsement by $3f+1$ votes of a single
L-notarisation already excludes incompatible certification, provided
the endorsed block lies strictly above the proposal's base for its
chain.\footnote{\new{The condition cannot be dropped: endorsement is
closed under ancestry, so every vote endorses the base and
everything below it without any voter having DA-voted those blocks,
and a carried base needs as little as one correct DA-voter on its
branch.}}}

\begin{lemma}[Exclusion]\label{lem:membership}
\new{Suppose $b$ lies strictly above the base of the chain proposal
of some L-notarisation $Q$ for $b$'s chain, and at least $3f+1$ of
$Q$'s votes endorse $b$. Then no block incompatible with $b$ is ever
certifiable.}
\end{lemma}

\begin{proof}
Suppose $b'$, incompatible with $b$, were certifiable, diverging
from $b$'s path at height $h_0$. By the least-height certifiable-block
argument of Lemma~\ref{lem:availability}(a), at least $n-2f-f_a$
correct processors DA-voted the block on $b'$'s path at $h_0$. At least $3f+1-f_a$
correct processors' votes endorse $b$. Consider any one of them,
$p_j$ say. Since $b$ lies strictly above the base, $p_j$'s endorsed tip
does too, so $p_j$ has genuinely DA-voted its endorsed path, from its
certified \new{base} upward, and $b$ lies on that path. If $p_j$'s \new{base}
has height at least $h_0$, \new{the base} is a certifiable block whose
path agrees with $b$'s at $h_0$, so is incompatible with $b'$, and
Lemma~\ref{lem:cert-compat} yields a contradiction with $b'$'s
certifiability directly. Otherwise $p_j$ DA-voted $b$'s path at
$h_0$. If the latter holds for all such $p_j$, then, correct
processors DA-voting once per height, the two sets are disjoint, and
$(3f+1-f_a)+(n-2f-f_a)\le n-f_a$, i.e.\ $f+1\le f_a$, a
contradiction.
\end{proof}

\paragraph{Block inclusion} We now spell out the censorship
resistance obtained, which is a pay-off of extension voting. The
key point to watch is which quantifiers depend on the leader: none
do.

\begin{theorem}[Block inclusion]\label{thm:inclusion}
Let $b$ be a block on chain $i$ produced by the correct processor
$p_i$, and let $v$ be a view in which some leader block
$\ell=(v,H(Q),C)$ receives an L-notarisation. Suppose that:
\begin{enumerate}[label=(\roman*),itemsep=0pt]
  \item every correct processor DA-voted $b$, and each ancestor of
    $b$ lying strictly above $\Tips(Q)_i$, before sending its view
    $v$ vote, and;
  \item $b$'s height exceeds that of $\Tips(Q)_i$ by at most $e$.
\end{enumerate}
Then:
\begin{enumerate}[label=(\alph*),itemsep=0pt]
  \item the vote of every correct processor that votes for $\ell$
    endorses $b$;
  \item at least $n-2f\ge 3f+1$ votes of \emph{every} L-notarisation
    for view $v$ endorse $b$; \new{in particular, no block incompatible with $b$ is
    ever certifiable, by Lemma~\ref{lem:membership} whenever $b$
    lies strictly above the proposal's base, and by
    Lemma~\ref{lem:cert-compat} when instead a DA-certified base
    lies at or above $b$};
  \item $b$ is an ancestor of (or equal to) $\Tips(W)_i$ for
    \emph{every} view $v$ V-notarisation $W$, and;
  \item $b$ appears in $\Emit(P)$ for every pool $P$ of $n-f$ or more
    votes for a leader block of any later view that receives an
    L-notarisation: $b$ is ordered by the next finalised leader
    block.
\end{enumerate}
\end{theorem}

\begin{proof}
Since $p_i$ is correct and signatures cannot be forged, the blocks of
chain $i$ form a single path, one block per height.

(a) Let $p_j$ be a correct processor voting for $\ell$, with vote
$(\pi,\varepsilon)$, and consider the base for chain $i$ in $C$. By
proposal validity (checked by the first correct voter, and inherited
along the vote-along rule as in Lemma~\ref{lem:leader-chain}), the
base extends or equals $\Tips(Q)_i$. If $b$ lies at or below the
base, then $b$ is an ancestor of every block $p_j$'s vote endorses on
chain $i$, and we are done; likewise if the position block
$\Prop(C,i,\pi(i))$ lies at or above $b$ on the chain's single
path, since the endorsed tip extends the position block. Otherwise
$b$ lies above the base, and by (ii) at most $e$ above $\Tips(Q)_i$,
hence at most $e$ above the base (which extends $\Tips(Q)_i$). The
position block is a real block of the chain, DA-voted by $p_j$ when
$\pi(i)\ge 1$ and equal to the base otherwise, and it is an
ancestor of $b$, since the real entries lie on the chain's single
path below $b$. By hypothesis (i), $p_j$ has DA-voted every block
strictly between $\Tips(Q)_i$ and $b$, so its DA-voted blocks run
consecutively from the position block to $b$, and $b$ lies within
$e$ of the position block. $\Extensions(\State,\ell)$ therefore
carries $p_j$'s endorsed tip on chain $i$ to $b$ or beyond, so $p_j$'s
vote endorses $b$.

(b) An L-notarisation contains votes from $n-f$ distinct processors,
at least $n-2f$ of them correct. By (a), each such vote endorses
$b$, and $n-2f\ge 3f+1$.

(c) Let $W$ be a view $v$ V-notarisation with tally $V_i$. Correct
processors appearing in $W$ appear with their actual votes. Since
$\ell$ receives an L-notarisation, at least $n-2f$ correct processors
voted for $\ell$, of which at most $f$ are absent from $W$, so at
least $n-3f\ge 2f+1$ votes in $V_i$ endorse $b$. If $b$ lies at or below the
position block $\Prop(C,i,k_i)$, then $b\preceq\Tips(W)_i$
immediately. Otherwise $b$ extends the position block (single path)
and has $2f+1$ endorsing votes in $V_i$, so $b$ is a carry candidate
and the carry applies. Candidates lie on the single path, and a
candidate of greatest height has height at least $b$'s, so extends $b$:
$\Tips(W)_i$ extends $b$.

(d) Let $\ell'$, of view $v'>v$, receive an L-notarisation, and let
$P$ be a pool of $n-f$ or more votes for it. By
Lemma~\ref{lem:leader-chain}, $\ell$ is an ancestor of $\ell'$, so
the chain of references from $\ell'$ passes through a V-QC $Q_v$
for view $v$, and $\Emit(P)$ extends $\Ord(Q_v)$. By (c),
$\Tips(Q_v)_i$ extends $b$. Finally, $b$ appears in $\Ord(Q_v)$:
along $Q_v$'s own reference chain the coordinates $\Tips(\cdot)_i$
are non-decreasing in height (each proposal's base extends or equals
the previous safe-to-extend tip) and all lie on chain $i$'s single
path, so the sweep segments partition the heights below
$\Tips(Q_v)_i$, and $b$ is appended in the segment covering its
height.
\end{proof}

\new{When the leader also proposes honestly for chain $i$, more
holds: if $b$ lies at or below the proposed tip then, by clause (b)
and the position rank, the finalised tips on chain $i$ extracted
from every pool of $n-f$ or more view $v$ votes extend $b$; and
when $b$ instead extends the proposed tip, the pool of all correct
votes extends $b$, every correct vote endorsing $b$ by clause (a).}

\paragraph{What the theorem gives} Hypotheses (i) and (ii) mention
only dissemination, and neither mentions the leader, while the
conclusions quantify over \emph{every} L-notarisation, pool and
V-notarisation of the view. The consequence is best put as a
dichotomy: \new{either view $v$ produces no L-notarisation at all,
or at least $n-2f$ votes of every pool of the view endorse $b$, no
competing block is ever certifiable, and $b$ is placed in the total
ordering by the next finalised leader block.} A Byzantine
leader cannot have its block finalised \emph{and} exclude a
well-circulated honest block. The only censorship available is to
sacrifice the view (producing nothing, harming every chain
equally, and forfeiting a leader slot). The guarantee applies again
in any later view satisfying hypotheses (i) and (ii). Contrast Raptr
(Section~\ref{sec:intuition-raptr}), where a single faulty
\emph{producer} voids the optimistic path for the whole block. What
remains leader-sensitive is only the in-view ordering of targeted
chains (deferral to a later finalised leader block,
Section~\ref{sec:extensions}, available to the leader alone,
at the price of attributability).

\paragraph{Accountability} \new{Misbehaviour in Multimmit falls
into three classes. \emph{Provable} misbehaviour is a
constant-size pair of signed objects that no correct processor can
produce: two blocks signed by one producer at one height (or a
signed parent-hash contradiction one height apart), two leader
blocks of one leader for one view, a vote beside a novote (or a
second vote) of one processor for one view, or two DA-votes of one
processor for one chain and height. Junk proposal entries are
merely \emph{attributable}: a challenge to exhibit the signed
header behind an entry is always answerable by a correct leader and
never answerable for junk, though a missing answer is not proof in
the offline sense, since under asynchrony a correct leader's answer
may simply be late. Omission is \emph{not even attributable},
being indistinguishable from slowness, which is why
Section~\ref{sec:achieves} treats omission separately, its damage
capped by extension votes. We do not formalise the responses
(slashing, reputation) that each class of evidence can support.}

\subsection{One-round prefix voting}\label{sec:oneround}
\begingroup
\color{black}
\hypersetup{linkcolor=black,citecolor=black,urlcolor=black}
The purpose of this subsection is to explain why strengthening Raptr's
resilience assumption does not, by itself, allow its second voting round
to be removed. We first give a simple setup for the extraction of
prefixes from votes. We then consider the deletion of Raptr's second
round directly, using its own view-change rule. The claims concern these
specified rules, not all protocols that might use prefix votes.

\paragraph{The setup: threshold prefix rules}
Fix $f\ge 1$ and $n\ge 3f+1$ (including the case $n\geq 5f+1$), with at most $f$ faulty processors.
A proposal specifies a finite sequence of transaction blocks by listing
their hashes:
\[
    C=(H(b_1),\dots,H(b_m)).
\]
A processor can receive $C$ without receiving the blocks. It checks a
block supplied at position $j$ by comparing its hash with $C$'s $j$th
entry. Write $h_C:=H(C)$ and call $(b_1,\dots,b_k)$ the prefix of length
$k$, the prefix of length $0$ being empty. We consider one fixed
proposal $C$ in one view $v$, granting agreement on the proposal so that
only the prefix length remains to be selected. The proposal carries no
prior certificate of its blocks' availability.

A vote is a signed tuple $(v,h_C,k)$ with $0\le k\le m$. A correct
processor sends at most one vote, reporting the length of the longest
prefix it holds when it votes. It does not wait for other processors'
votes, and retains the blocks it voted for.
Faulty processors may sign different positions for different recipients.
The rules below must work for every distribution of blocks and votes
consistent with these requirements. In particular, a faulty producer
may disclose a block to only some correct processors.

There are two integer thresholds, $\mu_{\mathrm{fin}},\mu_{\mathrm{rec}}\in\{1,\dots,n-f\}$.
From a set $P$ of $n-f$ votes for $C$, one per processor, the
\emph{finalisation rule} selects the greatest $k$ for which at least
$\mu_{\mathrm{fin}}$ votes report positions $\ge k$. From any such set $W$, the
\emph{recovery rule} selects the greatest $k$ for which at least $\mu_{\mathrm{rec}}$
votes report positions $\ge k$. A selected length is $0$ if no positive
length meets its threshold. The recovery rule specifies how much of
$C$ a later proposal must retain. It uses $C$ and $W$ alone, not a
certificate assembled from additional votes or the result of a further
exchange of messages. We call this pair of rules \emph{threshold prefix
rules}.

The sets $P$ and $W$ may contain different processors. A correct
processor occurring in both contributes the same vote; a faulty one
need not. The pair is \emph{safe} if, for every permitted collection of
signed votes, (i) every prefix selected by either rule is held by at
least one correct processor, and (ii) every prefix selected from $W$
extends every prefix finalised from $P$. Condition (ii) applies even
when no processor holding $W$ has seen $P$.

\begin{theorem}[Prefix thresholds]\label{thm:position}
Threshold prefix rules are safe if and only if
\[
    \mu_{\mathrm{rec}}\ge f+1 \qquad\text{and}\qquad \mu_{\mathrm{fin}}\ge \mu_{\mathrm{rec}}+2f.
\]
In particular, $\mu_{\mathrm{fin}}\ge 3f+1$ is necessary, and $\mu_{\mathrm{fin}}=3f+1$, $\mu_{\mathrm{rec}}=f+1$
suffice whenever $n-f\ge 3f+1$.
\end{theorem}

\begin{proof}
If $\mu_{\mathrm{rec}}\le f$, let only faulty processors hold $b_1$, and let $\mu_{\mathrm{rec}}$ of
them report position $1$ in $W$, all other reports being $0$. Recovery
selects an unavailable block. Thus $\mu_{\mathrm{rec}}\ge f+1$ is necessary.

For the second inequality, let a set $P$ contain exactly $\mu_{\mathrm{fin}}$ votes at
position $1$, the rest at $0$, with as many of the $\mu_{\mathrm{fin}}$ supporters faulty
as possible. There are $\max\{\mu_{\mathrm{fin}}-f,0\}$ correct supporters. Form $W$
by omitting as many as possible of these correct supporters, up to $f$,
and have all faulty processors report $0$ in $W$. All other correct
processors report $0$. The number of position-$1$ votes remaining in
$W$ is
\[
    \max\{\mu_{\mathrm{fin}}-2f,0\}.
\]
Finalisation from $P$ selects length $1$. If $\mu_{\mathrm{fin}}<\mu_{\mathrm{rec}}+2f$, recovery from
$W$ selects $0$, contradicting (ii). These vote sets are permitted:
deliver $b_1$ to precisely the correct supporters before they vote,
and delay the other deliveries until after the two sets are formed.

Conversely, suppose the inequalities hold. Each selected prefix has
more than $f$ supporting votes, so a correct supporter holds it. If a
prefix is finalised from $P$, at least $\mu_{\mathrm{fin}}-f$ correct processors voted
for it. At most $f$ are absent from $W$, leaving at least
$\mu_{\mathrm{fin}}-2f\ge \mu_{\mathrm{rec}}$ votes for that prefix in $W$. Recovery therefore selects
at least its length.
\end{proof}

\begin{corollary}[Well-circulated prefixes]\label{cor:slack}
Within this setup, safe threshold prefix rules can finalise a prefix
from every set of $n-f$ votes whenever every correct processor holds
it before voting, despite faulty processors reporting $0$, if and only
if $n\ge 5f+1$.
\end{corollary}

\begin{proof}
A set of $n-f$ votes can omit $f$ correct processors and include $f$
faulty ones reporting $0$, leaving only $n-2f$ supporting votes.
The required property is therefore equivalent to $\mu_{\mathrm{fin}}\le n-2f$.
By Theorem~\ref{thm:position}, a safe pair with this property exists
exactly when $3f+1\le n-2f$, or $n\ge 5f+1$.
\end{proof}

\paragraph{Raptr's second round}
Raptr~\cite[Sections 5.1--5.2]{tonkikh2025raptr} uses two kinds of
votes. A quorum certificate (QC) contains
$q:=\lceil(n+f+1)/2\rceil$ first-round votes for one proposal and
certifies the greatest prefix supported by at least $S$ of those
votes, where $f+1\le S\le q$. Processors then send second-round
votes carrying QCs, and a quorum of these votes finalises the minimum
certified prefix they report. On a timeout, a processor reports the
highest QC it has received, and a quorum of timeout messages permits a
new proposal extending the highest QC among those reports.

At $S=f+1$, a certified prefix can have only one correct holder.
Theorem~\ref{thm:position} shows why that availability threshold cannot
also be the finalisation threshold of the threshold prefix rules above,
however large $n$ is. Raptr's recovery reports carry QCs, not individual
prefix votes, so the following direct argument separately addresses
removing its second round while retaining its existing recovery rule.

\begin{proposition}[Deleting the second round]\label{prop:raptr-one-round}
Fix $f\ge 1$ and $n\ge 3f+1$  (including the case $n\geq 5f+1$). Modify Raptr by deleting its
second-round votes and finalising a QC's certified prefix as soon as
the QC is received, leaving its first-round voting and timeout rules
unchanged. The resulting protocol is inconsistent. This holds for
every $S\in\{f+1,\dots,q\}$, and also if its quorum size is
increased to $q=n-f$.
\end{proposition}

\begin{proof}
We work before GST, from a common finalised prefix. Choose a view whose
leader is faulty and whose next leader is correct. The faulty leader
proposes new data beginning with a transaction $\mathrm{tr}$. Deliver all the
proposed data to $q-f$ correct processors, including $p$, and have
these processors and the $f$ faulty processors vote for the full
proposal. Choose the next leader outside these correct recipients;
this is possible because $q\le n-f<n$. Deliver the $q$ votes to $p$,
which forms a full-prefix QC and finalises a sequence beginning with
$\mathrm{tr}$. This QC certifies the full prefix for every permitted $S$.

Delay this QC, and sufficiently many of the individual votes, to all
other correct processors. They can time out without having received
or assembled a QC for the view. Since $q\le n-f\le n-1$, a quorum
of timeout messages can omit $p$. Let every message in this quorum
report the common earlier QC, the faulty processors concealing their
knowledge of the new one.

The next leader can therefore propose from the earlier prefix. It has
not received $\mathrm{tr}$'s data; give it a different transaction $\mathrm{tr}'$ to propose
first instead. The timeout certificate justifies this proposal under
Raptr's unchanged rule. Having voted in the preceding view does not
prevent a correct processor from voting for it: the rule requires
extension of the QCs in the timeout certificate, not of each voter's
uncertified earlier vote. Deliver the new data and collect $q$ votes
from processors other than $p$. A correct recipient of this second
full-prefix QC finalises a sequence beginning with $\mathrm{tr}'$, incompatible
with $p$'s sequence. All the delayed messages can arrive after these
two finalisations, so the execution satisfies partial synchrony.
\end{proof}

The second round prevents this execution by distributing knowledge of
the certified prefix to a quorum before finality. Increasing $n$
relative to $f$ does not give a hidden first-round certificate that
property. A one-round replacement must change the finalisation or
recovery rules, not merely the resilience assumption. Multimmit makes
such a change: for positions fixed by a leader proposal, it uses
$3f+1$ endorsements for finalisation and $f+1$ in the V-QC extraction,
with leader-block finality establishing agreement on the proposal.
Extensions, which are not fixed by that proposal, require the separate
analysis of Section~\ref{sec:exact-threshold}.
\endgroup

\subsection{The exact extension threshold}\label{sec:exact-threshold}
Finally, we justify the
claim of Section~\ref{sec:spec} that the extraction rules cannot be
relaxed at $n=5f+1$, and determine exactly what relaxation becomes
available for larger $n$. \new{Consider the protocol of
Section~\ref{sec:spec}, with one number exposed as a parameter:
an extension block is finalised under the conditions of
Section~\ref{sec:spec}, with $s$ endorsing votes in an
L-notarisation in place of $n-f$; and an extraction rule assigns
safe-to-extend tips to each V-notarisation, as a function of the
V-notarisation alone. Say the pair is \emph{safe} if, in every
execution, the tips assigned to every view $v$ V-notarisation
extend every extension block so finalised in view $v$. This is
the contract that Lemma~\ref{lem:safe-extend} discharges for
the pair of Section~\ref{sec:spec}, and on which consistency
rests.}
We assume $e\ge1$.
Throughout, the producer is one of the $n$ processors, so its
equivocation consumes one of the $f$ faults, and $f\ge 1$: both
are used in the constructions, the producer's own accounted vote
carrying the tie-break, and the theorem is stated for this model.
Set $s^*:=\lceil (n+3f)/2\rceil$.

\begin{theorem}[Exact extension threshold]\label{thm:exact}
For every $n\ge 5f+1$:
\begin{enumerate}[label=(\alph*),itemsep=0pt]
  \item no safe pair exists with $s<s^*$, and;
  \item a safe pair exists with $s=s^*$.
\end{enumerate}
At $n=5f+1$ we have $s^*=n-f$, so the rules of
Section~\ref{sec:spec} are optimal there.
\end{theorem}

\begin{proof}
(a) We construct two executions sharing a view $v$ V-notarisation
$W$ bit for bit, in each of which some extension block gathers
$s^*-1$ endorsing votes in an L-notarisation, the two blocks
incompatible. No extraction rule is then safe in both executions,
since identical bytes force identical tips, and no block extends
both. The producer $p_i$ of the chain in question is faulty in both
executions. It signs incompatible extension blocks $b$ and $b'$, at
the same height above the proposed tip, showing $b$ to some
processors and $b'$ to others. $W$ consists of $n-f$ votes for the
view's leader block: $x:=\lceil(n-f)/2\rceil$ endorsing $b$, with
$p_i$'s own vote among them, and $y:=\lfloor(n-f)/2\rfloor$ endorsing
$b'$. The $f$ processors unaccounted in $W$ are correct in both
executions. In the first execution, the faulty processors are $p_i$
together with $f-1$ of the $b'$-side voters (there is room, since
$y\ge 2f$). The producer shows $b$ to the unaccounted processors,
who endorse it in their votes, and each faulty $b'$-side voter signs
a second vote endorsing $b$. The votes endorsing $b$ then come
from $x+f+(f-1)=x+2f-1=s^*-1$ distinct processors, and an
L-notarisation containing them all is assembled by filling the
quorum with the remaining votes of $W$. The second execution is the
mirror image, with faulty set $p_i$ together with $f-1$ of the
$b$-side voters, and one addition: $p_i$'s accounted vote endorses
$b$, so its \emph{second} signature, endorsing $b'$, now contributes
a distinct processor, and $b'$ gathers $y+f+(f-1)+1=y+2f\ge s^*-1$
endorsing votes. Both executions are otherwise unremarkable, and $W$
is a valid V-notarisation in each. Note the asymmetry: the
producer's accounted vote sits on one side, and whichever side holds
it is capped at $2f-1$ additions, the other at $2f$; this single
pinned vote is what separates $s^*-1$ from $s^*$. (For producers
external to the processor population the accounting differs: the
producer's equivocation then consumes no processor fault, the
pinned-vote asymmetry is unavailable, and the strict-majority
form $\lfloor(n+3f)/2\rfloor+1$, one larger in the even case, is
the natural threshold; we do not develop that model. At the
operating point $n=5f+1$ the even case does not arise.)

(b) Take $s=s^*$, and extract as in Section~\ref{sec:spec} except
that the extension carry runs at the threshold $s^*-2f$ in place of
$2f+1$, with one further rule: when two incompatible blocks each
gather exactly $s^*-2f$ accounted endorsing votes, which we show
below is possible only when $n+3f$ is even, discard the one that the
producer's own accounted vote endorses\new{; when the producer
is unaccounted in $W$, the V-notarisation, discard both. The
fallback forgoes an extension and arises in no execution with a
finalised extension block, so it supplies totality, never
safety.} First note that ties are
the only ambiguity. Votes endorse a single path, so the endorsement counts of
pairwise incompatible blocks sum to at most $n-f$, while two blocks
meeting the threshold account for $2(s^*-2f)\ge n-f$, with equality
only when $n+3f$ is even; three are excluded outright, and in the
odd case so are two.

Now suppose the extension block $b$ is finalised, with $s^*$
endorsing votes in some L-notarisation. Let $W$ be any view $v$
V-notarisation, in an execution with $f_a\le f$ faulty processors.
At most $f_a$ of the $s^*$ endorsers are faulty, so at least
$s^*-f_a$ are correct, and each endorses $b$'s path in its single
vote. At most $f$ of these are unaccounted in $W$. At least
$s^*-2f$ votes accounted in $W$ therefore endorse $b$, and the
carry reaches $b$.

Now let $b'$ be any block incompatible with $b$. A correct voter
endorses at most one of the two, so at most
$(n-f_a)-(s^*-f_a)=n-s^*$ correct voters endorse $b'$, and $b'$'s
accounted endorsement is at most $n-s^*+f_a\le s^*-2f$. If $n+3f$ is
odd, the inequality is strict, $b$ dominates every rival, and we
are done. So suppose $n+3f$ is even and $b'$ ties with $b$ at
$s^*-2f=(n-f)/2$. The tie forces every inequality above to be an
equality. So $f_a=f$, every faulty processor is accounted in $W$
with a vote endorsing $b'$, the $f$ unaccounted processors are
correct endorsers of $b$, and the $(n-f)/2$ accounted votes
endorsing $b$ all come from correct processors. Each branch now
carries at least $f+1$ accounted correct votes ($b'$ carries
$(n-f)/2-f=(n-3f)/2\ge f+1$, since $n+3f$ even gives $n\ne 5f+1$,
so $n\ge 5f+2$). Blocks on both branches were therefore genuinely
signed, so the producer equivocated and is faulty. The producer is
accounted, the unaccounted processors being correct, so its vote is
one of the $f$ faulty accounted votes, and every one of those
endorses $b'$. The override therefore discards $b'$ and carries
$b$, as required. Finally, on an honest producer's chain
incompatible blocks do not exist, so no tie can arise and the
override never fires.
\end{proof}

Three remarks. First, at $n=5f+1$ the theorem's thresholds are
$s^*=n-f$ and $s^*-2f=2f+1$, no tie is possible ($n+3f$ being odd),
and the safe pair of part (b) is precisely the pair of rules of
Section~\ref{sec:spec}. Second, the theorem concerns the extension
carry alone. The companion observation that, even at $s=n-f$, the
condition $k^*_i=m_i$ cannot be dropped \new{(an equivocating
producer arranges one execution finalising an extension block on
one branch and another finalising a proposed entry on the other,
with a single V-QC common to both, so that no extraction rule
applied to that V-QC is safe in both)}, is unaffected, as are the downstream
constants: \new{the exclusion threshold of Lemma~\ref{lem:membership}} and the counts of the block
inclusion theorem do not depend on $s$. Third, we do not install the
relaxation. At the headline assumption $n=5f+1$ it changes nothing,
and the rules of Section~\ref{sec:spec} stand as written; a
deployment with $n$ well above $5f+1$ may adopt the generalised
carry as a contained change. At $n=100$, $f=19$, it buys an in-view
ordering of the frontier that survives two omissions from the
quorum's $81$ votes, where unanimity tolerated none.

\subsection{Communication complexity}\label{sec:complexity}

At a fixed rate of transaction-block production, consensus-object
sizes need not grow with the amount of transaction data. The measured
certificate sizes do grow when more blocks are produced per view
(Section~\ref{sec:measurement-traffic}). Transaction data travels
separately on the chain layer. Consensus objects carry hashes and positions, at the
sizes of Section~\ref{sec:sizes}. So long as bandwidth handles the
consensus traffic with room to spare, the remainder of each link
carries payload. The operative question is therefore not asymptotic
but concrete: Are consensus objects large enough to intrude at the
deployment scales one expects in practice? At $n=50$ or $100$ they
are not, per the numbers of Section~\ref{sec:sizes}. At much larger
$n$, correspondingly more bandwidth must be provisioned.  

Here, we use Section~\ref{sec:spec}'s threshold-certificate format
and describe the asymptotics per view, taking hash and signature
sizes to be constant, along with $d$ and $e$. Message
complexity is $\Theta(n^2)$. Votes travel all-to-all, which is the
price of finalisation everywhere within $2\delta$ of the proposal. 
 In bits, a vote is $O(K)$. The vote
exchange therefore costs $\Theta(Kn^2)$ in total, and $\Theta(Kn)$
per processor per view in each direction, with proposals adding
$O(Kn)$. The one term that can exceed this is the dissemination of
V-QCs. The forwarding duty of Section~\ref{sec:spec}, taken
literally, sends up to $n^2$ copies of objects that degrade to
$\Theta(nK)$ (Section~\ref{sec:sizes}). This means $\Theta(n^3K)$
bits per view, and $\Theta(n^4)$ when $K=\Theta(n)$. The worst
case is an artifact of unconditional forwarding, however.   Votes and novotes already
travel all-to-all, so in the common case every processor assembles
each V-QC locally, and forwarded bodies are redundant. Simple
modifications \new{--- together, the \emph{dissemination
variant} ---} therefore suffice to remove the all-to-all
forwarding of V-QCs:\footnote{One might ask why the
specification of Section~\ref{sec:spec} does not simply adopt
these modifications. The reason is that the all-to-all
forwarding costs bandwidth rather than latency. As noted with
the forwarding rule of Section~\ref{sec:spec}, nothing in the
protocol waits on the redundant sends, which lie off the
critical path, and at the scales of Section~\ref{sec:sizes} the
bandwidth involved is small. Keeping the unconditional duty
keeps the specification and its analysis simpler.}
\begin{itemize}[itemsep=1pt]
  \item The forwarding duty is retained for nullifications and
    dropped for V-QCs, whose bodies then travel only with a
    correct leader's proposal, once per view.
  \item View progression runs on \emph{support certificates}: a
    constant-per-chain exit object assembled from threshold
    shares that travel only with their senders' own full votes,
    so a certificate for a view attests $f+1$ correct votes
    behind it; views are traversed strictly sequentially, never
    skipped, and a processor exiting on a certificate casts its
    own full vote for the certified block first, the block
    itself, of size $O(K)$, travelling inside the certificate.
    Departed views are resolved in the background into a genuine
    V-QC, assembled locally from the broadcast votes, or a
    nullification.
  \item Referenced V-QCs are recovered exactly, without any
    request--response protocol, from information already in
    flight: every correct signer's record of a V-QC was broadcast
    when cast, so at most $f$ records can be missing anywhere,
    and the processors that checked the proposal against the full
    V-QC spread a small erasure-coded repair of those $f$
    records, one $O(K)$-sized shard each, bound to the leader
    block by one further commitment field. Any $f+1$ shards
    reconstruct the exact certificate.
\end{itemize}
Per-view communication is then $O(Kn^2+n^3)$ bits in \emph{all}
regimes, i.e.\ $O(n^3)$ when $K=\Theta(n)$, matching the vote
exchange itself, with per-processor traffic $O((K+n)n)$ outside
the leader role. Good-case latency is unchanged: a leader holding
its anchor proposes on entry and finalises two delays later. The
costs appear only in adversarial corners: the conservative
timeout becomes $4\Delta$, a leader denied its anchor may wait up
to $3\delta$ preparing one, and the recovery coding is local
computation. The supplementary note accompanying this paper
specifies the construction exactly, block-format commitment
included, and proves every affected result of this section for
it in full; the paper deliberately carries only this sketch.
\new{The supplementary material accompanying this paper is
available at
\url{https://github.com/AndrewLewis-Pye/multimmit-supplementary};
the dissemination variant's full construction and proofs are its
communication supplement.}
\section{Optimisations}\label{sec:optimisations}

\paragraph{Skipping leaders that appear crashed}\label{sec:skipcrashed}
During synchrony, the one case in which a view finalises
nothing however the non-leader processors behave is a failed
leader: correct processors time out, the view nullifies, and every
chain waits a view. The wait is short, and
Section~\ref{sec:achieves} argues that the case is rare when
leaders have something to lose; a cheap measure shortens it
further. On entering a view, a processor checks whether it
has heard anything from the leader within some period of time (to be set as a
parameter). If it has not, while enough other processors have been
active to rule out its own disconnection, the processor treats its
timer for the view as already expired, and acts upon entry exactly as
it would at $2\Delta$ on its timer. Safety is unaffected, since the
analysis of Section~\ref{sec:verification} nowhere constrains how
timeouts are chosen. Liveness asks only that the rule eventually stop
firing on correct leaders, which holds for any period long enough that
every correct processor is heard from within it once the network
stabilises.

\paragraph{Ordering variants: reputation without fragility}
Multimmit is designed to be robust \emph{without} a reputation
mechanism, and reputation mechanisms bring complexities of their
own (Section~\ref{sec:intuition-raptr}). Still, there are
circumstances in which such a system may be useful in minimising
the one cost that halts at unsettled chains impose on the chains
behind them, faulty or not: the one-view placement wait, incurred
when emission halts at an unsettled chain's slot. The horizontal ordering of Section~\ref{sec:spec}
can be replaced.  An ordering that visits likely laggards last reduces the
placement wait directly. The blame rule should score the
chain at whose slot the sweep halted, never the chains deferred
behind it, which are bystanders. Concretely, let processors maintain a
\emph{blame set} of at most $f$ chains, by any rule on which all
correct processors agree. Keep the proposal pass before the extension
pass, applying the following replacement order separately within each
pass. Visit the positions of unblamed chains first, in horizontal order but with
the chain order shuffled each view by a VRF seed on the view
number, so that no chain index enjoys a standing positional
advantage.\footnote{A hash of the leader block would serve as a
seed, but is grindable by the leader; the VRF admits one value per
view.} Then visit the blamed chains vertically, one chain at a time
in the same shuffled order. Within that pass, a halt at one blamed
chain defers only the chains behind it; a proposal-pass halt also
abandons the extension pass. After completing the proposal pass,
restart this order on extension positions, using the same inherited
tips as the origins of offsets. Extension positions need only be
visited through each chain's proposed height plus $e$, independently
of the tips being swept towards.

This sweep suffices, whatever the blame rule. It changes only
the order in which positions are visited. That order is fixed per
view by data held by every processor computing the view's
segment, namely the proposal, the agreed blame set and the seed,
which travels with the leader block, and it does not depend on the tips being
swept towards. Each chain's positions are still visited in
increasing height, and the rule that appends, skips a settled
chain, or halts is untouched. The arguments of
Section~\ref{sec:verification} use nothing more. The prefix lemma
and the monotonicity of emission compare sweeps position by
position in the common order, and block inclusion and liveness
use only which blocks a segment contains, which reordering never
changes. So consistency, finality and inclusion are unaffected,
and a misblamed honest chain loses only placement timing, never
membership: a halt at an earlier unsettled chain can defer even a
settled chain to the next finalised leader block, blamed or not,
and the blame rule changes only which chains sit earlier in the
order.

We do not specify a blame rule, since attribution is subtle. A
chain may show short because its producer withheld blocks, because
the \emph{leader} proposed junk entries, or through passing
asynchrony; separating the first two cases requires the challenges
of the accountability paragraph of
Section~\ref{sec:verification}. We note only that a sound rule
must track the rank-$(3f+1)$ statistic at which placement halts,
not the rank-$(f+1)$ statistic at which safe-to-extend tips
advance: a producer disclosing each block to exactly $2f+1$ correct
processors keeps its tips advancing in every view, while halting
every correct processor's emission at its slot.

\paragraph{One signature scheme} The setup of
Section~\ref{sec:setup} equips processors with an aggregate scheme
for votes and with two threshold schemes, at thresholds $n-2f$ for
DA-certificates and $2f+1$ for nullifications. The threshold
schemes can be dispensed with. A DA-certificate becomes
$(\hd(b),W,\sigma)$, where $W$ is a bitmap of $n-2f$ signers and
$\sigma$ aggregates ordinary signatures on $(\da,\hd(b))$, and
nullifications are treated likewise. The message flow is unchanged,
and no use the analysis makes of a DA-certificate consumes the
threshold property: availability arguments count signers (and the
bitmap tells a retriever whom to ask), uniqueness of certified
blocks is quorum intersection over signer sets, and bases and the
DA window compare header values. Verification still costs one
pairing, since the messages are common, together with a popcount
that the scheme no longer enforces. The losses are modest.
Certificates grow by $n/8$ bytes ($13$ bytes at $n=100$), and are
no longer unique as bit strings, so deduplication must key on
content; the two places where the analysis invokes uniqueness of
threshold signatures (deduplication, and the observation in
Lemma~\ref{lem:leader-view} that what was forwarded is the
nullification itself) instead read `the first nullification
received or assembled', mirroring the treatment of V-QCs. The gains
are that the cryptographic model shrinks to the aggregate scheme
the protocol already requires, that no distributed key generation
is needed, and that a membership change becomes an update to a key
list rather than a resharing ceremony. One consideration cuts the
other way. A threshold certificate verifies against a single fixed
group key, while the aggregate requires the validator key list, so
an external consumer that verifies DA-certificates and nothing
else, a minimal light client, say, is arguably better served by the
threshold instantiation. The choice is a matter of deployment; the safety and liveness
arguments are unchanged.

\paragraph{Batch verification} Little here is specific to Multimmit,
but the impact on performance is significant for the values of $n$
considered in Section~\ref{sec:sizes}, since steady-state CPU cost is
dominated by signature verification. A producer collecting DA-vote shares
may aggregate optimistically and verify the combined signature once,
falling back to per-share checks only on failure, and the same
applies to nullify shares. A processor receiving a burst of votes for
one view may verify them as a single multi-pairing under random
scaling coefficients, bisecting on failure; and QC verification,
costing one pairing per distinct message in the aggregate, batches
across QCs the same way. In each case the optimistic path costs one
batched pairing-product check (Miller-loop work still scaling with
the number of distinct messages), and the fallback isolates invalid
objects for rejection. Attribution requires a valid sender signature
covering the rejected input; a failed signature alone is not proof
against its purported signer.
Transport offers similar
headroom: votes may travel content-addressed, a digest standing in
for a body the receiver can reconstruct, and the forwarding duty may
be advertise-and-pull, bringing good-case consensus bandwidth close
to the object sizes themselves.
\begingroup
\hypersetup{linkcolor=black,citecolor=black,urlcolor=black}
{\color{black}
\section{Payment chains}\label{sec:payments}

This section adds payments and \emph{payment receipts} to the protocol of
Sections~\ref{sec:intuition}--\ref{sec:spec}. An account is a public key,
and a payment is a signed transfer between accounts. Certain transaction
chains, called \emph{payment chains}, carry only payments and signed
instructions to change the chain permitted to debit an account. We call
these instructions \emph{moves}, with formats specified in
Section~\ref{subsec:pay-accounts}. Other transaction chains are called
\emph{ordinary chains}. A \emph{payment block} is a transaction block
of a payment chain. Transactions appearing on ordinary chains cannot
debit payment account balances.

The aim is the guarantee promised in Section~\ref{sec:intro}: a payment
has a receipt once its block has acknowledgments from $n-f$ processors,
roughly $2\delta$ after dissemination. A receipt is evidence that the
payment will execute, even if its producer withholds the assembled
certificate. Here and throughout this section, an \emph{acknowledgment}
is a DA-vote for the whole transaction block, not a separate signature
on a payment. A \emph{receipted block} is a payment block for which
$n-f$ distinct processors have sent acknowledgments, whether or not
these have been assembled into a certificate. The precise receipt and
its guarantees are given in
Sections~\ref{subsec:pay-validity} and~\ref{subsec:pay-verification}.

\paragraph{Epochs and notation}
The protocol runs in \emph{epochs}, numbered $\ep=1,2,\dots$. Each epoch
has a fixed committee $\Pi_\ep$ of $n$ processors and its own sequence of
views $v=1,2,\dots$, so a view is identified by $(\ep,v)$.
The committee and the chain permitted to debit each account may change
between epochs. Different epochs' views may overlap in time during a
transition. Write $\mathrm{log}_\ep$ for the cumulative log through the end of
epoch $\ep$, and let $\mathrm{log}_0$ be the prescribed initial log.
The initial account balances and sequence numbers are specified in
Section~\ref{subsec:pay-accounts}. The protocol selecting
$\mathrm{log}_\ep$ is specified below. For a payment chain $i$ retained
from the preceding epoch, its \emph{initial block} in epoch $\ep$ is its last block
in $\mathrm{log}_{\ep-1}$. A newly introduced chain has its prescribed genesis
block instead. The initial block is not executed again.

\paragraph{The idea}
Each account has one debiting chain in an epoch, so its payments can
be checked against that chain's history, counting credits from other
chains only once they have been ordered and executed.
Acknowledgments attest these checks as well as data availability.

The base protocol can order uncertified blocks through proposals and
extensions. On payment chains alone, we set $d_{\mathrm{prop}}=0$ and
$e=0$. Ordinary within-epoch emission then includes only certified
blocks and their ancestors, so it cannot select a competing branch in
place of a receipted block. The extraction thresholds remain unchanged.
At the end of an epoch, processors stop acknowledging its payment
blocks and report what they acknowledged. These reports determine a
finite path to append to each payment chain, containing every receipted
block, although the path need not itself end at a certified block.

Sections~\ref{subsec:pay-config}--\ref{subsec:pay-lanes} specify these
rules, while Section~\ref{subsec:pay-verification} proves their guarantees
and Section~\ref{subsec:pay-costs} accounts for their costs.
}

{\color{black}
\subsection{Configuration}\label{subsec:pay-config}

We first specify how payment chains use the existing proposals, votes
and certificates, and fix their parameters and block format.

Recall from Section~\ref{sec:spec} the three per-chain parameters:
the proposal window $d_{\mathrm{prop}}$, bounding the payload
hashes a chain proposal lists above its base, the DA window
$d_{\mathrm{DA}}$, bounding production and DA-voting beyond a
certified block, and the extension bound $e$. On every ordinary
chain, $d_{\mathrm{prop}}=d_{\mathrm{DA}}$.
}

{\color{black}
\paragraph{Payment chains} A \emph{payment chain} is a chain configured with
$d_{\mathrm{prop}}=0$ and $e=0$, and with $d_{\mathrm{DA}}\ge 1$
chosen freely. With $d_{\mathrm{prop}}=0$, a leader block lists
nothing above the chain's base and, with $e=0$, votes carry no extensions.
The chain's initial block, and the certificates permitted as bases,
are specified in Section~\ref{subsec:pay-lanes}.

On a payment chain the transaction-block format of
Section~\ref{sec:spec} gains two fields. A transaction block of payment
chain $i$ is a tuple
$b=(i,h,x,\ep,r,\Tr)$, signed by the producer on
$\hd(b)=(i,h,x,\ep,r,H(\Tr))$. Here $\ep$ is its epoch number,
$r$ is the hash of an L-QC of that epoch, and $\Tr$ consists of
payments followed by moves. As permitted in Section~\ref{sec:spec},
$H(\Tr)$ is the Merkle root of the transaction sequence. The fields
$i,h,x$ and the parent relation retain their meanings from that section.
Ordinary-chain blocks need no epoch field: their epoch of execution
is determined by their position in the log.

\paragraph{References} Let $Q$ be the L-QC identified by $r=H(Q)$.
The \emph{reference prefix}, or simply \emph{reference}, of the
transaction block $b$ is $R:=\Emit(Q)$. So $r$ is a hash, $Q$ is a
certificate, and $R$ is the ordered prefix obtained from that
certificate. The state obtained by executing $R$ is used to check
$b$'s payments. An L-QC, rather than a V-QC, is required because
$\Emit(Q)$ is a prefix of every eventual correct log, whereas a V-QC
need not determine a committed prefix. The epoch-transition rules
below preserve this property. DA-votes and DA-certificates sign
$\hd(b)$ as in Section~\ref{sec:spec}.
}

{\color{black}
\paragraph{Placing only certified blocks} A leader block's entry for a
payment chain consists of a base and nothing above it. The entry is
either the chain's tip in the referenced V-QC or a DA-certificate,
attached for verification, for a block of that chain at strictly
greater height than the tip. The DA window remains
$d_{\mathrm{DA}}$, so production can run ahead of certificates while
only ordinary within-epoch placement waits for a certificate.
A correct payment-chain producer broadcasts each certificate it
assembles: the certificate format and size are given in
Section~\ref{sec:sizes}. Broadcasting brings placement forward from
the end of the epoch to a subsequent correct-leader view
(Theorem~\ref{thm:receipt}(iii)).

We also refine the horizontal order of Section~\ref{sec:spec}.
In both $\Ord$ and $\Emit$, the sweep first visits the payment-chain
positions in horizontal order, in what we call the \emph{payment pass},
then performs the existing proposal and extension passes on the
ordinary chains. Each pass restarts at offset $1$, measured from the
same inherited tips. Since each payment chain is settled in every
view, an unsettled ordinary chain cannot stop the payment pass.
This statement concerns the ordering of block identifiers, while
obtaining missing block data remains subject to the retrieval convention.
}

{\color{black}
\paragraph{Spines} The \emph{spine} of payment chain $i$ in epoch
$\ep$ consists of its initial block and every epoch-$\ep$ block
that is certifiable, together with the ancestors of those blocks
above the initial block. Recall that \emph{certifiable} means that
$n-2f$ processors have sent DA-votes, whether or not a certificate
has been assembled. Lemma~\ref{lem:pay-spine} shows that the spine is
one path, which ordinary within-epoch emission uses. The closing
procedure may also append an uncertified suffix and may omit a
privately held certificate that has only $n-2f$ acknowledgments.
}

\paragraph{Signatures} A receipt uses acknowledgments from $n-f$
processors of one committee, and the base protocol's $(n-2f)$-threshold scheme
cannot express that quorum, so payment deployments instantiate the
one-signature-scheme variant of Section~\ref{sec:optimisations}.
Acknowledgments are individually signed and aggregated, a
certificate carries a signer bitmap of $\lceil n/8\rceil$ bytes,
and committee changes reduce to key-list updates.

{\color{black}
\subsection{Payments, accounts and homes}\label{subsec:pay-accounts}

We now define accounts, payments and moves, and specify which chain
may debit each account in each epoch.

\paragraph{Accounts} An \emph{account} is a public key, and its owner is
the party able to sign with the corresponding secret key. Each account has
a nonnegative balance and a next sequence number, specified initially
by the prescribed initial state and thereafter by execution of the log.
The \emph{account state} records these balances and next sequence
numbers for all accounts. The \emph{home} of account $a$ in epoch
$\ep$, denoted $\mathrm{home}_\ep(a)$, is the payment chain whose blocks
alone may debit it in that epoch.

\begingroup\emergencystretch=16pt
\paragraph{Payments} A \emph{payment} specifies a debited account $a$,
a sequence number $s$, and a finite nonempty list of credited accounts
and positive amounts, and is signed by $a$'s owner. It is
\emph{funded} at a state if $a$'s balance covers the sum of those
amounts. Successful execution subtracts that sum from $a$'s balance,
credits the listed accounts, and increases $a$'s next sequence number
by one. The payment must use $a$'s next sequence number and be carried
by $\mathrm{home}_\ep(a)$ to execute in epoch $\ep$. Two distinct
payments \emph{conflict} if they debit the same account at the same
sequence number.

\par\endgroup
\paragraph{Moves and homes} A \emph{move} specifies an account and the
identifier of the payment chain to become its home, and is signed by
the account's owner. Moves debit nothing and may appear on any chain.
A correctly signed move ordered in epoch $\ep$ takes effect in epoch
$\ep+2$. Among several moves, the last in log order determines the
home. So, for $\ep\ge3$, $\mathrm{home}_\ep(a)$ is determined by
the latest move for $a$ in $\mathrm{log}_{\ep-2}$, or by its initial home if
there is no such move. Homes in the first two epochs are prescribed
initially. A move to a chain whose producer is absent from the
relevant committee does not make another chain eligible to debit
the account, which remains unspendable until a further move becomes
effective at an available chain.

The committee $\Pi_\ep$ is likewise determined by $\mathrm{log}_{\ep-2}$ for
$\ep\ge3$, with the first two committees prescribed initially.
These two-epoch delays make the committee and homes of epoch
$\ep+1$ known before epoch $\ep$ closes.

Only debit authority is divided among chains, while every processor
keeps the whole account state and checks every payment block. Since all
of an account's debits in an epoch lie on one chain, they can be
checked along that chain before its blocks are ordered. Credits
from another chain are counted only after they have been ordered
and executed. A move submitted to a correct ordinary-chain producer
has the base protocol's inclusion protection, subject to the
cross-epoch liveness conditions below, so an account need not obtain
its current home's permission to move.
}

{\color{black}
\subsection{Validity and acknowledgments}\label{subsec:pay-validity}

This subsection specifies payment execution, the checks required before
acknowledging a payment block, and the receipts formed from those
acknowledgments.

\paragraph{Execution} Each correct processor $p$ maintains an
\emph{executed prefix} $\sigma_p$ of its log and the account state
obtained from it. Execution processes transactions in log order and 
a payment executes only if its signature, home, sequence number and
funding checks pass (otherwise it has no effect). Moves with valid
signatures are recorded for the home function. The receipt theorem
will show that every payment of a receipted block passes these
checks when reached. Payment execution reads the ordered block occurrences.
A payment ignored outside its home chain neither consumes its sequence number
nor suppresses its later execution on its home chain. Its sequence
number prevents repeated successful execution (the base duplicate
filter for ordinary transactions is unchanged).

\paragraph{Reference prefixes and ancestor paths}
Recall from Section~\ref{subsec:pay-config} that a transaction block
$b=(i,h,x,\ep,r,\Tr)$ of payment chain $i$ has reference prefix
$R=\Emit(Q)$, where $r=H(Q)$. For an ordered prefix $\sigma$, its
\emph{chain-$i$ tip} is the last transaction block of chain $i$
included in $\sigma$, or the epoch's initial block if no new such
block is included. All prefixes used for epoch-$\ep$ validation
include $\mathrm{log}_{\ep-1}$.

Let $\gamma(b,R)$ be the sequence, in increasing height order, of
$b$'s strict ancestors above the chain-$i$ tip in $R$. The tip must be an ancestor of $b$, and $\gamma(b,R)$ must
be a complete parent-linked path ending at $b$'s parent, otherwise
$b$ fails the validity test below. A processor must hold this path's
contents, not only its headers; the initial block is excluded from
the path, since its transactions have already been executed.

The validity test below uses a copy of the account state and does
not change the processor's log or executed prefix.

\paragraph{Validity at an executed prefix}
Let $b=(i,h,x,\ep,r,\Tr)$ be a transaction block of payment chain
$i$, with reference prefix $R$, and let $\sigma$ be an executed prefix
in epoch $\ep$ with $R\preceq\sigma$, whose
chain-$i$ tip is a strict ancestor of $b$. Starting from the state
at $\sigma$, process the blocks of $\gamma(b,R)$ not already in
$\sigma$ in increasing height order, and the transactions within
each block in their specified order. Each payment must pass its signature, home,
sequence-number and funding checks before its balance and
sequence-number changes are applied. Every move must have a valid
signature. If any check fails, $b$ is invalid at $\sigma$.
Otherwise, $b$ is \emph{valid at $\sigma$} if:
\begin{enumerate}
\item $\Tr$ is a sequence of payments followed by a sequence of
  moves, and nothing else;
\item every payment of $\Tr$ is signed by the owner of the
  account it debits, and every move by the owner of the account it
  moves;
\item every account debited in $\Tr$ has chain $i$ as its home in
  epoch $\ep$;
\item per debited account, the sequence numbers in $\Tr$ continue
  consecutively from its next sequence number in the state so
  obtained and through the earlier payments of $\Tr$; and
\item every payment of $\Tr$ is funded in that same state,
  advanced through the earlier payments of $\Tr$.
\end{enumerate}
The \emph{reference validation} of $b$ is this test with $\sigma=R$,
so the calculation includes every block of $\gamma(b,R)$.
}

{\color{black}
\paragraph{Acknowledgments} In addition to the DA-voting conditions
of Section~\ref{sec:spec}, a correct processor $p$ acknowledges the
payment block $b$ only if:
\begin{enumerate}[label=(\roman*),itemsep=1pt]
\item $p\in\Pi_\ep$ has accepted $\mathrm{log}_{\ep-1}$ as the completed
  preceding log, has not permanently stopped acknowledging
  epoch-$\ep$ payment blocks, and $r$ names an L-QC of that epoch;
\item $p$ has executed the reference $R$, so $R\preceq\sigma_p$;
\item $p$ holds every block of $\gamma(b,R)$, including certified
  ancestors it never personally acknowledged; and
\item $b$ passes reference validation.
\end{enumerate}
The initial DA-voting conditions still impose one vote per height
and the $d_{\mathrm{DA}}$ window. Requiring the contents of every
block of $\gamma(b,R)$ is an additional possession condition (a
certificate cannot substitute for the contents needed to check payments). Missing data must first
be retrieved under the convention of Section~\ref{sec:liveness}.

Conditions (i)--(iii) concern $p$'s readiness, while condition (iv)
depends only on the block, its named reference and its ancestry. In
particular, $p$ applies that same test even if its own execution has
already passed $b$, checking $b$ using its named reference, not
$p$'s later account state. The processor signs the whole block or nothing. The sole
exception to condition (i) is the explicitly conditional early
handoff in Section~\ref{subsec:pay-lanes}.
}

{\color{black}
\paragraph{Production} For each new payment block, a correct producer
chooses an L-QC $Q$ of the current epoch whose emitted prefix it has
executed, preferring one with the longest such prefix. It sets the
block's $r$-field to $H(Q)$ and validates the block at the reference
prefix $\Emit(Q)$, advanced through the unplaced ancestors on its own
chain, exactly as its acknowledgers will. It does not count additional
credits visible only in its longer local executed prefix. It admits
pending valid payments oldest-first, defers a payment that is not yet
funded, and drops one whose sequence number has already been used.
Moves must have valid signatures. The base production loop and
$d_{\mathrm{DA}}$ window otherwise remain unchanged.
}

{\color{black}
\paragraph{Receipts} A \emph{receipt} consists of a payment, the signed
header of the payment block containing it, a Merkle proof of its
inclusion in that block, and $n-f$ distinct processors' acknowledgments
of the header, aggregated with a signer bitmap. Verification checks
the producer's signature, the payment's signature, the Merkle proof,
and the acknowledgments against the epoch's committee and quorum
size (and requires no account state). The acknowledgments defined so
far are \emph{unconditional}, whereas the optional early-handoff
procedure of Section~\ref{subsec:pay-lanes} uses conditional
acknowledgments, which require additional evidence from consensus
before they can be used in a receipt.

A receipt guarantees that its payment will execute, while execution
itself occurs when the block is placed and its data are available. These distinctions are made precise in
Theorem~\ref{thm:receipt}.
}

{\color{black}
\subsection{Epochs and the close}\label{subsec:pay-close}

We next specify how an epoch ends and how its processors report enough
information to include every receipted block in its completed log.

\paragraph{The fault assumption} For each epoch, at most $f$ members
of its committee are ever faulty, before, during or after that epoch,
and $n\ge5f+1$. A processor once corrupted is counted faulty in every
committee it joins. This assumption protects certificates verified
against a retired committee's keys: later corruption of those keys
must not permit a different history to be signed. Key erasure or
additional authenticated checkpoints can support weaker assumptions,
but are not part of this specification~\cite{komatovic2025pos}.

\paragraph{Closing views} Fix a nominal epoch length $V\ge1$.
Views $(\ep,1),\dots,(\ep,V-1)$ are the epoch's \emph{non-closing
views}; views $(\ep,V),(\ep,V+1),\dots$ are its \emph{closing
views}. A \emph{closing block} is a leader block proposed in a closing
view. A \emph{closing vote} is a vote for such a leader block, and a
\emph{closing L-QC} is an L-QC for it.
At this transition, the \emph{old committee} is $\Pi_\ep$ and the
\emph{successor committee} is $\Pi_{\ep+1}$. The old committee continues its consensus protocol until a closing
L-QC is available, moving to the next view whenever a closing view
fails to produce one.

A closing L-QC supplies \emph{reports} (defined below) from $n-f$
processors, and the rule below uses them to construct a possible
completed log for epoch $\ep$, called a \emph{candidate completed log}. Different closing L-QCs may give
different candidate logs. Epoch $\ep+1$'s
consensus selects one, as specified in the next subsection. Merely
receiving a closing L-QC does not select its candidate log. A processor's
\emph{adoption} of the completed log means that it has verified this
selection, obtained the log and accepted it as the prefix to extend.
The precise conditions are given under `Participation and adoption'
in Section~\ref{subsec:pay-open}.

\paragraph{Fences} The close must account for receipts whose signatures may be
assembled only later, so we stop each correct processor's old-epoch
acknowledgments before it reports what it has acknowledged.
A processor \emph{fences} epoch $\ep$ by recording
an irrevocable decision to send no further DA-votes for payment blocks
of that epoch. A correct processor does this before sending any
message of the first closing view it reaches, but may still vote for
leader blocks in closing views. In the formal model the fence is a
bit of state that is never unset (a deployment must make it durable
before sending the first closing-view message).

\paragraph{Closing reports} \emph{Reports} record the voters' DA-voting
information, so that the close can include every receipted payment
block even when its certificate has been withheld. Each closing vote
signs, for each payment chain $i$, a report
\[
 R_i=(\hd(b_i),(\hd(c_1),\dots,\hd(c_k))).
\]
The \emph{base} of the report is the payment-chain block $b_i$.
A DA-certificate is \emph{eligible} if the rules of
Section~\ref{subsec:pay-lanes} permit its use as a base for that chain
and epoch. Once the preceding epoch's completed log has been adopted,
only certificates of the epoch being reported are eligible. The block
$b_i$ is the highest block for which the processor holds an eligible
certificate, or the epoch's initial block if it holds none above it.
The sequence $c_1,\dots,c_k$ is its
complete consecutive DA-voted suffix extending $b_i$, in increasing
height order, and
$0\le k\le d_{\mathrm{DA}}$. The eligibility of certificates is
specified in Section~\ref{subsec:pay-lanes}, including reports sent
before adoption. Earlier votes on a branch not extending $b_i$ are
not listed. The certificate for $b_i$, or evidence authorising the initial
block, accompanies the vote but need not be inside the signed
report, whose full headers already identify the blocks.

A report is well formed only if its base evidence verifies and its
listed headers have the required chain, heights and parent hashes.
Full headers are necessary because siblings differing only in their
reference fields can have the same payload hash. A closing vote is
valid only if every report is well formed. Reports are retained in
full when votes are represented in a QC, and signatures on their
different signed contents aggregate as in Section~\ref{sec:setup}.

\paragraph{Freezing} The log selected by the successor committee must
extend everything the old committee emits. To this end, we prevent later
closing views from advancing the ordinary emitted prefix: in a closing view, every chain's proposal entry
must be its tip in the referenced V-QC, with no additional entries.
Every vote must have an empty extension field on every chain and
include the reports just defined. In these views $\Emit$ uses $\rho:=\lambda$, so a closing block
emits only the ordering already determined by its reference. 

Along a sequence of closing blocks linked by references, $\Tips$
and $\Ord$ consequently do not advance. All closing L-QCs in an
epoch determine the same emitted prefix, since their leader blocks
lie on one reference chain. Their \emph{reports} may nevertheless differ,
and so may the additional payment blocks selected by the rule below.

\paragraph{The cut} The cut determines the last block of each payment
chain in the completed log. The chosen paths must preserve the emitted
blocks and include every receipted block, even when its certificate
has been withheld.

Fix a closing L-QC $Q$ of epoch $\ep$ and a payment chain $i$, and
let $u_i$ be its tip in $\Emit(Q)$. The \emph{top} of report $R_i=(\hd(b_i),(\hd(c_1),\dots,\hd(c_k)))$ is
$c_k$ if $k>0$, and $b_i$ otherwise. A report \emph{covers} a block
if that block is an ancestor of, or equal to, its top. Coverage
includes ancestors of a certified base, even when their headers
are not listed in the report. Missing headers are obtained under
the retrieval convention of Section~\ref{sec:liveness}.

For a strict descendant $b$ of $u_i$, let $s_i^Q(b)$ be the number
of $Q$'s reports that cover $b$. Define $\mathrm{cut}_i(Q)$ to be
the block of greatest height in
\[
 \{u_i\}\;\cup\;
 \{b:\ b\text{ strictly extends }u_i\text{ and }s_i^Q(b)\ge n-3f\}.
\]
Lemma~\ref{lem:pay-cut} proves that this block exists and is unique.
The candidate \emph{completed log} $\mathrm{log}_\ep(Q)$ is $\Emit(Q)$ followed,
in increasing payment-chain index, by the path strictly above
$u_i$ through $\mathrm{cut}_i(Q)$, in height order. Ordinary chains
receive no additional blocks in this operation. The selected
$\mathrm{log}_\ep$ is one such $\mathrm{log}_\ep(Q)$.

The reports used are always the named $Q$'s own $n-f$ signed reports.
If a closing leader block becomes fixed only because a descendant
is L-notarised, the ancestor need not itself have an L-QC. In that
case the descendant's closing L-QC can be named and supplies the reports.

\begin{algo}{\textcolor{black}{closing epoch $\ep$, instructions for $p_i$}}{alg:close}
\color{black}
\aline{Before sending its first message in a view $(\ep,v)$ with $v\ge V$:}
\aline[1]{Permanently stop DA-voting for epoch-$\ep$ payment blocks;}
\aline{In each closing view:}
\aline[1]{Propose the referenced V-QC's tips, with no further entries;}
\aline[1]{Vote under Algorithm~\ref{alg:consensus}, with every extension field empty;}
\aline[1]{Include a signed report $R_j$ for each payment chain $j$;}
\aline[1]{Attach the evidence for every report's base;}
\aline[1]{Compute $\Emit$ with $\rho:=\lambda$.}
\end{algo}
}

{\color{black}
\subsection{Opening the next epoch}\label{subsec:pay-open}
\label{subsec:pay-epochdefs}

This subsection specifies how the next committee selects the preceding
epoch's completed log and begins extending it. The committee for epoch
$\ep+1$ selects the closing L-QC that determines $\mathrm{log}_\ep$
through its ordinary consensus protocol, with the following changes to
the message formats and initial references. These are the per-epoch
interpretations of Section~\ref{sec:spec}.

\paragraph{Leader blocks and their references}
A leader block is now $\ell=(\ep,v,q,cq,C)$, with identifier
$H(\ell)=H(\ep,v,q,cq,H(C))$. The fields $v,q,C$ retain their roles
from Section~\ref{sec:spec}, subject to the initial-reference rule
below. The field $cq$, short for `closing quorum certificate', stores
the hash of the closing L-QC whose candidate log the leader block
proposes to extend. In epoch $1$, $cq=\bot$, and the prescribed initial
objects of Section~\ref{sec:spec} are used.

For $\ep>1$, a leader block has one of two forms.
\begin{enumerate}[label=(\roman*),itemsep=1pt]
\item An \emph{opening block} is a leader block that names a complete closing L-QC $Q$ of
  epoch $\ep-1$ and sets $q=cq=H(Q)$. This reference supplies the
  candidate initial log $\mathrm{log}_{\ep-1}(Q)$, not evidence of progress
  through epoch-$\ep$ views. For consensus ancestry, the parent
  of every opening block is the same distinguished genesis leader block
  $\ell_{\mathrm{gen}}^{\ep}$, assigned view $0$ and an
  identifier fixed by the epoch, independently of $Q$.
\item A \emph{continuing block} is a leader block that references an earlier-view V-QC
  $W$ of epoch $\ep$, sets $q=H(W)$, and copies the $cq$ of the
  leader block designated by $W$. Its consensus parent is that
  designated block, as in Section~\ref{sec:spec}.
\end{enumerate}
A proposal of an opening or continuing block is called an
\emph{opening proposal} or \emph{continuing proposal}, respectively.
All opening blocks belong to one consensus execution with committee
$\Pi_\ep$ and view numbers $1,2,\dots$, regardless of which $Q$ they
name (they do not start separate executions).

\paragraph{View progression and signatures}
A continuing proposal in view $(\ep,v)$, referencing view
$(\ep,u)$, requires epoch-$\ep$ nullifications for every
$u<w<v$, as in the base validity rule. An opening proposal in
view $(\ep,v)$ instead requires them for every $1\le w<v$.
Selection of the greatest-view anchor considers only epoch-$\ep$
V-QCs. If none is held, a correct leader uses the opening rule (an old closing L-QC's view number is irrelevant to these checks).

All consensus signatures bind their epoch. Votes bind the extended leader-block identifier, while novotes and
nullify shares sign $(\novote,\ep,v)$ and $(\nullify,\ep,v)$,
respectively.
Certificates are checked against the committee and thresholds of
that epoch. Each correct processor sends at most one vote per
$(\ep,v)$ and never both a vote and a novote for that pair.
As always, a QC counts at most one message from each processor. All other
voting and nullification conditions remain those of
Algorithm~\ref{alg:consensus}. A processor in two successive
committees maintains separate view variables and timers for their
consensus executions. Each committee uses the common round-robin
leader schedule of Section~\ref{sec:spec}, with its own ordered
membership list.

\paragraph{The initial ordering and proposal bases}
For an opening block naming $Q$, the inherited log in the definitions
of $\Ord$ and $\Emit$ is $\mathrm{log}_{\ep-1}(Q)$. The initial tips are its
payment-chain cut blocks and its frozen ordinary-chain tips. The
proposal is checked relative to these tips, which serve as its
candidate initial bases and become adopted bases when this opening
block or a descendant is finalised. For continuing leader blocks,
the usual recursive definitions of $\Ord$ and $\Emit$ apply, with
the same inherited $cq$. Finalising an opening block or a descendant
therefore selects one candidate log.  The logical consensus genesis is common even though the
candidate transaction logs are alternatives before selection.

\paragraph{Possession before voting}
In epoch $\ep>1$, consider an opening or closing leader block with
$cq=H(Q)$. A correct processor votes for it only after holding $Q$ in
full, including its votes and reports, and every transaction block and
referenced certificate needed to compute and execute
$\mathrm{log}_{\ep-1}(Q)$. We call these objects the \emph{boundary data}
for $Q$. This requirement applies to both voting routes in
Algorithm~\ref{alg:consensus}, including the route triggered by a
V-QC for the current view. A correct leader holds and disseminates these data with its opening
or closing proposal, meeting the additional possession requirement
on the proposal's existing dissemination leg.
Lemma~\ref{lem:pay-custody} gives the resulting holder counts.

The V-QC-triggered route to leaving such a view also requires these
data. If a processor receives that V-QC first, it requests the
missing data and defers the associated vote and view advance together,
while its usual timeout and nullification rules remain available.
A correct processor forwarding such a V-QC for view progression
sends the required boundary data with it. This
prevents an early certificate from making a processor leave the
view before it can cast a still-permitted vote.

\paragraph{Participation and adoption}
A member of $\Pi_{\ep+1}$ may start its consensus execution at
view $(\ep+1,1)$ after obtaining a closing L-QC of epoch $\ep$
in full, and may subsequently catch up using epoch-$(\ep+1)$ V-QCs
and nullifications. Starting this execution does not choose a cut.
A processor \emph{adopts epoch $\ep+1$} after verifying finality
of an epoch-$(\ep+1)$ leader block naming $Q$, obtaining
$\mathrm{log}_\ep(Q)$, and accepting that log as the completed prefix
$\mathrm{log}_\ep$. Its current adopted epoch is initially $1$.
Unconditional acknowledgments of epoch-$(\ep+1)$ payment blocks require this
adoption and the execution conditions of
Section~\ref{subsec:pay-validity}.

A processor in both committees may participate in epoch $\ep$'s
closing views and epoch $\ep+1$'s views concurrently before
adopting the latter. Its fence forbids further old-epoch payment
acknowledgments, not its closing consensus votes. A new committee
member need not participate in the old committee to verify its
certificates or acquire its log.

\paragraph{Retirement, transfer and synchronisation}
Closing L-QCs are relayed on first receipt, with their votes and
reports, as V-QCs are. An outgoing processor not in the successor
committee broadcasts the first closing L-QC it obtains and stops
proposing and voting in the old epoch, but continues serving requests
for data it holds until the following transfer condition is met.

We assume that, after GST, within bounded time of a closing L-QC
being held by a correct old-committee processor, every correct
successor member obtains a complete closing L-QC and the data needed
to compute and execute its candidate log. The candidates need not
be identical. View entry and catch-up use the successor's own
certificates, with the synchronisation property required by
Section~\ref{sec:liveness}. These are the state-transfer and
synchronisation assumptions used in the liveness results below.
The specification does not derive their implementation from
certificate existence alone.

After a successor leader block is finalised, the selected candidate's data
are held by correct successor voters
(Lemma~\ref{lem:pay-custody}) and remain available under the
retrieval convention. A later old closing L-QC remains valid, but
cannot change the inherited $cq$ or justify skipping any successor
view. Different candidates can have different cuts (they are not
interchangeable after one has been selected).

\begin{algo}{\textcolor{black}{opening epoch $\ep+1$, instructions for $p_i$}}{alg:open}
\color{black}
\aline{On obtaining a complete closing L-QC of epoch $\ep$:}
\aline[1]{If a member of $\Pi_{\ep+1}$, start or catch up its epoch-$(\ep+1)$ views;}
\aline{When proposing without a prior epoch-$(\ep+1)$ V-QC:}
\aline[1]{Choose a closing L-QC $Q$ whose candidate log is held in full;}
\aline[1]{Set $q=cq=H(Q)$ and disseminate $Q$ and its candidate log with the proposal;}
\aline{When proposing from an epoch-$(\ep+1)$ V-QC:}
\aline[1]{Copy the designated leader block's $cq$;}
\aline{On verifying finality of an epoch-$(\ep+1)$ leader block naming $Q$:}
\aline[1]{Obtain $\mathrm{log}_\ep(Q)$, adopt it as $\mathrm{log}_\ep$, and initialise the new payment chains.}
\end{algo}
}

{\color{black}
\subsection{Chains across the boundary}\label{subsec:pay-lanes}

We now specify how chains continue across an epoch boundary, including
when payment-block production and acknowledgment may begin before the
preceding log is selected.

A \emph{continuing producer} belongs to both committees at a transition.
For a payment chain $i$, its \emph{epoch-$\ep$ lane} consists of its
initial block and the chain-$i$ blocks carrying epoch number $\ep$
and extending that initial block. A faulty producer may sign incompatible blocks in a lane, but the
selected log uses one path, as proved below. Ordinary chains retain their existing
blocks across a transition. Payment lanes use the following
cut-based initialisation.
}

\paragraph{Ordinary chains} An ordinary chain does not require
reset. A continuing ordinary-chain
producer retains its already-produced chain and all unplaced
input across the boundary. It re-disseminates unplaced blocks
unchanged, continues production from its last produced block
rather than from the frozen ordering tip, and never reuses an
already-produced height for different content, so the log
resumes from the frozen tip without the producer's signed chain
resetting to it, and votes already cast at old heights stay
valid for the unchanged blocks. Payment lanes instead use the
epoch-scoped, cut-based initialisation below, and transactions
whose producer departs need the stated resubmission.

{\color{black}
\paragraph{Handoff} Suppose an epoch-$(\ep+1)$ leader block naming $Q$
has been finalised, so that $\mathrm{log}_\ep=\mathrm{log}_\ep(Q)$. A continuing
payment-chain producer starts its new lane above
$\mathrm{cut}_i(Q)$. The first transaction block's signed header additionally carries a
\emph{declaration} identifying that cut block by its full header.
The declaration is metadata, rather than an additional transaction in $\Tr$. A correct successor processor
acknowledges blocks on the new lane only after adopting $\mathrm{log}_\ep$ and checking
that the declaration names its selected cut. The ordinary
reference-validation rule then uses an L-QC of epoch $\ep+1$,
whose emitted prefix includes $\mathrm{log}_\ep$. A new producer's chain
starts at its prescribed genesis instead.
}

{\color{black}
\paragraph{Lane initialisation} For a continuing payment chain, the
rules of epoch $\ep+1$ are initialised as follows.
\begin{enumerate}[label=(\arabic*),itemsep=1pt]
\item The cut block counts as a certified initial base, whose
  evidence is the finalised epoch-$(\ep+1)$ leader block selecting
  $Q$, rather than a DA signature. The certified height,
  $d_{\mathrm{DA}}$ window, $\mathsf{Ready}$ predicate and
  DA-voting eligibility tests all use this base.
\item One-vote-per-height tests use only DA-votes of epoch
  $\ep+1$. Old votes at a height above the selected cut do not
  prevent the new lane from using that height. New blocks must
  extend the selected cut, and every old receipted block lies
  at or below it (Lemma~\ref{lem:pay-boundary}(b)).
\item The producer's previous-block variable is set to the cut
  block. Pending and unplaced payments and moves are retained,
  and their validity is re-evaluated against the inherited state.
\item Heights continue from the cut: the first new block has
  height one greater and names the cut block as its parent.
\item In proposal entries, DA-window tests and closing reports,
  the eligible base evidence is the adopted cut block's evidence
  or a DA-certificate of epoch $\ep+1$. An older certificate
  does not authorise a base above or outside the selected cut.
  In particular, a previously hidden old certificate cannot
  restore blocks that the completed log excluded.
\end{enumerate}
These rules also apply to a new payment chain with its prescribed
genesis in place of an old cut block.
}

{\color{black}
\paragraph{Early handoff} Before a closing L-QC $Q$ of epoch $\ep$
has been selected, a successor producer holding $Q$ and
$\mathrm{log}_\ep(Q)$ may form a \emph{prospective lane} above its
candidate cut: this is a lane for epoch $\ep+1$ conditional on that
particular $Q$ being selected. Transaction blocks on this lane, and
DA-votes acknowledging them, are called \emph{prospective}. Their
headers additionally bind $H(Q)$, and their DA signatures identify
the pair $(\ep+1,H(Q))$. This pair is their \emph{voting context},
specifying both the epoch and the candidate log in which they count.
Each choice of $Q$ has separate production and voting state. The fence for epoch $\ep+1$
forbids further acknowledgments in every one of its prospective
contexts, just as on an adopted lane. So the once-per-height tests apply
within one candidate, not across competing candidates.

For a prospective transaction block $b$, its reference prefix is
$R:=\mathrm{log}_\ep(Q)$, in place of the prefix emitted by an L-QC
of its own epoch. This is the exception to the reference rule in
Section~\ref{subsec:pay-config}. A correct successor processor must
hold this candidate prefix and compute its account state for validation,
then check the full path $\gamma(b,R)$ and $b$ using the
validity test in Section~\ref{subsec:pay-validity}. It does not append the candidate
prefix to its adopted log. The candidate cut serves as the initial
base for its $d_{\mathrm{DA}}$ window. A \emph{prospective certificate}
is a DA-certificate formed from $n-2f$ prospective acknowledgments
in this voting context. Such certificates cannot extend that window or serve as consensus proposal entries
before selection, so at most $d_{\mathrm{DA}}$ new blocks per chain
can be produced on this basis before selection.

A prospective acknowledgment \emph{matures} when an
L-notarised epoch-$(\ep+1)$ leader block names $Q$, either directly
as its opening choice or by inheritance. A receipt using prospective
acknowledgments must include that leader block and an L-QC proving
its finality, or refer to the same evidence already verified by the
recipient. These objects constitute its \emph{adoption evidence}, without
which the conditional signatures are not a receipt. After selection, the selected candidate's production and
DA-voting history become the history of the adopted lane. They are
not reset and cannot be used to sign different blocks at already-voted
heights. Prospective lanes for other choices of $Q$ are abandoned and their
certificates are ineligible.

Acknowledgments for at most one choice of $Q$ can mature, by
Lemma~\ref{lem:pay-cut}. Certificates formed from matured
acknowledgments satisfy the ordinary threshold and eligibility
rules for the adopted epoch. This permits limited production and validation to overlap the
transition, although receipt issuance still waits for the adoption
evidence. This does not provide an
unconditional receipt before the next epoch selects its preceding log.

\paragraph{Closing reports before adoption}
A closing vote for an epoch-$(\ep+1)$ leader block with $cq=H(Q)$
uses the initial blocks
and DA-voting histories for the candidate $\mathrm{log}_\ep(Q)$, whether or not
the voter has adopted that log.
Before adoption, $Q$ and the verified $\mathrm{log}_\ep(Q)$ authorise the
candidate initial blocks for reporting. Unmatured certificates do
not advance the reported base. The report includes the complete
DA-voted suffix in the context $(\ep+1,H(Q))$, bounded by
$d_{\mathrm{DA}}$ (no history from another candidate is included).
After adoption the ordinary reporting rule applies, retaining this
same history. Report formation therefore need not await the epoch's
first finalised leader block.
}

{\color{black}
\subsection{Verification}\label{subsec:pay-verification}\label{subsec:pay-receipts}

We first consider a single epoch, with its committee and each payment
chain's initial block fixed. For early handoff before adoption, we fix
one closing L-QC of the preceding epoch, together with its candidate
completed log, and consider only the prospective blocks and DA-votes
associated with that choice.
Only unconditional acknowledgments, or prospective acknowledgments
that have matured under the early-handoff rule, count towards the
cross-epoch claims below.
}

{\color{black}
\begin{lemma}[Spines]\label{lem:pay-spine}
Within an epoch, no two incompatible blocks of one payment chain
both carry $n-2f$ DA-votes. The certified and certifiable blocks,
with their ancestors and the epoch's initial block, form a
single path, the chain's spine for the epoch.
\end{lemma}
\begin{proof}
The committee is fixed for the epoch and an acknowledgment is a
DA-vote, so Lemma~\ref{lem:cert-compat} applies unchanged. Any two
certified blocks of the chain are compatible, and a set of
pairwise-compatible blocks of one chain is a path.
\end{proof}

Certificates may arrive in any order, and the spine is the
\emph{certification} path, not the adopted one. The epoch's
adopted path, from the frozen tip to the selected cut, can
include an uncertified cut suffix and exclude a certificate kept
private, which is why the close works with reports and cuts
rather than with certification alone.

The close selects the path that the next epoch extends, and
clause (i) of Theorem~\ref{thm:receipt} states the resulting
cross-epoch guarantee.
}

{\color{black}
\begin{lemma}[Certified content]\label{lem:pay-content}
Every tip extracted on a payment chain within an epoch is its
initial block or a block certified in that epoch. Every transaction
block newly emitted on a payment chain by $\Emit$ within that epoch
is on a path from the chain's initial block to such a certified tip
and has DA-votes from at least
$n-3f\ge f+1$ correct processors of that epoch (the correct processors
providing these DA-votes need not be the same for different blocks).
\end{lemma}

\begin{proof}
Induct on the sequence of leader-block references. There are no
proposed entries above the base, since $d_{\mathrm{prop}}=0$, and
no extensions, since $e=0$. Every new base is a permitted certificate
or the previous tip, so extracted tips have the stated form.

Fix a new emitted block $b$, and let $b^*$ be the least-height
certifiable block on its branch at or above it. Such a block exists
because the emitted tip is certified. At least $n-3f$ correct
processors DA-voted $b^*$. Each voted from a held certified base
strictly below $b$: a base at or above $b$ and below $b^*$ would
contradict the choice of $b^*$. DA-voting condition (iii) therefore
required each to vote for and hold $b$. Inherited initial blocks
are already part of the preceding log. Additional blocks selected
only by the cut are covered by Lemma~\ref{lem:pay-boundary}(a).
\end{proof}
}

\begin{lemma}[Report paths]\label{lem:pay-reportpath}
Let $b$ be a certifiable transaction block of a payment chain in epoch
$\ep$, acknowledged by a correct processor $p$ whose closing report
for that chain and epoch has base $c$. Then $c$
is compatible with $b$. When $c$ is a strict ancestor of $b$,
the report's suffix explicitly lists the header $\hd(b)$.
\end{lemma}

\begin{proof}
If $c$ is the initial block, compatibility follows from the
initialisation rule. Otherwise $c$ and $b$ carry $n-2f$ DA-votes
within the epoch and are compatible by Lemma~\ref{lem:pay-spine}.
Suppose $c$ is a
strict ancestor of $b$. The vote $p$ cast for $b$ chained, by
eligibility condition (iii), from a certified base $c_b$ through
sent votes for every ancestor of $b$ above $c_b$. Processor $p$ holds
$c_b$, and $c$ is the highest block for which $p$ holds a certificate,
so $c_b$'s height is at most $c$'s, and $p$'s votes at the heights
strictly between $c$ and $b$ are votes for $b$'s ancestors. By
eligibility condition (i) these are $p$'s only votes at those
heights, so the reported suffix lists the headers of the blocks on
$b$'s ancestry strictly above $c$ up to and including $b$. 
\end{proof}

{\color{black}
\begin{lemma}[Non-interference]\label{lem:pay-noninterference}
In a non-closing view, the emission sweep does not halt at a
payment-chain position. For a fixed proposal and vote pool,
deleting the payment-chain blocks from the emitted sequence leaves
exactly the sequence produced by the ordinary-chain passes with
those same ordinary-chain proposal and vote components.
\end{lemma}

\begin{proof}
A payment chain is settled in every pool. Each vote endorses its
base, there are no further proposed positions, and no vote can
endorse an extension because $e=0$. Hence $k_i^*=m_i=0$ and the
number of extending votes is zero. The payment pass completes
without a halt, and the subsequent two passes use only the ordinary
components and are unchanged.
\end{proof}
}

\begin{lemma}[Monotone validity]\label{lem:pay-monotone}
Let $b$ be a transaction block of payment chain $i$, with reference
prefix $R$, produced by a correct producer and valid at $R$. Then, within $b$'s epoch, $b$ is
valid at $\sigma_p$ for every
correct processor $p$ whose executed prefix contains $b$'s
reference, whose chain-$i$ tip in $\sigma_p$ is below $b$, and
who holds, validated, the complete chain of blocks from the
reference's chain-$i$ tip to $b$'s parent.
\end{lemma}

\begin{proof}
Conditions (1)--(3) of the
validity test in Section~\ref{subsec:pay-validity} do not mention the
prefix. For conditions (4) and (5), note that all debits of an account lie
on its home chain's one path (Lemma~\ref{lem:pay-spine}), which is
the path containing $\gamma(b,R)$, so the producer and $p$
advance the debited accounts through the same debits, in the same
order, up to the position below $b$. A prefix larger than the
reference adds, below $b$, only credits, which can neither invalidate
a sequence number nor unfund a payment. Thus $b$ remains valid at each prefix in the statement.
\end{proof}

{\color{black}
\begin{lemma}[The cut is well defined]\label{lem:pay-cut}
All finalised leader blocks of epoch $\ep+1$ carry the same $cq$,
naming one closing L-QC $Q$. For every closing L-QC $Q$ and payment
chain $i$, $\mathrm{cut}_i(Q)$ exists and is unique.
\end{lemma}

\begin{proof}
The successor's votes, novotes, nullifications and non-genesis
parents all belong to one fixed-committee consensus execution.
The one-vote, exclusivity and no-nullification arguments therefore
apply with view indices $1,2,\dots$ in that epoch. The proof of
Lemma~\ref{lem:leader-chain} then applies to its consensus parents.
In particular, after a leader block of view $v$ is L-notarised,
a later opening proposal cannot bypass it, since such a proposal
requires a successor nullification of $v$, which cannot exist. A continuing
proposal inherits successor ancestry in the usual way. Thus all
finalised successor leader blocks lie on one reference chain.
Its first block names $Q$, and every continuing block copies $cq$.
All finalised successor leader blocks consequently name that same $Q$.

Fix $Q$ and payment chain $i$, and let $u_i$ again be the chain's
tip in $\Emit(Q)$. Two strict descendants of $u_i$ each covered by
$n-3f$ of its $n-f$ reports share a covering report, since
\[
 2(n-3f)>n-f \qquad (n\ge5f+1).
\]
Both are ancestors of that report's top and are therefore comparable.
The candidate set is finite and contains $u_i$, which is comparable
with all its other members, so it has a unique member of greatest
height.
\end{proof}
}

{\color{black}
\begin{lemma}[Availability at the transition]\label{lem:pay-custody}
If an opening block is L-notarised, at least $n-2f$ correct members
of its committee hold its complete named candidate log and closing
L-QC. If an opening block is an ancestor of an L-notarised block,
at least $f+1$ correct members hold that data. So, every
selected completed log remains retrievable from the successor
committee.
\end{lemma}

\begin{proof}
The possession condition applies to every correct vote for an
opening block, and direct L-notarisation supplies $n-f$ voters,
of which at least $n-2f$ are correct. An opening block on the reference
chain of a later L-notarised block is designated by a V-QC and has
at least $2f+1$ voters, of which at least $f+1$ are correct.
These are successor-committee processors and continue to serve
retrieval requests.
\end{proof}
}

{\color{black}
\begin{lemma}[Boundary safety]\label{lem:pay-boundary}
\begin{enumerate}[label=(\alph*),itemsep=1pt]
\item For every closing L-QC $Q$ of epoch $\ep$, $\mathrm{log}_\ep(Q)$
  extends every correct processor's epoch-$\ep$ emission.
  Each block appended by its cut extends the emitted tip of its
  payment chain and has DA-votes from at least $f+1$ correct
  members of $\Pi_\ep$.
\item Every transaction block $b$ of payment chain $i$ in epoch $\ep$
  with $n-f$ unconditional or matured acknowledgments is at or below
  $\mathrm{cut}_i(Q)$
  for every closing L-QC $Q$ of that epoch, regardless of when
  its signatures are assembled or disclosed.
\item After GST, every correct successor member, and every
  correct old member that continues following the log, eventually
  adopts the selected $\mathrm{log}_\ep$.
\end{enumerate}
\end{lemma}

\begin{proof}
(a) The leader-chain and prefix lemmas put every earlier emitted
prefix below the ordering named by a closing L-QC. All subsequent
closing blocks on that reference chain have the same tips and
add nothing to $\Ord$ or $\Emit$. Since L-notarised closing
blocks lie on one such chain, their $\Emit$ values coincide.
The definition of $\mathrm{log}_\ep(Q)$ only appends to this prefix.

Fix a cut-appended block $b$. If one of its covering reports has
a certified base at or above $b$, the per-ancestor argument of
Lemma~\ref{lem:pay-content} supplies at least $f+1$ correct
DA-voters of $b$. Otherwise every report covering $b$ lists it
in its voted suffix, and at least $n-3f-f=n-4f\ge f+1$ of
those reporters are correct. The definition of the cut restricts all
appended blocks to descendants of the emitted tip.

(b) Fix such a block $b$ of chain $i$ and a closing L-QC $Q$ of
its epoch. Write $u_i$ for the chain-$i$ tip in $\Emit(Q)$.
All unconditional or matured acknowledgments of $b$ belong to
the candidate selected by $Q$'s leader block. Reports cast before
adoption use this same candidate and include its prospective votes.
Let $F$ be a set of $n-f$ acknowledgers of $b$ and let $S$ be the
$n-f$ closing voters of $Q$. At least
\[
 (n-f)+(n-f)-n-f=n-3f
\]
members of $F\cap S$ are correct. Every such processor acknowledged $b$ before
fencing and signed its report afterwards. If its base is at or
above $b$, the report covers $b$. Otherwise the Report paths
lemma ensures that its complete suffix lists $\hd(b)$. Nothing is
truncated because the DA window bounds that suffix by
$d_{\mathrm{DA}}$ above the highest block for which the voter holds
a certificate.
The block $b$ is certifiable, so it is compatible with the emitted
tip $u_i$. If it is at or below $u_i$, it is already emitted.
Otherwise these $n-3f$ reports make it a cut candidate. Cut
uniqueness then puts it at or below $\mathrm{cut}_i(Q)$.
This argument uses the acknowledgments that were sent, not
possession of their aggregate certificate.

(c) Close liveness, proved next, gives a finalised successor block
selecting $Q$. Its candidate log extends every old emitted prefix
by (a), is retrievable by Lemma~\ref{lem:pay-custody}, and is
adopted under the entry rule and the transfer assumption.
\end{proof}
}

{\color{black}
\begin{lemma}[Close liveness]\label{lem:pay-liveness}
Under the state-transfer and synchronisation assumptions stated in
Section~\ref{subsec:pay-open}, every epoch that begins eventually
has a closing L-QC, and a
leader block of the successor epoch is eventually finalised,
selecting the preceding epoch's completed log. After GST, once the correct processors have the required data, the old
committee obtains a closing L-QC within a bounded number of closing
views, and the successor committee finalises a leader block within
a bounded number of its views.
\end{lemma}

\begin{proof}
If a closing L-QC already exists at a correct processor, transfer
to the successor can begin. Otherwise the old committee continues
view progression. A subsequent closing proposal by a correct leader  has
frozen tips and supplies its candidate log where required. Even before
adoption, voters can form reports from that candidate's initial blocks
and their own DA-votes, without awaiting finality or other voters'
reports. The possession and catch-up rules preserve the base
vote-before-advance argument, so the view produces a closing L-QC.

The transfer assumption supplies the successor members with
complete candidate logs. Their view counters start at $1$ and
catch up using their own certificates. Thereafter the fixed-committee
view-progression argument gives a correct-leader view. If its leader
uses an opening or closing proposal, it sends the complete candidate
data with the proposal, meeting the possession condition on the same
message-delay bound. A continuing proposal inherits its referenced
leader block's $cq$. In a non-closing view the ordinary proposal rule
applies.
In either case the correct-leader lemma gives L-notarisation and
fixes the selected $cq$. Data for an inherited opening candidate
are available by Lemma~\ref{lem:pay-custody}. The state-transfer
bound and the base view bounds give the asserted eventuality and
bounded-view claims.
\end{proof}
}

Call a payment \emph{persistently admissible} from a time
when:
\begin{enumerate}[label=(\alph*),itemsep=1pt]
\item it is signed by its account's owner;
\item its sequence number is the account's next until it executes,
  and no conflicting payment with that sequence number is ever signed;
\item the account's balance covers it until it executes;
\item the account's home for the relevant epochs has a correct
  producer and stays unchanged until execution; and
\item after the current transition, a non-closing interval is long
  enough for the retained payment to be produced, disseminated and
  reference-validated, and for the acknowledgment round to complete
  before correct processors fence. This is a real-time admission
  condition, not merely a comparison between a number of views and
  a message delay.
\end{enumerate}

\begin{proposition}[Payment liveness]\label{prop:pay-liveness}
After GST, (a) every persistently admissible payment
submitted to the producer of its home chain is eventually executed, through placement or
through the cut of its epoch, and (b) if moreover its block is
disseminated early enough that the acknowledgment round trip and
validation complete before the fences, the producer assembles a
receipt
for it. If instead the home's producer is faulty, the payment may
be delayed indefinitely while the account remains homed there. The
owner's escape is a move, includable on an ordinary chain under
the block-inclusion guarantee, after which the payment is
admissible at the new home two epochs later, subject to the
conditions of Section~\ref{subsec:pay-costs}.
\end{proposition}

\begin{proof}
For (a), the correct home's production rule includes the payment,
oldest-first among admissible ones, in a block $b$ within bounded
time, persistence ruling out its being dropped for a consumed
sequence number. The block $b$ is valid at every correct processor whose
executed prefix contains $b$'s reference and whose chain tip is
below $b$ (Lemma~\ref{lem:pay-monotone}, with the production rule
checking exactly that prefix). After GST all correct processors
reach the reference within bounded time, fetching under the
retrieval convention where needed, and acknowledge. If $n-f$
acknowledgments are assembled, clauses (ii) and (iii) of
Theorem~\ref{thm:receipt} bound placement and execution. If the
epoch closes first, at least $n-3f$ of the closing reports cover
$b$ whenever $b$'s acknowledgers reach $n-3f$ correct processors
among the reporters, and otherwise $b$ is re-included afresh in
the next epoch, its payments still persistently admissible. Either
way the epoch-length hypothesis guarantees a full admission
window in the next epoch, and some epoch executes the payment
within bounded time. Without that hypothesis, a short-epoch
schedule whose fences fall before each recreated block's
acknowledgments can defer the payment forever, which is why the
condition is part of the admissibility definition.

For (b), the stated margin lets every correct processor acknowledge
$b$ before fencing. Even a processor that has already executed beyond
$b$'s reference uses reference validation, checking the named reference
rather than its current executed prefix. The producer therefore
assembles $n-f$ acknowledgments.
\end{proof}

{\color{black}
\begin{proposition}[Multi-epoch ordering]\label{prop:pay-generic}
For every $\ep\ge1$, the selected completed log $\mathrm{log}_\ep$
extends $\mathrm{log}_{\ep-1}$. After GST, transactions submitted
to correct ordinary-chain producers are eventually ordered, provided
that unplaced transactions are retained by a continuing correct
producer or resubmitted to a correct producer if the original producer
departs, and that the state-transfer and synchronisation assumptions
hold and a subsequent non-closing interval is long enough for
admission, dissemination and settlement.
\end{proposition}

\begin{proof}
Within an epoch, the successor leader-chain argument fixes one
initial completed log, above which the base safe-extension and
prefix lemmas apply. Across a boundary, all finalised
successor leader blocks name the same $Q$ by Lemma~\ref{lem:pay-cut},
and $\mathrm{log}_\ep(Q)$ extends every old emission by
Lemma~\ref{lem:pay-boundary}(a). Induction on epochs therefore
preserves prefix comparability.

For liveness, close liveness and transfer let the successor proceed.
A continuing correct ordinary producer retains its signed chain
and unplaced input. Under the interval hypothesis, the
fixed-committee dissemination and settlement argument completes
for that input before the next closing phase. The resulting
emitted prefix is preserved by every subsequent boundary.
If the producer departs, input absent from the completed log must
instead be resubmitted to a correct producer. The interval hypothesis
is needed because an infinite-run liveness argument does not by
itself guarantee progress within a finite epoch.
\end{proof}
}

{\color{black}
\begin{theorem}[Receipts]\label{thm:receipt}
Let $b$ be a transaction block of payment chain $i$ in epoch $\ep$
with acknowledgments from $n-f$ processors, unconditional or matured under the early-handoff
rule. Under the fault assumptions above:
\begin{enumerate}[label=(\roman*),itemsep=1pt]
\item No block incompatible with $b$ is ever included on that
  chain in an adopted log, and no payment conflicting with a
  payment of $b$ successfully executes.
\item Every candidate completed log $\mathrm{log}_\ep(Q)$ contains $b$,
  even if its acknowledgments and certificate are never broadcast.
  Under the post-GST transfer and synchronisation assumptions,
  the epoch closes and $b$ enters the adopted log.
\item If, after GST, a correct processor broadcasts an eligible
  DA-certificate for $b$ or a descendant, then $b$ is placed by
  the end of the first non-closing view of epoch $\ep$ that
  begins at least $\delta$ later and has a correct leader,
  provided such a view occurs. If the epoch closes first,
  clause (ii) applies instead.
\item When execution reaches $b$, every payment of $b$ executes
  exactly once, funded and in sequence.
\end{enumerate}
\end{theorem}

\begin{proof}
We first establish inclusion and block compatibility, then execution.

\emph{Clause (ii).} This is Lemma~\ref{lem:pay-boundary}(b):
every closing L-QC retains at least $n-3f$ correct acknowledgers
in its report set, and each of their reports covers $b$. Close liveness and
adoption supply the eventual statement. Prospective votes are
counted only in the selected context, whose history is retained
on adoption. The same fences and report argument apply to them.

\emph{Block compatibility in clause (i).} The block $b$ is
certifiable. Every transaction block emitted by $\Emit$ on its
payment chain is an ancestor of a certified tip, and so is comparable with
$b$ by Lemmas~\ref{lem:pay-content} and~\ref{lem:pay-spine}.
Every cut-appended block lies on the path to the selected cut,
which extends $b$ by clause (ii). Thus no block incompatible with
$b$ is adopted during epoch $\ep$. The next epoch starts its
lane at that cut and accepts only new blocks extending it. Applying
this argument at each later boundary excludes incompatible blocks
in all later epochs. The preceding adopted log is likewise an
initial segment of $b$'s ancestry on that chain. This reasoning
uses each committee's quorum bounds separately, never an
intersection between different committees.

\emph{Clause (iv).} Write $R$ for $b$'s reference prefix. At least
$n-2f$ correct processors validated $b$ using the
validity test in Section~\ref{subsec:pay-validity}, starting
from the state at its named reference and including every block of
$\gamma(b,R)$. That reference is a prefix of the adopted log. By
block compatibility, these chain blocks occur in that log on one
path, in the same height order, before $b$. For an account debited
in $b$, every debit after the reference and before $b$ is on this
one home chain, and the validation calculation includes all those
debits.
Interleaved blocks of other chains can only add credits to that
account. An induction through the payments checked during validation
therefore shows that each succeeds when executed, with the same sequence
number and at least the balance available during reference
validation. The same argument covers the payments of $b$ in their
internal order. The adopted log includes each block on the path
once, and sequence numbers exclude a second execution of any
payment. For a matured prospective block its candidate reference
is exactly the completed log selected at adoption, so the argument
is unchanged.

\emph{Payment conflicts in clause (i).} A conflicting payment
uses the same account and sequence number. If it successfully
executed before $b$'s payment, that number would already be consumed
in $b$'s reference state or among the home-chain payments checked
during reference validation,
contradicting successful reference validation. If ordered after
$b$'s payment, it fails execution's next-sequence check. A block
on an incompatible path is excluded by the compatibility argument.
Later epochs inherit the consumed sequence number regardless of
any change of home.

\emph{Clause (iii).} The broadcast places the certificate at all
correct members of the epoch's committee within $\delta$. In the
specified view its correct leader therefore proposes a certified
base at or above $b$. Payment chains have no proposed entries or
extensions above that base and are settled. The correct-leader
lemma gives finality, and the payment pass emits the base and its
ancestors without waiting for an ordinary-chain halt. Missing
block bodies remain subject to retrieval, as in the base protocol.
\end{proof}

The inclusion guarantee of clause (ii) holds whether or not a
certificate is broadcast. Broadcasting the certificate gives the
earlier placement guarantee of clause (iii).
A producer that withholds an assembled receipt can leave its payee
without evidence even though the payment is guaranteed inclusion.
The payee should require the receipt or collect the acknowledgments
itself, as in FastPay~\cite{baudet2020fastpay}.

The $n-2f$ acknowledgments establish certification-path
compatibility within the epoch, but not the receipt's inclusion
guarantee at the close. Their signer set intersects a closing
L-QC in at least $n-3f$ identities, only $n-4f$ of which are
guaranteed correct. This is $f$ below the cut's threshold. If a
block with these acknowledgments is excluded by a cut, some signer
in that intersection must have supplied a closing report not
covering it. Its acknowledgment and report demonstrate a violation
of the fence or report rule. Such weaker evidence is therefore
accountable, but is not an unconditional receipt as defined here.
}

\subsection{Costs}\label{subsec:pay-costs}

\paragraph{Receipt latency} Away from an epoch transition, an unconditional receipt is
held by the  producer's one round trip after dissemination. With
validation on the acknowledgment path, this is roughly $2\delta$
plus the time to check the block, provided the acknowledgers have
executed the referenced prefix and hold the path below the block.
If an acknowledger has not received the payment block, it must fetch
the missing data before acknowledging, which can increase this latency
(Section~\ref{subsec:pay-validity}).

\paragraph{Placement latency} In the good case (message
delivery within $\delta$ and the parties on the path correct), a payment block
disseminated at
time $t$ has its acknowledgments, and so its certificate, at the
producer by around $t+2\delta$. The producer's broadcast then puts the
certificate at the next leader by $t+3\delta$. The first proposal
at or after that lists the certificate as the chain's base, and a
single round of voting places the entry and its ancestry. So the
block is placed in the total order by around $t+5\delta$ when it
just catches a proposal, and around $t+7\delta$ when it just
misses one, with views running every $2\delta$. The average is
$t+6\delta$, about two message delays behind placement on ordinary
chains. Once a receipt exists, placement no longer determines whether its payment will execute.
The receipt guarantees eventual execution, while actual execution
and the use of incoming funds on another chain wait for placement. Prospective receipts also wait for their adoption
evidence. A
Byzantine leader can instead set the chain's base at a stale
certificate, which is indistinguishable from slowness and delays
the chain's placement by that leader's view. Receipts are
unaffected, and the delay ends at the first correct leader holding
the certificate, by clause (iii) of Theorem~\ref{thm:receipt}.

\paragraph{The home} The one remaining party able to delay a
payment before it is receipted is the producer of the account's home
chain. A producer that omits the account's payments delays them for
as long as the account retains that home, and can do no more than delay, since signed payments remain
valid and balances are untouched. The bounded escape is the move, which is 
includable on any chain, protected by block inclusion, and effective
two epochs later.  A committee change
adds one more case: an account whose home's producer leaves the
committee, or whose newly chosen home is absent when a move takes
effect, is frozen until a further move activates, up to two
epochs provided the move is included and the destination home is
present in its epoch. Deployments should surface announced committee changes so
owners can move ahead of them.

\paragraph{Deployment choices} Nothing in the guarantees requires
a bound on a block's size or verification cost, but a deployment
will likely impose one on payment chains all the same. Validation
sits on the DA round there, whereas ordinary-chain blocks are
checked off the critical path.

\endgroup
\section{Experiments}\label{sec:experiments}

We compare Multimmit with BlueBottle-Core~\cite{vos2025bluebottle} and Raptr~\cite{tonkikh2025raptr} using 50 validators on global and North American testbeds. The workload consists of synthetic transactions without application execution.

\subsection{Setup and methodology}\label{sec:measurement-setup}

Table~\ref{tab:setup} summarises the testbeds, the protocol
configurations, and the workload. Each testbed has 50 validators,
each on an AWS \texttt{c8g.12xlarge} instance with 48 vCPUs, 96\,GiB
RAM, and a 1{,}000\,GiB io2 volume provisioned for 16{,}000 IOPS. The
global testbed spans 13 regions across North America, Europe, Asia,
and Australia, with three or four validators per region. The North
American testbed has 17 validators in Northern Virginia, 17 in
Ohio, and 16 in Canada. The global topology matches the one used in
BlueBottle's published evaluation. All three implementations
use the same instance types and storage provisioning.

We ran all three implementations ourselves, on the same testbeds,
with the configurations described below.

\begin{table}[!htbp]
\centering\small
\begin{tabular}{@{}lll@{}}
\toprule
 & Global & North America \\
\midrule
Validators & 50 & 50 \\
Regions & 13, on four continents & N.~Virginia, Ohio, Canada \\
Validators per region & 3--4 & 16--17 \\
Instance & \multicolumn{2}{l}{AWS \texttt{c8g.12xlarge}, 48 vCPUs, 96\,GiB RAM} \\
Storage & \multicolumn{2}{l}{1{,}000\,GiB io2, 16{,}000 provisioned IOPS} \\
\midrule
Multimmit & \multicolumn{2}{l}{$K=50$ chains, $d=32$, $e=8$, $f=9$, quorum 41, 2\,s view timeout} \\
BlueBottle & \multicolumn{2}{l}{artifact \texttt{e02aeba213fa}, two leaders per round, quorum 41, 1\,s timeout} \\
Raptr & \multicolumn{2}{l}{$f=16$, quorum 34, eight subblocks, 200\,ms QC-vote timer} \\
\midrule
Transactions & \multicolumn{2}{l}{synthetic 512-byte payloads, no execution} \\
Blocks & \multicolumn{2}{l}{Multimmit batches 512 transactions into 256\,KiB producer blocks} \\
Offered rates & \multicolumn{2}{l}{50k, 100k, 250k, 500k, 1M tx/s. BlueBottle and Raptr through 500k} \\
Scenarios & \multicolumn{2}{l}{healthy, 1 crash, 9 crashes, $0.1\%$ message loss} \\
Multimmit only & \multicolumn{2}{l}{payment chains; censoring leaders at 500k tx/s; storage sync at 500k tx/s} \\
Schedule & \multicolumn{2}{l}{120--240\,s warmup, 120\,s measured cohort, up to 120\,s drain} \\
\bottomrule
\end{tabular}
\caption{Testbeds, protocol configuration, and workload.}
\label{tab:setup}
\end{table}

\paragraph{Configurations} Multimmit uses one producer chain per validator, $d=32$, $e=8$, $f=9$,
and a quorum of 41.\footnote{Multimmit is measured at source revision
\texttt{336fd5fb}. The benchmark build adds workload instrumentation and
message-loss injection. For runs with storage sync disabled, the build
bypasses file, directory, and per-write sync barriers while retaining
filesystem writes. The source is linked under `Reproducibility' below.}
Its view timeout is two seconds. Processors also apply the rule of
Section~\ref{sec:optimisations} for skipping leaders that appear
crashed. Here the rule fires when nothing has been heard from the
leader for ten seconds while at least 41 validators have been active
in that interval. The timeout is rarely exercised in these runs.
Healthy views would be unimpacted by much smaller timeouts, and views led by stopped validators
nullify on entry via the skipping rule, so a shorter timeout would
leave the results unchanged.

BlueBottle uses the published artifact at \texttt{e02aeba213fa}, with
benchmark patches.\footnote{BlueBottle's benchmark source is revision \texttt{f49ed07b}, available at \url{https://github.com/clabby/mysticeti/commits/measure}.} The published
artifact provides neither our workload nor our measurement points,
and it has no controls for message loss or fsync. The patches add
these and nothing else. They do not change consensus or block
retention.\footnote{We enable BlueBottle's consensus-only mode through both of
its switches, ``\texttt{consensus\_only}'' and
``\texttt{CONSENSUS\_ONLY}''.} BlueBottle uses two
leaders per round, a one-second leader timeout, and a quorum of 41.
It makes the same $n\geq 5f+1$ resilience assumption as Multimmit,
so the two systems run with the same fault budget and the same
quorum. Its WAL fsync is disabled throughout. Multimmit's storage sync barriers
are likewise disabled, except in the experiment of
Section~\ref{sec:measurement-durability}, which measures the cost of
enabling them.

Raptr uses the published implementation at \texttt{a4c98047cf8a}, with
benchmark instrumentation and a minimal payload-retention fix.\footnote{Raptr's
benchmark source is available at
\url{https://github.com/clabby/aptos-core/tree/measure}, revision
\texttt{233a90f9}. The patches add the
scheduled workload, finality
timestamps, traffic counters, message-loss injection, and fsync controls.
They also adjust batching and storage limits and prevent historical
payloads from retaining fetched transaction bodies. These changes do not
alter the consensus protocol.}
Raptr uses $f=16$ and a consensus quorum of 34, eight subblocks, and
a QC-vote timer of 200\,ms.\footnote{The QC-vote timer bounds how long a validator waits for
missing batch data before voting for the longest prefix it holds (this
is the parameter written $\epsilon\Delta$ in the Raptr paper).} These
are its native parameters under its $n\geq 3f+1$ resilience
assumption. Its quorum is smaller than the other two systems' quorum
of 41, which, if anything, favours its latency. The nine-crash
experiment uses the same stopped validators across implementations,
but does not exhaust Raptr's larger fault budget. Raptr forms batches of up to 450 transactions and pulls up
to 2\,250 transactions, with a sender limit of five batches. Its Quorum
Store quotas per validator are 300\,MB in memory and 1\,GB on disk. The receiver limit
is 2\,250 transactions and the dynamic rate limit is 15k transactions/s
per validator. A preload shim bypasses \texttt{fsync} and
\texttt{fdatasync} while retaining filesystem writes. All Raptr results
use this shim.

The three protocols use different batching rules, and Raptr a
different fault budget. The comparisons below therefore concern the
measured implementations, and do not isolate any one protocol feature.

\paragraph{Workload} Transactions are synthetic 512-byte payloads, with no execution.
Measuring consensus without execution is standard in this literature.\footnote{Execution time depends on the application and on parts of
an implementation that lie outside consensus, so including it would
blur comparisons.}
The evaluations of BlueBottle~\cite{vos2025bluebottle} and
Raptr~\cite{tonkikh2025raptr} both stop their clocks before
execution. We report
throughput in transactions per second (tx/s). Multimmit admits
transactions in batches of 512, which fill one 256\,KiB producer
block. BlueBottle receives submissions every 100\,ms and retains its
own block formation. Raptr also submits every 100\,ms and retains its
own batch formation. The \emph{offered rate} is the incoming transaction
rate we schedule for submission. Transactions are scheduled on a
fixed timetable, divided among the surviving producers, and the
timetable does not slow when a system falls behind. Offered rates are
50k, 100k, 250k, 500k, and 1M tx/s for Multimmit, and the same rates
through 500k for BlueBottle.\footnote{BlueBottle was also offered
1M tx/s, but exhausted memory at that rate on the global testbed with fresh processes
and consensus-only operation. We therefore omit 1M from its
measurements. The benchmark build retains the published consensus and
block-retention code, so the memory use is that of the published
implementation. The failure shows that this implementation cannot
handle the rate with 96\,GiB of RAM. It does not show that the
protocol has a throughput limit there. An implementation that manages
memory differently, or a machine with more of it, might sustain the
rate.}

Raptr uses the same offered rates through 500k. Of its 32 conditions,
22 pass the acceptance checks described below and ten are excluded for overload. These
exclusions cover all global conditions at 250k and 500k, and North
American operation with nine crashes at those rates. The global fault
conditions at these rates were not run after the healthy conditions
overloaded.

\paragraph{Faults} At each rate we test healthy operation, one crash, nine crashes, and
$0.1\%$ message loss without crashes. Crashed validators remain
crashed throughout the run, and all three implementations crash the same validators. These are one
validator in Northern Virginia for the one-crash case, and one in each
of nine regions globally or three in each North American region for
the nine-crash case. Loss is sampled independently for complete
application messages at the receiver, after TCP delivery and before
protocol processing. Discovery and connection-maintenance messages are
exempt. TCP therefore does not repair the injected loss.

Three further sets of runs concern Multimmit alone. In the payment
runs, all 50 chains are configured as payment chains
(Section~\ref{sec:payment-results}), under the same rates and the
healthy, crash and loss scenarios. In the censorship runs, at 500k tx/s, one or nine
leaders omit five designated victim chains from their proposals while continuing to vote,
and the other producer chains are unaffected by this instruction
(Section~\ref{sec:censorship-results}). Finally, one healthy run at
500k tx/s on each testbed enables storage sync
(Section~\ref{sec:measurement-durability}).

\paragraph{Schedule} Each run starts with fresh data and committee state. It begins with a
warmup of 120 seconds, extended by at most two further periods of 60
seconds until the system is stable. The warmup is followed by 120 seconds of scheduled
submissions. We call the transactions scheduled in that window the
\emph{cohort} of the run. When the window closes, some of the cohort has
not yet been placed in the total ordering. The run therefore continues without new
input for up to another 120 seconds, so that these transactions can
finish and be measured. We call this final period the \emph{drain}.
Faulty Multimmit runs end the drain early once the whole cohort has
been admitted and placed, subject to a ten-second minimum and
passing health checks. Other runs use the full 120 seconds.

\paragraph{Measurements} Latency starts at scheduled submission, before admission,
transaction block construction, or signing. For Multimmit, it ends
when the transaction's block is placed in the total ordering at the
transaction block producer, i.e., when the producer's emitted ordering
first contains the block (Section~\ref{sec:spec}). For BlueBottle
and Raptr, latency ends at first local consensus finality
at the producer, which in those protocols also fixes the transaction's
position in the ordering. Raptr records consensus finality before fetching transaction bodies.
In each case both timestamps use the producer's clock, and the
latency excludes ordered application delivery and execution. For
payment chains, the receipt and the placement of each block are recorded
separately. All three systems are thus measured from the same starting
point to the event that fixes a transaction's position in the ordering.

Note that the measurement clock differs from the
analytical clock of Section~\ref{sec:intuition}, which starts at
the dissemination of the transaction block, isolating the
protocol's rounds from batching: the measured latency additionally
includes each transaction's wait for a block to form, which is
largest at the lowest offered rates.

We count each transaction once and
weight latency statistics by transaction count, each transaction in a
Multimmit producer block inheriting the block's placement latency. \emph{Achieved throughput}
is the number of cohort transactions placed during the 120-second
submission window, divided by the length of that window. The figures plot median
latency with bars spanning P25 to P75, together with P99, against
the offered rate. The achieved throughput is used only in the
acceptance checks, described next. Both axes of the latency figures are logarithmic.

\paragraph{Acceptance} A plotted observation must complete its cohort, achieve at least
$95\%$ of the offered rate, retain the expected live validators, and
pass checks for input backlog, clock agreement (within 2\,ms) and
complete telemetry. We call a run passing these checks \emph{accepted}.
Wide latency distributions are retained when these checks pass.
Incomplete or overloaded conditions have no latency point, since excluding unfinished
transactions would bias the result. Each point represents one observed
cohort, not an estimate of run-to-run variability: the bars at a point
show the spread of transaction latencies within that cohort.

{\emergencystretch=4em\relax\sloppy
\paragraph{Reproducibility} The source code is available online, together with the experiment configuration, driver, analysis scripts and reproduction instructions.\footnote{The source code is at
\url{https://github.com/commonwarexyz/monorepo/tree/cl/multimmit-paper-measurement-2}.
The experiment configuration, driver, analysis scripts and reproduction instructions are provided at
\url{https://gist.github.com/clabby/c5d8d8233c0dba6e1f03ebdf0a12e15d}.}
A separate machine-readable numerical supplement, \texttt{multimmit-90-run-summary.csv}, contains the results of all 90 accepted Multimmit runs. The companion \texttt{multimmit-censored-view-summary.csv} reports the eight censorship conditions restricted to the selected censoring views. These aggregate summaries are available in a separate repository.\footnote{\url{https://github.com/AndrewLewis-Pye/multimmit-data}} \par}

\begin{figure}[!htbp]
\centering\begingroup
\definecolor{mmblue}{RGB}{0,114,178}
\definecolor{mmorange}{RGB}{213,94,0}
\definecolor{mmgreen}{RGB}{0,130,100}
\begin{minipage}{.49\linewidth}\centering
\begin{tikzpicture}
\begin{axis}[width=.96\linewidth,height=4.05cm,xmode=log,ymode=log,xmin=42000,xmax=1250000,
xtick={50000,100000,250000,500000,1000000},xticklabels={50k,100k,250k,500k,1M},
ymin=100,ymax=100000,
title={Global: Healthy},title style={font=\scriptsize\bfseries},xlabel={Offered rate (tx/s)},ylabel={Latency (ms)},
xlabel style={font=\tiny},ylabel style={font=\tiny},tick label style={font=\tiny},
ymajorgrids,grid style={black!10},axis line style={black!40},tick align=outside,clip=false]
\addplot[mmorange,mark=*,mark size=1.3pt,line width=.7pt,error bars/.cd,y dir=both,y explicit,error bar style={mmorange,opacity=.55,line width=.55pt},error mark options={mmorange,opacity=.55,line width=.55pt,mark size=1.1pt,rotate=90}] coordinates {(50000,596.767) += (0,136.173) -= (0,93.947) (100000,595.111) += (0,134.159) -= (0,90.781) (250000,898.056) += (0,784.314) -= (0,207.946) (500000,2652.047) += (0,1688.133) -= (0,851.847)};
\addplot[mmorange,dashed,mark=*,mark size=1.3pt,line width=.7pt] coordinates {(50000,1498.249) (100000,1879.418) (250000,4757.094) (500000,13377.168)};
\addplot[mmgreen,mark=triangle*,mark size=1.3pt,line width=.7pt,error bars/.cd,y dir=both,y explicit,error bar style={mmgreen,opacity=.55,line width=.55pt},error mark options={mmgreen,opacity=.55,line width=.55pt,mark size=1.1pt,rotate=90}] coordinates {(50000,539.888) += (0,96.252) -= (0,78.088) (100000,538.689) += (0,96.171) -= (0,78.629)};
\addplot[mmgreen,dashed,mark=triangle*,mark size=1.3pt,line width=.7pt] coordinates {(50000,905.726) (100000,909.290)};
\addplot[mmblue,mark=square*,mark size=1.3pt,line width=.7pt,error bars/.cd,y dir=both,y explicit,error bar style={mmblue,opacity=.55,line width=.55pt},error mark options={mmblue,opacity=.55,line width=.55pt,mark size=1.1pt,rotate=90}] coordinates {(50000,393) += (0,50) -= (0,53) (100000,383) += (0,54) -= (0,59) (250000,372) += (0,57) -= (0,57) (500000,389) += (0,61) -= (0,58) (1000000,411) += (0,66) -= (0,63)};
\addplot[mmblue,dashed,mark=square*,mark size=1.3pt,line width=.7pt] coordinates {(50000,546) (100000,548) (250000,543) (500000,624) (1000000,693)};\end{axis}
\end{tikzpicture}
\end{minipage}
\hfill
\begin{minipage}{.49\linewidth}\centering
\begin{tikzpicture}
\begin{axis}[width=.96\linewidth,height=4.05cm,xmode=log,ymode=log,xmin=42000,xmax=1250000,
xtick={50000,100000,250000,500000,1000000},xticklabels={50k,100k,250k,500k,1M},
ymin=10,ymax=10000,
title={North America: Healthy},title style={font=\scriptsize\bfseries},xlabel={Offered rate (tx/s)},ylabel={Latency (ms)},
xlabel style={font=\tiny},ylabel style={font=\tiny},tick label style={font=\tiny},
ymajorgrids,grid style={black!10},axis line style={black!40},tick align=outside,clip=false]
\addplot[mmorange,mark=*,mark size=1.3pt,line width=.7pt,error bars/.cd,y dir=both,y explicit,error bar style={mmorange,opacity=.55,line width=.55pt},error mark options={mmorange,opacity=.55,line width=.55pt,mark size=1.1pt,rotate=90}] coordinates {(50000,101.872) += (0,74.848) -= (0,57.822) (100000,120.691) += (0,133.719) -= (0,73.881) (250000,215.810) += (0,324.720) -= (0,128.250) (500000,279.106) += (0,189.564) -= (0,117.106)};
\addplot[mmorange,dashed,mark=*,mark size=1.3pt,line width=.7pt] coordinates {(50000,1626.503) (100000,1868.224) (250000,2286.020) (500000,1554.386)};
\addplot[mmgreen,mark=triangle*,mark size=1.3pt,line width=.7pt,error bars/.cd,y dir=both,y explicit,error bar style={mmgreen,opacity=.55,line width=.55pt},error mark options={mmgreen,opacity=.55,line width=.55pt,mark size=1.1pt,rotate=90}] coordinates {(50000,145.858) += (0,40.082) -= (0,37.038) (100000,147.885) += (0,38.275) -= (0,36.925) (250000,175.305) += (0,42.795) -= (0,31.595) (500000,243.535) += (0,47.045) -= (0,32.605)};
\addplot[mmgreen,dashed,mark=triangle*,mark size=1.3pt,line width=.7pt] coordinates {(50000,228.288) (100000,227.644) (250000,276.777) (500000,367.492)};
\addplot[mmblue,mark=square*,mark size=1.3pt,line width=.7pt,error bars/.cd,y dir=both,y explicit,error bar style={mmblue,opacity=.55,line width=.55pt},error mark options={mmblue,opacity=.55,line width=.55pt,mark size=1.1pt,rotate=90}] coordinates {(50000,55) += (0,5) -= (0,5) (100000,54) += (0,6) -= (0,4) (250000,61) += (0,4) -= (0,8) (500000,60) += (0,11) -= (0,8) (1000000,76) += (0,9) -= (0,12)};
\addplot[mmblue,dashed,mark=square*,mark size=1.3pt,line width=.7pt] coordinates {(50000,76) (100000,76) (250000,82) (500000,154) (1000000,198)};\end{axis}
\end{tikzpicture}
\end{minipage}
\par\vspace{3pt}
\begin{minipage}{.49\linewidth}\centering
\begin{tikzpicture}
\begin{axis}[width=.96\linewidth,height=4.05cm,xmode=log,ymode=log,xmin=42000,xmax=1250000,
xtick={50000,100000,250000,500000,1000000},xticklabels={50k,100k,250k,500k,1M},
ymin=100,ymax=100000,
title={Global: Nine crashes},title style={font=\scriptsize\bfseries},xlabel={Offered rate (tx/s)},ylabel={Latency (ms)},
xlabel style={font=\tiny},ylabel style={font=\tiny},tick label style={font=\tiny},
ymajorgrids,grid style={black!10},axis line style={black!40},tick align=outside,clip=false]
\addplot[mmorange,mark=*,mark size=1.3pt,line width=.7pt,error bars/.cd,y dir=both,y explicit,error bar style={mmorange,opacity=.55,line width=.55pt},error mark options={mmorange,opacity=.55,line width=.55pt,mark size=1.1pt,rotate=90}] coordinates {(50000,1035.902) += (0,128.368) -= (0,123.502) (100000,1040.392) += (0,158.498) -= (0,138.082) (250000,6990.940) += (0,3786.460) -= (0,3557.590)};
\addplot[mmorange,dashed,mark=*,mark size=1.3pt,line width=.7pt] coordinates {(50000,1588.707) (100000,3476.865) (250000,17517.578)};
\addplot[mmgreen,mark=triangle*,mark size=1.3pt,line width=.7pt,error bars/.cd,y dir=both,y explicit,error bar style={mmgreen,opacity=.55,line width=.55pt},error mark options={mmgreen,opacity=.55,line width=.55pt,mark size=1.1pt,rotate=90}] coordinates {(50000,849.140) += (0,381.540) -= (0,213.500) (100000,7834.671) += (0,893.689) -= (0,914.201)};
\addplot[mmgreen,dashed,mark=triangle*,mark size=1.3pt,line width=.7pt] coordinates {(50000,1762.510) (100000,10730.305)};
\addplot[mmblue,mark=square*,mark size=1.3pt,line width=.7pt,error bars/.cd,y dir=both,y explicit,error bar style={mmblue,opacity=.55,line width=.55pt},error mark options={mmblue,opacity=.55,line width=.55pt,mark size=1.1pt,rotate=90}] coordinates {(50000,540) += (0,140) -= (0,108) (100000,532) += (0,160) -= (0,107) (250000,552) += (0,171) -= (0,118) (500000,617) += (0,174) -= (0,135) (1000000,719) += (0,197) -= (0,157)};
\addplot[mmblue,dashed,mark=square*,mark size=1.3pt,line width=.7pt] coordinates {(50000,1119) (100000,1132) (250000,1335) (500000,1371) (1000000,1934)};\end{axis}
\end{tikzpicture}
\end{minipage}
\hfill
\begin{minipage}{.49\linewidth}\centering
\begin{tikzpicture}
\begin{axis}[width=.96\linewidth,height=4.05cm,xmode=log,ymode=log,xmin=42000,xmax=1250000,
xtick={50000,100000,250000,500000,1000000},xticklabels={50k,100k,250k,500k,1M},
ymin=35,ymax=10000,
title={North America: Nine crashes},title style={font=\scriptsize\bfseries},xlabel={Offered rate (tx/s)},ylabel={Latency (ms)},
xlabel style={font=\tiny},ylabel style={font=\tiny},tick label style={font=\tiny},
ymajorgrids,grid style={black!10},axis line style={black!40},tick align=outside,clip=false]
\addplot[mmorange,mark=*,mark size=1.3pt,line width=.7pt,error bars/.cd,y dir=both,y explicit,error bar style={mmorange,opacity=.55,line width=.55pt},error mark options={mmorange,opacity=.55,line width=.55pt,mark size=1.1pt,rotate=90}] coordinates {(50000,194.088) += (0,63.632) -= (0,79.748) (100000,233.716) += (0,62.284) -= (0,77.836) (250000,423.031) += (0,196.409) -= (0,137.411) (500000,999.283) += (0,268.567) -= (0,224.663)};
\addplot[mmorange,dashed,mark=*,mark size=1.3pt,line width=.7pt] coordinates {(50000,1182.135) (100000,1005.002) (250000,1576.819) (500000,4768.963)};
\addplot[mmgreen,mark=triangle*,mark size=1.3pt,line width=.7pt,error bars/.cd,y dir=both,y explicit,error bar style={mmgreen,opacity=.55,line width=.55pt},error mark options={mmgreen,opacity=.55,line width=.55pt,mark size=1.1pt,rotate=90}] coordinates {(50000,1243.799) += (0,684.961) -= (0,689.569) (100000,1258.363) += (0,663.517) -= (0,661.983)};
\addplot[mmgreen,dashed,mark=triangle*,mark size=1.3pt,line width=.7pt] coordinates {(50000,2602.486) (100000,2596.685)};
\addplot[mmblue,mark=square*,mark size=1.3pt,line width=.7pt,error bars/.cd,y dir=both,y explicit,error bar style={mmblue,opacity=.55,line width=.55pt},error mark options={mmblue,opacity=.55,line width=.55pt,mark size=1.1pt,rotate=90}] coordinates {(50000,57) += (0,88) -= (0,11) (100000,58) += (0,23) -= (0,12) (250000,85) += (0,77) -= (0,26) (500000,180) += (0,181) -= (0,86) (1000000,242) += (0,153) -= (0,98)};
\addplot[mmblue,dashed,mark=square*,mark size=1.3pt,line width=.7pt] coordinates {(50000,229) (100000,244) (250000,535) (500000,1847) (1000000,1250)};\end{axis}
\end{tikzpicture}
\end{minipage}
\par\vspace{3pt}
\begin{minipage}{.49\linewidth}\centering
\begin{tikzpicture}
\begin{axis}[width=.96\linewidth,height=4.05cm,xmode=log,ymode=log,xmin=42000,xmax=1250000,
xtick={50000,100000,250000,500000,1000000},xticklabels={50k,100k,250k,500k,1M},
ymin=100,ymax=100000,
title={Global: $0.1\%$ loss},title style={font=\scriptsize\bfseries},xlabel={Offered rate (tx/s)},ylabel={Latency (ms)},
xlabel style={font=\tiny},ylabel style={font=\tiny},tick label style={font=\tiny},
ymajorgrids,grid style={black!10},axis line style={black!40},tick align=outside,clip=false]
\addplot[mmorange,mark=*,mark size=1.3pt,line width=.7pt,error bars/.cd,y dir=both,y explicit,error bar style={mmorange,opacity=.55,line width=.55pt},error mark options={mmorange,opacity=.55,line width=.55pt,mark size=1.1pt,rotate=90}] coordinates {(50000,18874.246) += (0,5572.924) -= (0,4239.546) (100000,21366.878) += (0,6784.582) -= (0,5884.288)};
\addplot[mmorange,dashed,mark=*,mark size=1.3pt,line width=.7pt] coordinates {(50000,41452.748) (100000,60085.232)};
\addplot[mmgreen,mark=triangle*,mark size=1.3pt,line width=.7pt,error bars/.cd,y dir=both,y explicit,error bar style={mmgreen,opacity=.55,line width=.55pt},error mark options={mmgreen,opacity=.55,line width=.55pt,mark size=1.1pt,rotate=90}] coordinates {(50000,541.700) += (0,95.670) -= (0,80.590) (100000,545.155) += (0,95.955) -= (0,79.135)};
\addplot[mmgreen,dashed,mark=triangle*,mark size=1.3pt,line width=.7pt] coordinates {(50000,907.282) (100000,918.887)};
\addplot[mmblue,mark=square*,mark size=1.3pt,line width=.7pt,error bars/.cd,y dir=both,y explicit,error bar style={mmblue,opacity=.55,line width=.55pt},error mark options={mmblue,opacity=.55,line width=.55pt,mark size=1.1pt,rotate=90}] coordinates {(50000,383) += (0,55) -= (0,45) (100000,381) += (0,51) -= (0,59) (250000,373) += (0,58) -= (0,56) (500000,398) += (0,64) -= (0,60) (1000000,424) += (0,70) -= (0,68)};
\addplot[mmblue,dashed,mark=square*,mark size=1.3pt,line width=.7pt] coordinates {(50000,544) (100000,554) (250000,556) (500000,641) (1000000,740)};\end{axis}
\end{tikzpicture}
\end{minipage}
\hfill
\begin{minipage}{.49\linewidth}\centering
\begin{tikzpicture}
\begin{axis}[width=.96\linewidth,height=4.05cm,xmode=log,ymode=log,xmin=42000,xmax=1250000,
xtick={50000,100000,250000,500000,1000000},xticklabels={50k,100k,250k,500k,1M},
ymin=25,ymax=1000000,
title={North America: $0.1\%$ loss},title style={font=\scriptsize\bfseries},xlabel={Offered rate (tx/s)},ylabel={Latency (ms)},
xlabel style={font=\tiny},ylabel style={font=\tiny},tick label style={font=\tiny},
ymajorgrids,grid style={black!10},axis line style={black!40},tick align=outside,clip=false]
\addplot[mmorange,mark=*,mark size=1.3pt,line width=.7pt,error bars/.cd,y dir=both,y explicit,error bar style={mmorange,opacity=.55,line width=.55pt},error mark options={mmorange,opacity=.55,line width=.55pt,mark size=1.1pt,rotate=90}] coordinates {(50000,24204.272) += (0,6738.988) -= (0,5137.252) (100000,20849.987) += (0,8063.353) -= (0,6759.457)};
\addplot[mmorange,dashed,mark=*,mark size=1.3pt,line width=.7pt] coordinates {(50000,55054.270) (100000,113868.234)};
\addplot[mmgreen,mark=triangle*,mark size=1.3pt,line width=.7pt,error bars/.cd,y dir=both,y explicit,error bar style={mmgreen,opacity=.55,line width=.55pt},error mark options={mmgreen,opacity=.55,line width=.55pt,mark size=1.1pt,rotate=90}] coordinates {(50000,150.451) += (0,39.749) -= (0,36.811) (100000,148.906) += (0,39.504) -= (0,35.536) (250000,175.010) += (0,43.990) -= (0,31.660) (500000,244.226) += (0,47.414) -= (0,32.556)};
\addplot[mmgreen,dashed,mark=triangle*,mark size=1.3pt,line width=.7pt] coordinates {(50000,233.619) (100000,233.401) (250000,276.225) (500000,366.058)};
\addplot[mmblue,mark=square*,mark size=1.3pt,line width=.7pt,error bars/.cd,y dir=both,y explicit,error bar style={mmblue,opacity=.55,line width=.55pt},error mark options={mmblue,opacity=.55,line width=.55pt,mark size=1.1pt,rotate=90}] coordinates {(50000,55) += (0,5) -= (0,5) (100000,55) += (0,5) -= (0,5) (250000,60) += (0,5) -= (0,8) (500000,61) += (0,12) -= (0,8) (1000000,77) += (0,10) -= (0,12)};
\addplot[mmblue,dashed,mark=square*,mark size=1.3pt,line width=.7pt] coordinates {(50000,77) (100000,75) (250000,115) (500000,151) (1000000,224)};\end{axis}
\end{tikzpicture}
\end{minipage}
\par\smallskip {\footnotesize \tikz[baseline=-.5ex]{\draw[mmblue,line width=.7pt] (0,0)--(.62,0);\fill[mmblue] (.31,0) ++(-1.3pt,-1.3pt) rectangle ++(2.6pt,2.6pt);}\,Multimmit; \tikz[baseline=-.5ex]{\draw[mmorange,line width=.7pt] (0,0)--(.62,0);\fill[mmorange] (.31,0) circle (1.3pt);}\,BlueBottle; \tikz[baseline=-.5ex]{\draw[mmgreen,line width=.7pt] (0,0)--(.62,0);\fill[mmgreen] (.31,0) ++(0,2pt)--++(-1.8pt,-3.1pt)--++(3.6pt,0)--cycle;}\,Raptr. \tikz[baseline=-.5ex]{\draw[line width=.7pt] (0,0)--(.42,0);}\,median; \tikz[baseline=-.5ex]{\draw[dashed,line width=.7pt] (0,0)--(.42,0);}\,P99.}
\endgroup
\caption{Multimmit placement latency and first-local-finality latency for BlueBottle and Raptr, with fsync disabled. Solid lines show median with P25--P75 bars and dashed lines show P99, against offered throughput. Both axes are logarithmic. Other conditions, including one crash, are discussed in the text. Multimmit measurements include the full 50k--1M range.}
\label{fig:latency-global}\label{fig:latency-na}
\end{figure}

\begin{table}[!htbp]\centering\scriptsize
\setlength{\tabcolsep}{2.0pt}
\resizebox{\linewidth}{!}{\begin{tabular}{@{}llrrrrrrrr@{}}\toprule
& & \multicolumn{2}{c}{Multimmit} & \multicolumn{2}{c}{Payment receipt} & \multicolumn{2}{c}{BlueBottle} & \multicolumn{2}{c}{Raptr} \\
Testbed & Condition & P50 & P99 & P50 & P99 & P50 & P99 & P50 & P99 \\\midrule
Global & Healthy, 50k & 393 & 546 & 267 & 334 & 597 & 1,498 & 540 & 906 \\
Global & Healthy, 100k & 383 & 548 & 241 & 303 & 595 & 1,879 & 539 & 909 \\
Global & Healthy, 500k & 389 & 624 & 236 & 320 & 2,652 & 13,377 & -- & -- \\
Global & Healthy, 1M & 411 & 693 & 245 & 466 & -- & -- & -- & -- \\
Global & Nine crashes, 500k & 617 & 1,371 & 422 & 806 & -- & -- & -- & -- \\
Global & $0.1\%$ loss, 100k & 381 & 554 & 245 & 306 & 21,367 & 60,085 & 545 & 919 \\
North America & Healthy, 50k & 55 & 76 & 38 & 48 & 102 & 1,627 & 146 & 228 \\
North America & Healthy, 100k & 54 & 76 & 37 & 46 & 121 & 1,868 & 148 & 228 \\
North America & Healthy, 500k & 60 & 154 & 39 & 47 & 279 & 1,554 & 244 & 367 \\
North America & Healthy, 1M & 76 & 198 & 41 & 60 & -- & -- & -- & -- \\
North America & Nine crashes, 500k & 180 & 1,847 & 53 & 261 & 999 & 4,769 & -- & -- \\
North America & $0.1\%$ loss, 100k & 55 & 75 & 38 & 47 & 20,850 & 113,868 & 149 & 233 \\
\bottomrule\end{tabular}}
\caption{Latency in milliseconds. Multimmit placement and payment receipts are from the 90 accepted Multimmit runs; the BlueBottle and Raptr numbers are from our measurements of those systems (Section~\ref{sec:measurement-setup}). The receipt requires 41 acknowledgments. Dashes denote overload, insufficient quorum, or rates not offered, as described in the text.}
\label{tab:latency-summary}
\end{table}

\subsection{Placement latency}\label{sec:latency-results}
Figure~\ref{fig:latency-global} shows placement latency with P25--P75 bars and P99s; Table~\ref{tab:latency-summary} gives selected numerical comparisons. At 50k tx/s the global Multimmit median is 393\,ms (P99 546\,ms), while in North America it is 55\,ms (P99 76\,ms). Raising the offered rate to 1M changes these medians only to 411 and 76\,ms, respectively. Under healthy operation at 100k globally, BlueBottle and Raptr have medians of 595 and 539\,ms, compared with 383\,ms for Multimmit. The global BlueBottle median reaches 2,652\,ms at 500k and Raptr overloads from 250k. BlueBottle's attempted 1M run exhausted memory. These are properties of the measured implementations and workloads, not protocol throughput upper bounds.

\paragraph{Faults and tails} Under nine crashes at 500k, Multimmit's global median and P99 rise from 389 and 624\,ms to 617 and 1,371\,ms. BlueBottle fails to retain quorum at 500k under nine crashes; its median is already 6,991\,ms at 250k. In North America at 500k under nine crashes, Multimmit's median is 180\,ms and its P99 is 1,847\,ms, compared with BlueBottle's 999 and 4,769\,ms. The large difference between the median and P99, and between P25 and P75, is visible in the error bars; those bars describe transaction-to-transaction dispersion within a single accepted cohort, not variability between independent runs.

\paragraph{Message loss} With 0.1\% receiver-side message loss at 100k globally, Multimmit's median and P99 are 381 and 554\,ms, close to the healthy 383 and 548\,ms. Raptr's are 545 and 919\,ms. BlueBottle rises to a 21,367\,ms median and a 60,085\,ms P99. In North America the corresponding BlueBottle median is 20,850\,ms. Its published implementation could not sustain the higher offered rates under loss, so these conditions have no plotted point. This does not establish that BlueBottle's underlying consensus protocol violates its guarantees.

\paragraph{One crash} A single stopped validator has relatively little effect on Multimmit: at 500k globally, its median increases from 389 to 396\,ms, with P99 remaining 624\,ms. The one-crash condition is measured at all five rates but is omitted from the six-panel figure for legibility.

\subsection{Payment receipts and placement}\label{sec:payment-results}

For the payment experiments, all 50 producer chains are payment chains with $d=0$, $e=0$, and a production window of 32 blocks above the certified anchor. A receipt is obtained after acknowledgments from $n-f=41$ validators. A receipt establishes the payment guarantee before the block is placed; placement additionally waits for the certified-prefix path. The experiments exclude application execution and committee changes.

\begin{table}[!htbp]
\centering\scriptsize
\begin{tabular}{@{}lrrrr@{}}
\toprule
Condition & Receipt P50 & P99 & Placement P50 & P99 \\
\midrule
\multicolumn{5}{l}{\emph{Global}} \\
Healthy, 50k & 267 & 334 & 563 & 828 \\
Healthy, 100k & 241 & 303 & 596 & 776 \\
Healthy, 500k & 236 & 320 & 603 & 895 \\
Healthy, 1M & 245 & 466 & 658 & 1074 \\
Nine crashes, 500k & 422 & 806 & 862 & 1537 \\
$0.1\%$ loss, 100k & 245 & 306 & 592 & 780 \\
\midrule
\multicolumn{5}{l}{\emph{North America}} \\
Healthy, 50k & 38 & 48 & 87 & 102 \\
Healthy, 100k & 37 & 46 & 88 & 99 \\
Healthy, 500k & 39 & 47 & 94 & 159 \\
Healthy, 1M & 41 & 60 & 142 & 240 \\
Nine crashes, 500k & 53 & 261 & 184 & 981 \\
$0.1\%$ loss, 100k & 38 & 47 & 89 & 104 \\
\bottomrule
\end{tabular}
\caption{Payment-chain receipt and placement latency in milliseconds, with fsync disabled. All 50 chains are payment chains. Receipt requires acknowledgments from 41 of 50 validators; placement is the block\textquotesingle s position in the total order, both measured from scheduled submission at the producer.}
\label{tab:payment-results}
\end{table}

Globally, receipt medians in healthy operation are 241, 236 and 245\,ms at 100k, 500k and 1M tx/s; the corresponding placement medians are 596, 603 and 658\,ms. In North America the receipt medians are 37, 39 and 41\,ms, while placement is 88, 94 and 142\,ms. The receipt median remains below placement under nine crashes and $0.1\%$ message loss, as Table~\ref{tab:payment-results} shows.

\subsection{Censoring leaders}\label{sec:censorship-results}

We also measure the effect of censorship on block placement. In these runs one or nine validators are live and voting, but, whenever one of them leads a view, it proposes the lowest permitted anchor for five designated producer chains and omits their fresh blocks. The leader schedule is round-robin, so nine censoring validators lead nine consecutive views in every cycle of fifty. The experiments run at 500k tx/s and compare the usual extension bound $e=8$ with extensions disabled ($e=0$).

For this comparison we consider the victim blocks that, had all leaders been honest, would have been placed during a censoring view. We estimate the corresponding view for each victim block from the median view first finalising nonvictim blocks submitted within 25\,ms of it. The comparison group consists of blocks of other chains first finalised in those same censoring views. These are thus \emph{censored-view} results, rather than latency statistics over all blocks of the victim chains. The numbers of selected victim blocks are respectively 215, 260, 2,105 and 2,120 globally, and 225, 305, 2,135 and 2,160 in North America, in the order of the four conditions in Figure~\ref{fig:censorship-results}.

\begin{figure}[!htbp]\centering
\includegraphics[width=.99\linewidth]{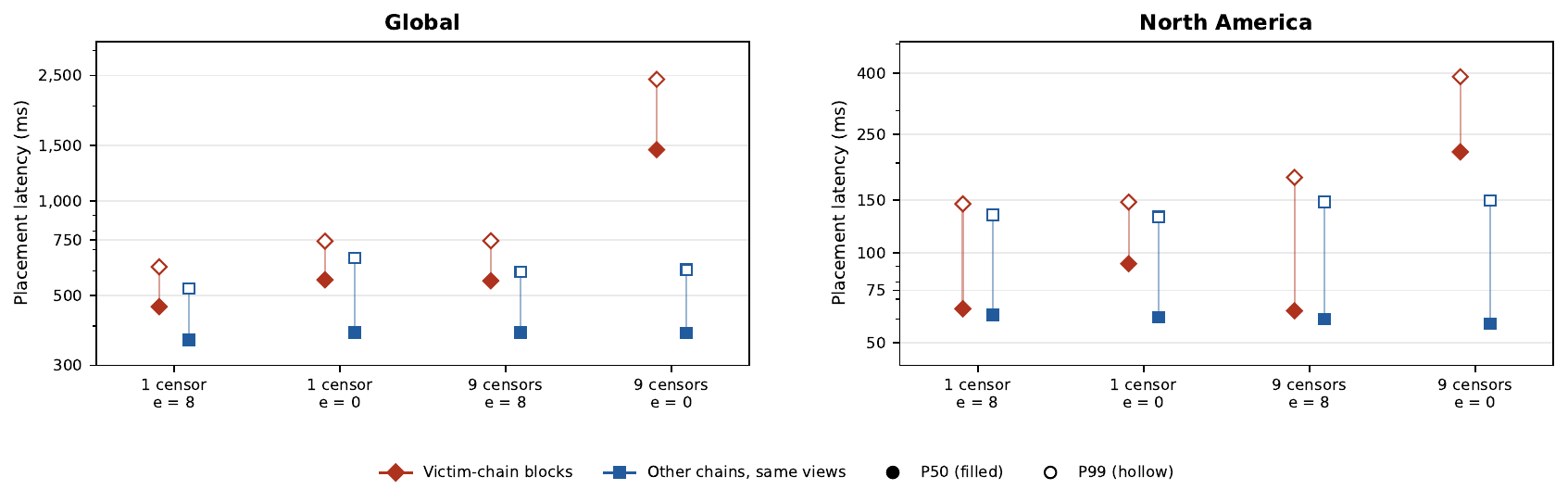}
\caption{Placement latency in the selected censoring views at 500k tx/s. Filled marks show medians and hollow marks show P99; the thin connecting segments join these two quantiles, not interquartile ranges. Victim-chain blocks are compared with blocks from other chains finalised in the same views. One or nine live validators act as censoring leaders, with extensions enabled ($e=8$) or disabled ($e=0$). The nine censoring leaders occupy consecutive views. The selection rule is explained in the text.}
\label{fig:censorship-results}
\end{figure}

\begin{table}[!t]\centering\small
\begin{tabular}{@{}llrrrr@{}}\toprule
& & \multicolumn{2}{c}{Victim blocks} & \multicolumn{2}{c}{Other chains}\\
Testbed & Censors & P50 & P99 & P50 & P99\\\midrule
Global & 1, $e=8$ & 461 & 617 & 362 & 527\\
Global & 1, $e=0$ & 561 & 744 & 382 & 659\\
Global & 9, $e=8$ & 557 & 746 & 382 & 594\\
Global & 9, $e=0$ & 1,452 & 2,429 & 380 & 604\\\addlinespace[2pt]
North America & 1, $e=8$ & 65 & 146 & 62 & 134\\
North America & 1, $e=0$ & 92 & 148 & 61 & 132\\
North America & 9, $e=8$ & 64 & 179 & 60 & 148\\
North America & 9, $e=0$ & 218 & 389 & 58 & 150\\\bottomrule
\end{tabular}
\caption{Placement latency in milliseconds for the selected censoring views, at 500k tx/s. Each entry is a P50 or P99; the number of victim blocks selected differs between conditions.}
\label{tab:censorship-new}
\end{table}

In the nine-censor global experiment, victim blocks have median and P99 placement latency 557 and 746\,ms with extensions, compared with 1,452 and 2,429\,ms without them. In North America the corresponding pairs are 64 and 179\,ms, against 218 and 389\,ms. The latency of the comparison blocks stays much closer to its healthy level in both testbeds. Of the selected victim blocks, 82\% in the global nine-censor experiment and 59\% in the North American experiment are first finalised through extensions with $e=8$. With one censor the corresponding proportions are 48\% in both testbeds. This explains why the effect of nine successive censoring leaders is substantially larger when extensions are disabled. These measurements concern the stated censoring policy and selected views, and do not measure arbitrary Byzantine behaviour.

\subsection{Communication}\label{sec:measurement-traffic}

\paragraph{Communication} Consensus messages are small relative to block-body dissemination at the measured rates. At 500k globally, the base protocol sends 12.69\,GB/s of block bodies and repair versus 281.1\,MB/s of consensus traffic (2.2\% of application bytes), summed across validators. At 1M the corresponding consensus share is 2.1\%. On the global testbed the mean base-protocol consensus certificate grows from about 4.5\,KB at 100k to 44\,KB at 1M, as more block digests are covered per view; payment-chain certificates remain near 6\,KB. With all chains configured for payments, the consensus share is 0.4\% globally at 1M, because certificates carry neither uncertified proposed paths nor extensions. The byte categories are broken down in Figure~\ref{fig:traffic-breakdown}.

\begin{figure}[!htbp]
\centering\begingroup
\definecolor{mmgrey}{RGB}{210,210,210}
\definecolor{mmda}{RGB}{173,91,65}
\definecolor{mmvotes}{RGB}{68,104,189}
\definecolor{mmcert}{RGB}{34,139,97}
\begin{minipage}{.49\linewidth}\centering
\begin{tikzpicture}
\begin{axis}[width=\linewidth,height=5.3cm,ybar stacked,bar width=5pt,ymin=0,ymax=108,xminorticks=false,xtick=data,symbolic x coords={50b,50p,100b,100p,250b,250p,500b,500p,1000b,1000p},xticklabels={50B,50P,100B,100P,250B,250P,500B,500P,1MB,1MP},tick label style={font=\tiny},ylabel={Percent of total bytes},ylabel style={font=\tiny},title={Global},title style={font=\small\bfseries},ymajorgrids,grid style={black!12},axis line style={black!40},xticklabel style={rotate=40,anchor=east}]
\addplot[draw=white,fill=mmgrey] coordinates {(50b,96.3740) (50p,96.1298) (100b,97.3938) (100p,97.2046) (250b,97.9333) (250p,98.7535) (500b,97.8327) (500p,99.3199) (1000b,97.8580) (1000p,99.6094)};
\addplot[draw=white,fill=mmda] coordinates {(50b,0.1298) (50p,0.1374) (100b,0.1351) (100p,0.1390) (250b,0.1333) (250p,0.1413) (500b,0.1334) (500p,0.1422) (1000b,0.1324) (1000p,0.1426)};
\addplot[draw=white,fill=mmvotes] coordinates {(50b,0.4046) (50p,0.2290) (100b,0.3012) (100p,0.1274) (250b,0.2450) (250p,0.0502) (500b,0.2290) (500p,0.0245) (1000b,0.2084) (1000p,0.0115)};
\addplot[draw=white,fill=mmcert] coordinates {(50b,3.0916) (50p,3.5038) (100b,2.1699) (100p,2.5290) (250b,1.6884) (250p,1.0549) (500b,1.8049) (500p,0.5134) (1000b,1.8012) (1000p,0.2365)};
\node[font=\tiny] at (axis cs:50b,103.5) {3.6\%};
\node[font=\tiny] at (axis cs:50p,103.5) {3.9\%};
\node[font=\tiny] at (axis cs:100b,103.5) {2.6\%};
\node[font=\tiny] at (axis cs:100p,103.5) {2.8\%};
\node[font=\tiny] at (axis cs:250b,103.5) {2.1\%};
\node[font=\tiny] at (axis cs:250p,103.5) {1.2\%};
\node[font=\tiny] at (axis cs:500b,103.5) {2.2\%};
\node[font=\tiny] at (axis cs:500p,103.5) {0.7\%};
\node[font=\tiny] at (axis cs:1000b,103.5) {2.1\%};
\node[font=\tiny] at (axis cs:1000p,103.5) {0.4\%};
\end{axis}
\end{tikzpicture}
\end{minipage}
\hfill
\begin{minipage}{.49\linewidth}\centering
\begin{tikzpicture}
\begin{axis}[width=\linewidth,height=5.3cm,ybar stacked,bar width=5pt,ymin=0,ymax=108,xminorticks=false,xtick=data,symbolic x coords={50b,50p,100b,100p,250b,250p,500b,500p,1000b,1000p},xticklabels={50B,50P,100B,100P,250B,250P,500B,500P,1MB,1MP},tick label style={font=\tiny},ylabel={Percent of total bytes},ylabel style={font=\tiny},title={North America},title style={font=\small\bfseries},ymajorgrids,grid style={black!12},axis line style={black!40},xticklabel style={rotate=40,anchor=east}]
\addplot[draw=white,fill=mmgrey] coordinates {(50b,83.5098) (50p,83.4803) (100b,90.3536) (100p,90.3893) (250b,95.1483) (250p,95.3012) (500b,96.8772) (500p,96.9645) (1000b,97.6479) (1000p,98.8228)};
\addplot[draw=white,fill=mmda] coordinates {(50b,0.1111) (50p,0.1184) (100b,0.1250) (100p,0.1286) (250b,0.1263) (250p,0.1335) (500b,0.1295) (500p,0.1358) (1000b,0.1296) (1000p,0.1385)};
\addplot[draw=white,fill=mmvotes] coordinates {(50b,1.7190) (50p,1.5263) (100b,1.0357) (100p,0.8214) (250b,0.5565) (250p,0.3190) (500b,0.3735) (500p,0.1540) (1000b,0.2768) (1000p,0.0484)};
\addplot[draw=white,fill=mmcert] coordinates {(50b,14.6601) (50p,14.8750) (100b,8.4857) (100p,8.6607) (250b,4.1689) (250p,4.2463) (500b,2.6198) (500p,2.7457) (1000b,1.9457) (1000p,0.9903)};
\node[font=\tiny] at (axis cs:50b,103.5) {16.5\%};
\node[font=\tiny] at (axis cs:50p,103.5) {16.5\%};
\node[font=\tiny] at (axis cs:100b,103.5) {9.6\%};
\node[font=\tiny] at (axis cs:100p,103.5) {9.6\%};
\node[font=\tiny] at (axis cs:250b,103.5) {4.9\%};
\node[font=\tiny] at (axis cs:250p,103.5) {4.7\%};
\node[font=\tiny] at (axis cs:500b,103.5) {3.1\%};
\node[font=\tiny] at (axis cs:500p,103.5) {3.0\%};
\node[font=\tiny] at (axis cs:1000b,103.5) {2.4\%};
\node[font=\tiny] at (axis cs:1000p,103.5) {1.2\%};
\end{axis}
\end{tikzpicture}
\end{minipage}
\par\smallskip
{\footnotesize \textcolor{mmgrey}{\rule{8pt}{6pt}} bodies and repair; \textcolor{mmda}{\rule{8pt}{6pt}} DA messages; \textcolor{mmvotes}{\rule{8pt}{6pt}} votes; \textcolor{mmcert}{\rule{8pt}{6pt}} certificates. B: base protocol; P: payment chains.}
\endgroup
\caption{Traffic breakdown for Multimmit in healthy operation, summed across 50 validators. Rates span 50k--1M tx/s; each pair compares the base protocol (B) with payment chains (P); the number above each bar gives consensus traffic as a percentage of all application bytes. The unshaded majority consists of block bodies and repair. The three smaller components comprise DA headers and certificates, voting messages, and consensus certificates.}
\label{fig:traffic-breakdown}
\end{figure}

\subsection{Durability}\label{sec:measurement-durability}
A coordinated upgrade or a common fault can stop every validator. Replication in memory alone does not protect against this event, so implementations must persist both signing decisions and the data whose availability they promise. Recovery assumes that durable storage survives the shutdown; these measurements do not themselves test a full-cluster restart.

Multimmit waits for the block body and custody metadata to become durable before exposing a producer header or issuing a data-availability share. It waits for the safety journal to persist each signing choice before exposing the resulting signature. This covers producer headers, availability shares, leader proposals, votes, and novote and nullify decisions used on timeout. Adjacent journal records may share a sync. Checkpoints become durable before old journal records are pruned, and finalised archives and delivery acknowledgments are synced before their durable progress is published. Forwarding an already verifiable certificate does not require another fsync.

Table~\ref{tab:durability} compares healthy 500k tx/s runs with fsync barriers enabled and disabled. Enabling sync raises the median by 8\,ms on both testbeds, from 389 to 397\,ms globally and from 60 to 68\,ms in North America. The measured P99 is slightly lower in the sync-on cohorts, which does not establish that sync improves tail latency generally.

\begin{table}[t]\centering\small
\begin{tabular}{@{}lrr@{}}\toprule
Testbed and mode & Median & P99 \\ \midrule
Global, fsync off & 389 & 624 \\
Global, fsync on & 397 & 608 \\
North America, fsync off & 60 & 154 \\
North America, fsync on & 68 & 142 \\
\bottomrule\end{tabular}
\caption{Multimmit placement latency in milliseconds at 500k tx/s in healthy operation, comparing storage sync barriers disabled and enabled.}
\label{tab:durability}\end{table}

\section{Related work}\label{sec:related}

\paragraph{Classical Byzantine consensus} The study of consensus in
the presence of Byzantine faults was introduced by Lamport, Shostak
and Pease~\cite{lamport1982byzantine}, with a treatment of SMR given
later by Schneider~\cite{schneider1990implementing}; reliable
broadcast under asynchrony is due to
Bracha~\cite{bracha1987asynchronous}. Dwork, Lynch and
Stockmeyer~\cite{DLS88} showed that $n\ge 3f+1$ is optimal for
partial synchrony. Standard protocols using this assumption, such as
PBFT~\cite{castro1999practical,castro2002practical},
Tendermint~\cite{buchman2016tendermint,buchman2018latest},
HotStuff~\cite{yin2019hotstuff} and Simplex~\cite{chan2023simplex},
satisfy 3-round finality, which is optimal at that
resilience~\cite{abraham2021good}. A large literature refines the
leader-based form, including
pipelining and proposal timing~\cite{doidge2024moonshot,
jalalzai2023fasthotstuff,kang2024hotstuff1,giridharan2021nocommit,
sui2022marlin,shoup2023sing}, communication
topology~\cite{neiheiser2021kauri,lokhava2019stellar},
multi-leader deployments~\cite{stathakopoulou2019mir},
awareness of network conditions~\cite{cohen2022aware}, and the
synchronous setting~\cite{abraham2020sync}.

\paragraph{Finality in one round of voting} A long line of
work~\cite{brasileiro2001consensus,kursawe2002optimistic,
friedman2005simple,martin2006fast,guerraoui2007refined,
song2008bosco,pass2018thunderella} considers protocols with a `fast
path', allowing quick finalisation under favourable conditions.
FaB~\cite{martin2006fast} introduced a parameterised model with
$n\ge 3f+2p+1$ processors, achieving finality in one round of voting
when at most $p$ processors are Byzantine; Kuznetsov, Tonkikh and
Zhang~\cite{kuznetsov2021revisiting} later showed that $n\ge 3f+2p-1$
is optimal (see also~\cite{abraham2017revisiting}, which identifies
bugs in several protocols of this family, including
Zyzzyva~\cite{kotla2007zyzzyva}; SBFT~\cite{gueta2019sbft} extends
the approach to tolerate crashes on the fast path). Recent protocols
at the $n\ge 5f+1$ design point include
Minimmit~\cite{chou2025minimmit}, whose consensus skeleton we
inherit, Alpenglow~\cite{alpen}, Kudzu~\cite{shoup2025kudzu},
Hydrangea~\cite{shrestha2025hydrangea,shresthahydrangeaplusplus},
Banyan~\cite{vonlanthen2024banyan} and
ChonkyBFT~\cite{francca2025chonkybft}. 
Flexible and multi-threshold
quorums~\cite{malkhi2019flexible,momose2021multi,
xiang2021strengthened} and resilience beyond optimal
bounds~\cite{lewis2025beyond} study related trade-offs between
thresholds and guarantees.

\paragraph{Multi-proposer dissemination} Narwhal~\cite{danezis2022narwhal}
demonstrated that separating dissemination from ordering, with all
processors sharing responsibility for transmitting transactions,
yields far higher throughput than passing data through leaders.
DAG-based protocols~\cite{keidar2021all,spiegelman2022bullshark,
keidar2022cordial,gkagol2018aleph,baird2016hashgraph,malkhi2023bbca,
spiegelman2024shoal,arun2024shoalpp,shrestha2025sailfish} reach
consensus on the graph structure itself, with a line of work
reducing their round-latency towards that of leader-based designs;
Mysticeti~\cite{babel2023mysticeti} does so by referencing
uncertified data, at the robustness cost measured
in~\cite{arun2024shoalpp,tonkikh2025raptr} and discussed in
Section~\ref{sec:intro}. Fides~\cite{xie2025fides} obtains
censorship resistance in this setting using trusted hardware.
Odontoceti~\cite{vandervos2025odontoceti} and
BlueBottle~\cite{vos2025bluebottle} bring the $n\ge 5f+1$ assumption
to uncertified DAGs, committing in two rounds of voting.
Closest to this paper are the multi-chain designs.
Autobahn~\cite{giridharan2024autobahn} is our starting point:
producer chains with availability certificates, cut through by
leader blocks, with certification priced into every transaction's
latency. Raptr~\cite{tonkikh2025raptr}, built on Quorum
Store~\cite{spiegelman2023quorumstore} and
Jolteon~\cite{gelashvili2022jolteon}, is the state of the art
treated at length in Sections~\ref{sec:intro}
and~\ref{sec:intuition}. Cadence~\cite{babel2026cadence},
concurrent work, runs a separate consensus instance per slot, with
several proposers per slot disseminating proposals independently;
validators vote yes or no on each proposal's arrival, then
commit-vote on a deterministic merge of the accepted proposals, in
three communication rounds at $n\ge 3f+1$, with a short-term
censorship-resistance guarantee. An offline proposer costs only its
own contribution, but partial dissemination or equivocation by one
proposer sends the whole slot to a fallback path; Multimmit's
per-chain votes confine equivocation too, not only silence, to the
offending producer's own chain. Cadence also assumes synchronised
clocks, and, its proposals being mutually hidden and merged by fee
order, forgoes \emph{predictable validity} in the sense of
Gatling~\cite{scaffino2026gatling}: a proposer cannot validate
transactions against the state at which they will execute.

\paragraph{Consensusless and fast-path payments}
FastPay~\cite{baudet2020fastpay} showed that payments need no
consensus at all: a quorum of authorities certifies a transfer in
one round trip, Byzantine consistent broadcast replacing
agreement, and Astro~\cite{collins2020online}, the consensus
number of a cryptocurrency~\cite{guerraoui2019consensus} and
ABC~\cite{sliwinski2019abc} develop the same observation. Sui
Lutris~\cite{blackshear2024sui} and
Mysticeti~\cite{babel2023mysticeti} combine such a fast-path
certificate with a consensus path for general objects, and face
exactly the problem of Section~\ref{subsec:pay-close}: carrying
fast-path certificates safely across an epoch change, which Sui
Lutris resolves with end-of-epoch checkpoints. Four guarantees are worth separating in such systems: finality,
that no conflicting operation can ever take effect; settlement,
that the operation's effects are usable, spendable funds included;
certificate availability, who must hold or publish what for the
guarantee to be exercised; and epoch survival, whether the
guarantee outlives the committee that issued it. Sui Lutris
distinguishes the first two explicitly and buys the fourth with
end-of-epoch checkpoints. Multimmit's receipt gives finality and
epoch survival without requiring publication of the receipt or
DA-certificate (Theorem~\ref{thm:receipt}(ii)); settlement onward from
other chains follows placement, a few message delays behind; and
the fast path is not a second protocol beside consensus but the
ordering protocol's own data-availability layer, so there are no
object locks, no reconciliation of two paths, and nothing to
unlock at an epoch boundary beyond the close itself. The receipt quorum is $n-f$ of $n\ge 5f+1$ rather than
FastPay's $2f+1$ of $3f+1$ because a receipt here buys more than
exclusivity: it must survive the close, and $n-f$ acknowledgers
meet $n-f$ reporters in $n-3f$ correct processors, the cut's
threshold, while $n-2f$ acknowledgments already give the
FastPay-style exclusivity within an epoch
(Section~\ref{subsec:pay-verification}).

\paragraph{Prefix consensus} The authors of Raptr have since
formalised prefix voting as a primitive. \emph{Prefix
consensus}~\cite{xiang2026prefix} has each party output two
vectors, one to commit and one to extend, such that every
committed vector, at any party, is a prefix of every extendable
one --- the relationship that Lemma~\ref{lem:safe-extend}
establishes here between finalised tips and safe-to-extend tips. They prove
tight round bounds for the primitive in the asynchronous,
leaderless setting: no two-exchange protocol exists when $n\le
4f$, three exchanges suffice at $n=3f+1$, and an appendix gives a
two-exchange protocol at $n\ge 5f+1$. Atop the primitive they
build a leaderless protocol in which each Byzantine party censors
at most one slot after GST, with a commit latency of four
exchanges from synchronised slot starts and three without
failures. Multimmit makes a related split per chain at $n\ge5f+1$:
its single voting round determines finalised and safe-to-extend tips,
with in-view placement subject to the emission sweep. Its inclusion
guarantee covers fresh, well-circulated blocks of correct producers
(Theorem~\ref{thm:inclusion}). Section~\ref{sec:oneround} gives the
sharp $3f+1$ finalisation and $f+1$ recovery thresholds for its
specified threshold prefix rules, and separately shows why Raptr's
second round cannot simply be deleted while retaining its first-round
and timeout rules. Their two-exchange construction uses $n-2f$
first-round supporters, also $3f+1$ at $n=5f+1$.

\paragraph{Inclusion guarantees} AMP~\cite{cason2026amp} runs
dedicated proposers over Tendermint at $n\ge 3f+1$, with validators
attesting the payloads they have received in \emph{vote extensions}
attached to their consensus votes. The mechanism resembles our
extension votes, but with significant differences. In AMP, the
attestations cast in one height constrain only the \emph{next}
height's proposal: a payload attested by all correct validators at
height $h$ must be included at height $h+1$, and a payload reaching
fewer correct validators may be omitted. Fresh data is thus never
finalised in the height whose votes attest it, and the good-case
path from payload dissemination to finalisation is five message
delays. In Multimmit, extension votes finalise fresh blocks within
the view whose votes endorse them, two message delays after
dissemination at best, and the block inclusion theorem settles
ledger membership in-view even when in-view placement fails.
Abraham, Efron and Ren~\cite{abraham2025latency} prove a two-round
latency cost inherent to inclusion guarantees of AMP's
next-block form; Multimmit's guarantee is shaped differently,
settling membership within the view, with a place in the total
ordering arriving immediately in the good case, or else via the next finalised leader block.
Hermes~\cite{ranchalpedrosa2026hermes}, concurrent work, shares
Minimmit's quorum sizes and targets the tail rather than censorship:
every processor proposes at view entry, votes require comparability
rather than equality, and the view of a failed leader itself
finalises the heaviest common prefix of $n-f$ fallback votes, where
Multimmit nullifies and skips. Hermes's votes are for a single value, so the
common prefix excludes the fresh suffix on which honest proposals
differ, and one lagging vote truncates what a view finalises;
Multimmit's votes are per chain, confining each lag to its own
chain's coordinate. The availability machinery of Hermes covers
consensus values, with payload availability inherited from the
dissemination layer beneath, while Multimmit's votes attest
per-chain possession, making availability a theorem
(Section~\ref{sec:verification}).

\paragraph{Bandwidth and models} Cachin and
Tessaro~\cite{cachin2005asynchronous} introduced verifiable
information dispersal, combining erasure codes with commitments so
that data can be recovered from any sufficiently large set of
verifiable fragments; see
also~\cite{alhaddad2021succinct,yang2022dispersedledger}. The Pipes
framework~\cite{lewis2025pipes} analyses how the data expansion
rate of erasure coding governs throughput in leader-based
protocols, and Carnot~\cite{lewispye2026carnot} establishes the
achievable rates at 2-round and 3-round finality. Multimmit removes
the leader bottleneck by different means, dissemination travelling
once from each producer to all, and Section~\ref{sec:final}
discusses the regimes in which the two approaches meet. Finality
gadgets and dynamic
participation~\cite{buterin2017casper,neu2021ebb,
lewis2023permissionless,lewis2024lumiere,lewis2025morpheus} treat
settings beyond the fixed-committee model assumed here. Our
cryptographic components are
standard~\cite{boneh2001short,boneh2003aggregate,shoup2000practical,
abspoel2020malicious}.
\section{Final comments}\label{sec:final}

We have presented Multimmit, established consistency, liveness and
the availability, inclusion and localisation guarantees, and
quantified its latency. We highlight some questions that the paper
leaves open.

\paragraph{Limitations} Three deployment matters sit outside the
analysis. First, resource bounds: as specified, $\ProduceNext$
places no bound on a transaction block's size, and correct
processors DA-vote, store and serve every certified block. A
deployment must bound block size and admission at the mempool
layer, and the analysis is indifferent to any such bound, since
liveness requires only that every pending transaction is eventually
included. Second, worst-case consensus data: the
tens-of-kilobytes-per-view object sizes of Section~\ref{sec:sizes} are
the good case, and maximally deviant votes inflate QCs to roughly
half a megabyte per view at $K=100$ (Section~\ref{sec:sizes}), a
bandwidth exposure under $f$ Byzantine voters that implementations
should provision for. Third, setup: the threshold schemes require
distributed key generation, and the validator set is fixed
within an epoch (Section~\ref{sec:payments} adds committee changes
between epochs); Section~\ref{sec:optimisations} describes an
instantiation without threshold schemes, under which a membership
change reduces to an update of a key list.

\paragraph{Accountability} The classification of misbehaviour in
Section~\ref{sec:verification} stops short of a formal
treatment. Formalising the guarantee, in the style of the
literature on BFT forensics, and specifying the challenge mechanism
and blame rules sketched in Section~\ref{sec:optimisations}, would
put the incentive-level discussion of greedy leaders on the same
footing as the protocol's other guarantees.

\paragraph{Syncing} Replaying the ordering from a trusted
leader block requires the historical V-QCs; an extra commitment
field in leader blocks, binding the sequence of safe-to-extend
tips, would let a syncing processor verify the history without
fetching them, and we leave the design as an engineering option.

\paragraph{Bandwidth} Multimmit removes the leader's bandwidth as a
throughput bottleneck, since transaction data travels only from each
producer to all. When the number of chains is comparable to $n$, the
load is balanced and each processor receives each byte of transaction
data once, so there is no analogue of the data expansion that erasure
coding exists to remove in the leader-based
setting~\cite{lewispye2026carnot}. When the number of chains is
small, each producer's outgoing bandwidth bounds the rate at which
its own chain can grow, and the trade-offs begin to resemble those of
the leader-based setting; we leave the comparison of the two regimes,
along with the further optimisations listed in
Section~\ref{sec:optimisations}, for future work. Finally, we have
assumed for simplicity that the block producers are the validators,
and the intended deployments separate the two roles; the
generalisation, with $K$ producer chains for $K$ unrelated to $n$,
changes nothing of substance in the analysis, apart from the
exact extension threshold of Section~\ref{sec:exact-threshold},
whose even-case form for external producers is noted there.

\section{Generative AI usage}\label{sec:generative-ai-usage}
All ideas and novel techniques used in this paper were human generated. Generative AI was used for editorial purposes in this manuscript, and all outputs were inspected by the authors to ensure accuracy and originality. In particular, generative AI was used to: (a) produce preliminary drafts of proofs for some lemmas formulated by the authors, which were subsequently checked by the authors; (b) help identify errors and inconsistencies in draft arguments, with suggested corrections independently assessed by the authors; and (c) rephrase parts of the text to improve clarity of exposition.


\begin{thebibliography}{10}

\bibitem{abraham2025latency}
Ittai Abraham, Yuval Efron, and Ling Ren.
\newblock The latency cost of censorship resistance.
\newblock Cryptology ePrint Archive, Paper 2025/2136, 2025.

\bibitem{abraham2017revisiting}
Ittai Abraham, Guy Gueta, Dahlia Malkhi, Lorenzo Alvisi, Rama Kotla, and
  Jean-Philippe Martin.
\newblock Revisiting fast practical byzantine fault tolerance.
\newblock {\em arXiv preprint arXiv:1712.01367}, 2017.

\bibitem{abraham2020sync}
Ittai Abraham, Dahlia Malkhi, Kartik Nayak, Ling Ren, and Maofan Yin.
\newblock Sync hotstuff: Simple and practical synchronous state machine
  replication.
\newblock In {\em 2020 IEEE Symposium on Security and Privacy (SP)}, pages
  106--118. IEEE, 2020.

\bibitem{abraham2021good}
Ittai Abraham, Kartik Nayak, Ling Ren, and Zhuolun Xiang.
\newblock Good-case latency of byzantine broadcast: A complete categorization.
\newblock In {\em Proceedings of the 2021 ACM Symposium on Principles of
  Distributed Computing}, pages 331--341, 2021.

\bibitem{abspoel2020malicious}
Mark Abspoel, Thomas Attema, and Matthieu Rambaud.
\newblock Malicious security comes for free in consensus with leaders.
\newblock Cryptology ePrint Archive, 2020.

\bibitem{alhaddad2021succinct}
Nicolas Alhaddad, Sisi Duan, Mayank Varia, and Haibin Zhang.
\newblock Succinct erasure coding proof systems.
\newblock {\em Cryptology ePrint Archive}, 2021.

\bibitem{arun2024shoalpp}
Balaji Arun, Zekun Li, Florian Suri-Payer, Sourav Das, and Alexander
  Spiegelman.
\newblock Shoal++: High throughput {DAG BFT} can be fast!
\newblock {\em arXiv preprint arXiv:2405.20488}, 2024.

\bibitem{baudet2020fastpay}
Mathieu Baudet, George Danezis, and Alberto Sonnino.
\newblock Fastpay: High-performance byzantine fault tolerant settlement.
\newblock In {\em Proceedings of the 2nd ACM Conference on Advances in
  Financial Technologies}, 2020.

\bibitem{babel2023mysticeti}
Kushal Babel, Andrey Chursin, George Danezis, Anastasios Kichidis, Lefteris
  Kokoris-Kogias, Arun Koshy, Alberto Sonnino, and Mingwei Tian.
\newblock Mysticeti: Reaching the limits of latency with uncertified dags.
\newblock {\em arXiv preprint arXiv:2310.14821}, 2023.

\bibitem{babel2026cadence}
Kushal Babel, Fatima Elsheimy, Lioba Heimbach, Mohammad~Mussadiq Jalalzai,
  Tobias Klenze, Jovan Komatovic, Jason Milionis, Mike Setrin, and Victor
  Shoup.
\newblock Cadence: Extreme pipelining with multiple concurrent proposers.
\newblock {\em arXiv preprint arXiv:2607.02275}, 2026.

\bibitem{baird2016hashgraph}
Leemon Baird.
\newblock The {Swirlds} hashgraph consensus algorithm: Fair, fast, {Byzantine}
  fault tolerance.
\newblock Technical report, Swirlds Tech Report SWIRLDS-TR-2016-01, 2016.

\bibitem{boneh2003aggregate}
Dan Boneh, Craig Gentry, Ben Lynn, and Hovav Shacham.
\newblock Aggregate and verifiably encrypted signatures from bilinear maps.
\newblock In {\em Advances in Cryptology---EUROCRYPT 2003}. Springer, 2003.

\bibitem{boneh2001short}
Dan Boneh, Ben Lynn, and Hovav Shacham.
\newblock Short signatures from the weil pairing.
\newblock In {\em International conference on the theory and application of
  cryptology and information security}, pages 514--532. Springer, 2001.

\bibitem{bracha1987asynchronous}
Gabriel Bracha.
\newblock Asynchronous {Byzantine} agreement protocols.
\newblock {\em Information and Computation}, 75(2):130--143, 1987.

\bibitem{brasileiro2001consensus}
Francisco Brasileiro, Fab{\'\i}ola Greve, Achour Most{\'e}faoui, and Michel
  Raynal.
\newblock Consensus in one communication step.
\newblock In {\em International Conference on Parallel Computing Technologies},
  pages 42--50. Springer, 2001.

\bibitem{buchman2016tendermint}
Ethan Buchman.
\newblock {\em Tendermint: Byzantine fault tolerance in the age of
  blockchains}.
\newblock PhD thesis, University of Guelph, 2016.

\bibitem{buchman2018latest}
Ethan Buchman, Jae Kwon, and Zarko Milosevic.
\newblock The latest gossip on bft consensus.
\newblock {\em arXiv preprint arXiv:1807.04938}, 2018.

\bibitem{buterin2017casper}
Vitalik Buterin and Virgil Griffith.
\newblock Casper the friendly finality gadget.
\newblock {\em arXiv preprint arXiv:1710.09437}, 2017.

\bibitem{cachin2005asynchronous}
Christian Cachin and Stefano Tessaro.
\newblock Asynchronous verifiable information dispersal.
\newblock In {\em 24th IEEE Symposium on Reliable Distributed Systems
  (SRDS'05)}, pages 191--201. IEEE, 2005.

\bibitem{blackshear2024sui}
Sam Blackshear, Andrey Chursin, George Danezis, Anastasios Kichidis, Lefteris
  Kokoris-Kogias, Xun Li, Mark Logan, Ashok Menon, Todd Nowacki, Alberto
  Sonnino, Brandon Williams, and Lu Zhang.
\newblock Sui lutris: A blockchain combining broadcast and consensus.
\newblock In {\em Proceedings of the 2024 ACM SIGSAC Conference on Computer
  and Communications Security}, 2024.

\bibitem{cason2026amp}
Daniel Cason, Gordon Liao, Sergio Mena, Nenad Milo{\v{s}}evi{\'c}, Adi
  Seredinschi, Alessandro Sforzin, Jo{\~a}o Sousa, and Preston Vander~Vos.
\newblock {AMP}: Arc multi-proposer protocol with bounded inclusion guarantees.
\newblock {\em arXiv preprint arXiv:2605.23677}, 2026.

\bibitem{castro1999practical}
Miguel Castro and Barbara Liskov.
\newblock Practical byzantine fault tolerance.
\newblock In {\em OSDI}, volume~99, pages 173--186, 1999.

\bibitem{castro2002practical}
Miguel Castro and Barbara Liskov.
\newblock Practical byzantine fault tolerance and proactive recovery.
\newblock {\em ACM Transactions on Computer Systems (TOCS)}, 20(4):398--461,
  2002.

\bibitem{chan2023simplex}
Benjamin~Y Chan and Rafael Pass.
\newblock Simplex consensus: A simple and fast consensus protocol.
\newblock In {\em Theory of Cryptography Conference}, pages 452--479. Springer,
  2023.

\bibitem{chou2025minimmit}
Brendan~Kobayashi Chou, Andrew Lewis-Pye, and Patrick O'Grady.
\newblock Minimmit: Fast finality with even faster blocks.
\newblock {\em arXiv preprint arXiv:2508.10862}, 2025.

\bibitem{cohen2022aware}
Shir Cohen, Rati Gelashvili, Lefteris Kokoris-Kogias, Zekun Li, Dahlia Malkhi,
  Alberto Sonnino, and Alexander Spiegelman.
\newblock Be aware of your leaders.
\newblock In {\em International Conference on Financial Cryptography and Data
  Security}. Springer, 2022.

\bibitem{collins2020online}
Daniel Collins, Rachid Guerraoui, Jovan Komatovic, Petr Kuznetsov, Matteo
  Monti, Matej Pavlovic, Yvonne-Anne Pignolet, Dragos-Adrian Seredinschi,
  Andrei Tonkikh, and Athanasios Xygkis.
\newblock Online payments by merely broadcasting messages.
\newblock In {\em 2020 50th Annual IEEE/IFIP International Conference on
  Dependable Systems and Networks (DSN)}, 2020.

\bibitem{danezis2022narwhal}
George Danezis, Lefteris Kokoris-Kogias, Alberto Sonnino, and Alexander
  Spiegelman.
\newblock Narwhal and tusk: a dag-based mempool and efficient bft consensus.
\newblock In {\em Proceedings of the Seventeenth European Conference on
  Computer Systems}, pages 34--50, 2022.

\bibitem{doidge2024moonshot}
Isaac Doidge, Raghavendra Ramesh, Nibesh Shrestha, and Joshua Tobkin.
\newblock Moonshot: Optimizing chain-based rotating leader bft via optimistic
  proposals.
\newblock {\em arXiv preprint arXiv:2401.01791}, 2024.

\bibitem{DLS88}Cynthia Dwork, Nancy~A. Lynch, and Larry Stockmeyer.
\newblock Consensus in the presence of partial synchrony.
\newblock {\em Journal of the ACM}, 35(2):288--323, 1988.

\bibitem{francca2025chonkybft}
Bruno Fran{\c{c}}a, Denis Kolegov, Igor Konnov, and Grzegorz Prusak.
\newblock Chonkybft: Consensus protocol of zksync.
\newblock {\em arXiv preprint arXiv:2503.15380}, 2025.

\bibitem{friedman2005simple}
Roy Friedman, Achour Mostefaoui, and Michel Raynal.
\newblock Simple and efficient oracle-based consensus protocols for
  asynchronous byzantine systems.
\newblock {\em IEEE Transactions on Dependable and Secure Computing},
  2(1):46--56, 2005.

\bibitem{gkagol2018aleph}
Adam G{\k{a}}gol and Micha{\l} {\'S}wi{\k{e}}tek.
\newblock Aleph: A leaderless, asynchronous, byzantine fault tolerant consensus
  protocol.
\newblock {\em arXiv preprint arXiv:1810.05256}, 2018.

\bibitem{gelashvili2022jolteon}
Rati Gelashvili, Lefteris Kokoris-Kogias, Alberto Sonnino, Alexander
  Spiegelman, and Zhuolun Xiang.
\newblock Jolteon and {Ditto}: Network-adaptive efficient consensus with
  asynchronous fallback.
\newblock In {\em International Conference on Financial Cryptography and Data
  Security}. Springer, 2022.

\bibitem{giridharan2021nocommit}
Neil Giridharan, Heidi Howard, Ittai Abraham, Natacha Crooks, and Alin Tomescu.
\newblock No-commit proofs: Defeating livelock in {BFT}.
\newblock Cryptology ePrint Archive, 2021.

\bibitem{guerraoui2019consensus}
Rachid Guerraoui, Petr Kuznetsov, Matteo Monti, Matej Pavlovi{\v c}, and
  Dragos-Adrian Seredinschi.
\newblock The consensus number of a cryptocurrency.
\newblock In {\em Proceedings of the 2019 ACM Symposium on Principles of
  Distributed Computing}, 2019.

\bibitem{giridharan2024autobahn}
Neil Giridharan, Florian Suri-Payer, Ittai Abraham, Lorenzo Alvisi, and Natacha
  Crooks.
\newblock Autobahn: Seamless high speed bft.
\newblock In {\em Proceedings of the ACM SIGOPS 30th Symposium on Operating
  Systems Principles}, pages 1--23, 2024.

\bibitem{guerraoui2007refined}
Rachid Guerraoui and Marko Vukoli{\'c}.
\newblock Refined quorum systems.
\newblock In {\em Proceedings of the twenty-sixth annual ACM symposium on
  Principles of distributed computing}, pages 119--128, 2007.

\bibitem{gueta2019sbft}
Guy~Golan Gueta, Ittai Abraham, Shelly Grossman, Dahlia Malkhi, Benny Pinkas,
  Michael Reiter, Dragos-Adrian Seredinschi, Orr Tamir, and Alin Tomescu.
\newblock Sbft: A scalable and decentralized trust infrastructure.
\newblock In {\em 2019 49th Annual IEEE/IFIP international conference on
  dependable systems and networks (DSN)}, pages 568--580. IEEE, 2019.

\bibitem{jalalzai2023fasthotstuff}
Mohammad~M Jalalzai, Jianyu Niu, Chen Feng, and Fangyu Gai.
\newblock Fast-{HotStuff}: A fast and robust {BFT} protocol for blockchains.
\newblock {\em IEEE Transactions on Dependable and Secure Computing}, 2023.

\bibitem{kang2024hotstuff1}
Dakai Kang, Suyash Gupta, Dahlia Malkhi, and Mohammad Sadoghi.
\newblock {HotStuff}-1: Linear consensus with one-phase speculation.
\newblock {\em arXiv preprint arXiv:2408.04728}, 2024.

\bibitem{keidar2021all}
Idit Keidar, Eleftherios Kokoris-Kogias, Oded Naor, and Alexander Spiegelman.
\newblock All you need is dag.
\newblock In {\em Proceedings of the 2021 ACM Symposium on Principles of
  Distributed Computing}, pages 165--175, 2021.

\bibitem{keidar2022cordial}
Idit Keidar, Oded Naor, Ouri Poupko, and Ehud Shapiro.
\newblock Cordial miners: Fast and efficient consensus for every eventuality.
\newblock {\em arXiv preprint arXiv:2205.09174}, 2022.

\bibitem{alpen}
Quentin Kniep, Jakub Sliwinski, and Roger Wattenhofer.
\newblock Solana alpenglow consensus.
\newblock {\em
  \url{https://www.scribd.com/document/895233790/Solana-Alpenglow-White-Paper}},
  2025.

\bibitem{kotla2007zyzzyva}
Ramakrishna Kotla, Lorenzo Alvisi, Mike Dahlin, Allen Clement, and Edmund Wong.
\newblock Zyzzyva: speculative byzantine fault tolerance.
\newblock In {\em Proceedings of twenty-first ACM SIGOPS symposium on Operating
  systems principles}, pages 45--58, 2007.

\bibitem{kursawe2002optimistic}
Klaus Kursawe.
\newblock Optimistic byzantine agreement.
\newblock In {\em 21st IEEE Symposium on Reliable Distributed Systems, 2002.
  Proceedings.}, pages 262--267. IEEE, 2002.

\bibitem{komatovic2025pos}
Jovan Komatovic, Andrew Lewis-Pye, Joachim Neu, Tim Roughgarden, and Ertem
  Nusret Tas.
\newblock From permissioned to proof-of-stake consensus.
\newblock In {\em Proceedings of the 7th Conference on Advances in Financial
  Technologies}, 2025.

\bibitem{kuznetsov2021revisiting}
Petr Kuznetsov, Andrei Tonkikh, and Yan~X Zhang.
\newblock Revisiting optimal resilience of fast byzantine consensus.
\newblock In {\em Proceedings of the 2021 ACM Symposium on Principles of
  Distributed Computing}, pages 343--353, 2021.

\bibitem{lamport1982byzantine}
Leslie Lamport, Robert Shostak, and Marshall Pease.
\newblock The byzantine generals problem.
\newblock {\em ACM Transactions on Programming Languages and Systems (TOPLAS)},
  4(3):382--401, 1982.

\bibitem{lewis2024lumiere}
Andrew Lewis-Pye, Dahlia Malkhi, Oded Naor, and Kartik Nayak.
\newblock Lumiere: Making optimal bft for partial synchrony practical.
\newblock In {\em Proceedings of the 43rd ACM Symposium on Principles of
  Distributed Computing}, pages 135--144, 2024.

\bibitem{lewis2025pipes}
Andrew Lewis-Pye, Kartik Nayak, and Nibesh Shrestha.
\newblock The pipes model for latency analysis.
\newblock {\em Cryptology ePrint Archive}, 2025.

\bibitem{lewispye2026carnot}
Andrew Lewis-Pye and Patrick O'Grady.
\newblock The {Carnot} bound: Limits and possibilities for bandwidth-efficient
  consensus.
\newblock {\em arXiv preprint arXiv:2603.11797}, 2026.

\bibitem{lewis2023permissionless}
Andrew Lewis-Pye and Tim Roughgarden.
\newblock Permissionless consensus.
\newblock {\em arXiv preprint arXiv:2304.14701}, 2023.

\bibitem{lewis2025beyond}
Andrew Lewis-Pye and Tim Roughgarden.
\newblock Beyond optimal fault tolerance.
\newblock {\em arXiv preprint arXiv:2501.06044}, 2025.

\bibitem{lewis2025morpheus}
Andrew Lewis-Pye and Ehud Shapiro.
\newblock Morpheus consensus: Excelling on trails and autobahns.
\newblock {\em arXiv preprint arXiv:2502.08465}, 2025.

\bibitem{lokhava2019stellar}
Marta Lokhava, Giuliano Losa, David Mazi{\`e}res, Graydon Hoare, Nicolas Barry,
  Eli Gafni, Jonathan Jove, Rafa{\l} Malinowsky, and Jed McCaleb.
\newblock Fast and secure global payments with {Stellar}.
\newblock In {\em Proceedings of the 27th ACM Symposium on Operating Systems
  Principles}, 2019.

\bibitem{malkhi2019flexible}
Dahlia Malkhi, Kartik Nayak, and Ling Ren.
\newblock Flexible {Byzantine} fault tolerance.
\newblock In {\em Proceedings of the 2019 ACM SIGSAC Conference on Computer and
  Communications Security}, 2019.

\bibitem{malkhi2023bbca}
Dahlia Malkhi, Chrysoula Stathakopoulou, and Maofan Yin.
\newblock {BBCA}-chain: One-message, low latency {BFT} consensus on a {DAG}.
\newblock {\em arXiv preprint arXiv:2310.06335}, 2023.

\bibitem{martin2006fast}
J-P Martin and Lorenzo Alvisi.
\newblock Fast byzantine consensus.
\newblock {\em IEEE Transactions on Dependable and Secure Computing},
  3(3):202--215, 2006.

\bibitem{momose2021multi}
Atsuki Momose and Ling Ren.
\newblock Multi-threshold {Byzantine} fault tolerance.
\newblock In {\em Proceedings of the 2021 ACM SIGSAC Conference on Computer and
  Communications Security}, 2021.

\bibitem{neiheiser2021kauri}
Ray Neiheiser, Miguel Matos, and Lu{\'\i}s Rodrigues.
\newblock Kauri: Scalable {BFT} consensus with pipelined tree-based
  dissemination and aggregation.
\newblock In {\em Proceedings of the ACM SIGOPS 28th Symposium on Operating
  Systems Principles}, 2021.

\bibitem{neu2021ebb}
Joachim Neu, Ertem~Nusret Tas, and David Tse.
\newblock Ebb-and-flow protocols: A resolution of the availability-finality
  dilemma.
\newblock In {\em 2021 IEEE Symposium on Security and Privacy (SP)}, pages
  446--465. IEEE, 2021.

\bibitem{pass2018thunderella}
Rafael Pass and Elaine Shi.
\newblock Thunderella: Blockchains with optimistic instant confirmation.
\newblock In {\em Advances in Cryptology--EUROCRYPT 2018: 37th Annual
  International Conference on the Theory and Applications of Cryptographic
  Techniques, Tel Aviv, Israel, April 29-May 3, 2018 Proceedings, Part II 37},
  pages 3--33. Springer, 2018.

\bibitem{ranchalpedrosa2026hermes}
Alejandro Ranchal-Pedrosa, Dakai Kang, Neil Giridharan, Mohammad Sadoghi,
  Dahlia Malkhi, and Ben Marsh.
\newblock Hermes: Low tail-latency via prefix consensus.
\newblock {\em arXiv preprint arXiv:2607.25916}, 2026.

\bibitem{scaffino2026gatling}
Giulia Scaffino, Max Resnick, and Joachim Neu.
\newblock Gatling: Rapid-fire consensus from parallel composition.
\newblock {\em arXiv preprint arXiv:2606.18220}, 2026.

\bibitem{schneider1990implementing}
Fred~B Schneider.
\newblock Implementing fault-tolerant services using the state machine
  approach: A tutorial.
\newblock {\em ACM Computing Surveys (CSUR)}, 22(4):299--319, 1990.

\bibitem{shoup2000practical}
Victor Shoup.
\newblock Practical threshold signatures.
\newblock In {\em International conference on the theory and applications of
  cryptographic techniques}, pages 207--220. Springer, 2000.

\bibitem{shoup2023sing}
Victor Shoup.
\newblock Sing a song of simplex.
\newblock {\em Cryptology ePrint Archive}, 2023.

\bibitem{shoup2025kudzu}
Victor Shoup, Jakub Sliwinski, and Yann Vonlanthen.
\newblock Kudzu: Fast and simple high-throughput bft.
\newblock {\em arXiv preprint arXiv:2505.08771}, 2025.

\bibitem{shresthahydrangeaplusplus}
Nibesh Shrestha and Aniket Kate.
\newblock Hydrangea++: Enhancing hydrangea with optimistic proposals.
\newblock White paper,
  \url{https://supra.com/documents/hydrangea-plus-plus.pdf}, 2025.

\bibitem{sliwinski2019abc}
Jakub Sliwinski and Roger Wattenhofer.
\newblock {ABC}: Proof-of-stake without consensus.
\newblock {\em arXiv preprint arXiv:1909.10926}, 2019.

\bibitem{shrestha2025hydrangea}Nibesh Shrestha, Aniket Kate, and Kartik Nayak.
\newblock Hydrangea: Optimistic two-round partial synchrony with improved fault
  resilience.
\newblock {\em Cryptology ePrint Archive}, 2025.

\bibitem{shrestha2025sailfish}
Nibesh Shrestha, Rohan Shrothrium, Aniket Kate, and Kartik Nayak.
\newblock Sailfish: Towards improving the latency of dag-based bft.
\newblock In {\em 2025 IEEE Symposium on Security and Privacy (SP)}, pages
  1928--1946. IEEE, 2025.

\bibitem{song2008bosco}
Yee~Jiun Song and Robbert Van~Renesse.
\newblock Bosco: One-step byzantine asynchronous consensus.
\newblock In {\em International Symposium on Distributed Computing}, pages
  438--450. Springer, 2008.

\bibitem{spiegelman2024shoal}
Alexander Spiegelman, Balaji Arun, Rati Gelashvili, and Zekun Li.
\newblock Shoal: Improving dag-bft latency and robustness.
\newblock In {\em International Conference on Financial Cryptography and Data
  Security}, pages 92--109. Springer, 2024.

\bibitem{spiegelman2023quorumstore}
Alexander Spiegelman and Brian Cho.
\newblock Quorum {Store}: How consensus horizontally scales on the {Aptos}
  blockchain.
\newblock Aptos Labs blog, \url{https://medium.com/aptoslabs}, 2023.

\bibitem{spiegelman2022bullshark}
Alexander Spiegelman, Neil Giridharan, Alberto Sonnino, and Lefteris
  Kokoris-Kogias.
\newblock Bullshark: Dag bft protocols made practical.
\newblock In {\em Proceedings of the 2022 ACM SIGSAC Conference on Computer and
  Communications Security}, pages 2705--2718, 2022.

\bibitem{stathakopoulou2019mir}
Chrysoula Stathakopoulou, Tudor David, and Marko Vukoli{\'c}.
\newblock Mir-{BFT}: High-throughput {BFT} for blockchains.
\newblock {\em arXiv preprint arXiv:1906.05552}, 2019.

\bibitem{sui2022marlin}
Xiao Sui, Sisi Duan, and Haibin Zhang.
\newblock Marlin: Two-phase {BFT} with linearity.
\newblock In {\em 2022 52nd Annual IEEE/IFIP International Conference on
  Dependable Systems and Networks (DSN)}. IEEE, 2022.

\bibitem{tonkikh2025raptr}
Andrei Tonkikh, Balaji Arun, Zhuolun Xiang, Zekun Li, and Alexander Spiegelman.
\newblock Raptr: Prefix consensus for robust high-performance bft.
\newblock {\em arXiv preprint arXiv:2504.18649}, 2025.

\bibitem{vandervos2025odontoceti}
Preston Vander~Vos.
\newblock Odontoceti: Ultra-fast {DAG} consensus with two round commitment.
\newblock {\em arXiv preprint arXiv:2510.01216}, 2025.

\bibitem{vonlanthen2024banyan}
Yann Vonlanthen, Jakub Sliwinski, Massimo Albarello, and Roger Wattenhofer.
\newblock Banyan: Fast rotating leader bft.
\newblock In {\em Proceedings of the 25th International Middleware Conference},
  pages 494--507, 2024.

\bibitem{vos2025bluebottle}
Preston~Vander Vos, Alberto Sonnino, Giorgos Tsimos, Philipp Jovanovic, and
  Lefteris Kokoris-Kogias.
\newblock Bluebottle: Fast and robust blockchains through subsystem
  specialization.
\newblock {\em arXiv preprint arXiv:2511.15361}, 2025.

\bibitem{xiang2021strengthened}
Zhuolun Xiang, Dahlia Malkhi, Kartik Nayak, and Ling Ren.
\newblock Strengthened fault tolerance in {Byzantine} fault tolerant
  replication.
\newblock In {\em 2021 IEEE 41st International Conference on Distributed
  Computing Systems (ICDCS)}. IEEE, 2021.

\bibitem{xiang2026prefix}
Zhuolun Xiang, Andrei Tonkikh, and Alexander Spiegelman.
\newblock Prefix consensus for censorship resistant {BFT}.
\newblock {\em arXiv preprint arXiv:2602.02892}, 2026.

\bibitem{xie2025fides}
Shaokang Xie, Dakai Kang, Hanzheng Lyu, Jianyu Niu, and Mohammad Sadoghi.
\newblock Fides: Scalable censorship-resistant {DAG} consensus via trusted
  components.
\newblock {\em arXiv preprint arXiv:2501.01062}, 2025.

\bibitem{yang2022dispersedledger}
Lei Yang, Seo~Jin Park, Mohammad Alizadeh, Sreeram Kannan, and David Tse.
\newblock {DispersedLedger}: High-throughput {Byzantine} consensus on variable
  bandwidth networks.
\newblock In {\em 19th USENIX Symposium on Networked Systems Design and
  Implementation (NSDI 22)}, 2022.

\bibitem{yin2019hotstuff}
Maofan Yin, Dahlia Malkhi, Michael~K Reiter, Guy~Golan Gueta, and Ittai
  Abraham.
\newblock Hotstuff: Bft consensus with linearity and responsiveness.
\newblock In {\em Proceedings of the 2019 ACM Symposium on Principles of
  Distributed Computing}, pages 347--356, 2019.

\end{thebibliography}
\end{document}